\documentclass[preprint]{elsarticle}

\usepackage{amssymb}
\usepackage{amsbsy,amsthm}
\usepackage{amsmath}
\usepackage{xcolor}
\usepackage{graphicx}
\usepackage{pdflscape}
\usepackage{afterpage}
\usepackage{capt-of}
\usepackage{fullpage}
\usepackage{mathtools}

\usepackage[normal,tight,center]{subfigure}
\usepackage{array}
\usepackage{calc}
\usepackage[thinlines]{easytable}
\usepackage{tabu}

\newcommand{\colk}[1]{\textcolor{black}{#1}}

\newcommand{\req}[1] {Eq.~(\ref{eq:#1})}

\newcommand{\rtab}[1] {Table~\ref{tab:#1}}
\newcommand{\rtabs}[2] {Tables~\ref{tab:#1} and \ref{tab:#2}}
\newcommand{\rsec}[1] {section~\ref{sec:#1}}
\newcommand{\rsecs}[2] {sections~\ref{sec:#1} and \ref{sec:#2}}
\newcommand{\rfig}[1] {Fig.~\ref{fig:#1}}
\newcommand{\eqn}[1]{Eq.~(\ref{eqn:#1})}
\newcommand{\fig}[1]{fig.~\ref{fig:#1}}
\newcommand{\Fig}[1]{Fig.~\ref{fig:#1}}

\newcommand{\eqnlabel}[1]{\label{eqn:#1}}
\newcommand{\figlabel}[1]{\label{fig:#1}}

\newcommand{\Ltwo}[1]{\left\|#1\right\|_{\mathcal{L}_2}}
\newcommand{\shift}[1]{\vo{\Phi}_{#1}}

\newcommand{\Ad}{\vo{A}_d}
\newcommand{\Adp}{\Ad^{\vo{\Phi}}}

\newcommand{\real}{\mathbb{R}}

\newcommand{\ad}{\vo{a}_d}

\newcommand{\vo}[1]{\mathbf{#1}}
\newcommand{\bdel}{\boldsymbol{\delta}}
\newtheorem{lem}{Lemma}

\def\dx {\Delta x}
\def\dt {\Delta t}
\def\ns {M}
\def\kdx {\eta}
\def\ksdx {\tilde{\eta}}
\def\II  {{\cal I}(y,x)}

\def\mw {\ksdx}

\journal{Journal of Computational Physics}

\begin{document}

\begin{frontmatter}



\title{A Unified Approach for Deriving Optimal Finite Differences}


\author[label1]{Komal Kumari}
\address[label1]{Department of Aerospace Engineering, Texas A\&M University, College Station, TX 77843, United States}
\author[label1]{Raktim Bhattacharya\corref{cor1}}
\ead{raktim@tamu.edu}
\author[label1]{Diego A.\ Donzis\corref{cor1}}
\ead{donzis@tamu.edu}
\cortext[cor1]{Corresponding author.}

\begin{abstract}
A unified approach to derive optimal finite differences is presented
which combines three critical elements for numerical performance
especially for multi-scale physical problems, namely,  
order of accuracy, spectral resolution and stability.
The resulting mathematical framework reduces to a minimization
problem subjected to equality and inequality constraints. 
\colk{
We show that the framework
can provide 
analytical results for optimal schemes and their 
numerical performance including, for example, the type of 
errors that appear for spectrally optimal schemes.
By coupling the problem in this unified framework, one
can effectively decouple the requirements for order of
accuracy and spectral resolution, for example.
Alternatively, we show how the framework exposes the tradeoffs 
between e.g.\ accuracy and stability and how this can be 
used to construct explicit schemes that remain stable 
with very large time steps. We also show how spectrally 
optimal schemes only bias odd-order derivatives to remain stable,
at the expense of accuracy,
}
while leaving even-order derivatives with symmetric coefficients.
\colk{
Schemes constructed within this framework are tested for 
diverse model problems with an emphasis on reproducing the physics
accurately. 
}
\end{abstract}


\end{frontmatter}

%\begin{comment}
% \linenumbers

\section{Introduction}

Ordinary differential equations (ODEs) as well
as partial differential equations (PDEs) are pervasive 
in science and engineering as they model accurately 
a large number of natural and man-made systems.
Unfortunately, 
many of these equations are exceedingly complex 
at realistic conditions and analytical solutions 
are virtually impossible.  Thus, significant
advances in understanding these phenomena have
relied on the use of computer simulations for which
an appropriate numerical method needs to be used
to assure certain degree of accuracy in the solution.

Perhaps the most widely used method to discretize
these governing equations in order to solve them on a 
computer is the so-called finite differences. 
The derivation of explicit finite difference schemes 
is in general very well known and has been studied 
extensively \cite{hirsch.I}. The general idea is to linearly 
combine the values of the function to be differentiated 
at neighboring points. The weights in this linear
combination are determined so as to minimize error in 
some sense. The specific choice of this objective function,
we show here, has a critical effect on the resulting schemes.
%In a traditional derivation, this is 
%based on maximizing the order of the leading order
%term in the truncation error.

More formally, the standard procedure starts with an
approximation of the derivative of a function $f$ at a point $x_i$ of the form
\begin{equation}
\left. f'_i = {\partial f\over \partial x}\right|_{x_i} \approx 
{1\over\dx}\sum_{m=-\ns}^{\ns}{a_m f_{i+m}} + {\cal O}(\dx^{p+1}) .
\label{eq:du}
\end{equation}
The last term indicates that the truncation error of the approximation
is of order $p+1$.
In a traditional derivation 
one first selects the stencil size, that is 
the number of neighboring points to use in the approximation 
which is $(2\ns+1)$ in \req{du}, 
and then finds the coefficients $a_m$ such that
$p+1$, the order of the truncation error, is largest. This is 
done by eliminating all terms of order lower than $p+1$ 
in a Taylor expansion of the right-hand-side of \req{du}.

However, in many problems, especially those which involve multi-scale
phenomena, one is concerned not only with a global measure of the 
error but also with the error incurred at 
different scales or, in Fourier space, at different wavenumbers.
\colk{The range of wavenumbers and, thus, the resolution 
requirements depend on non-dimensional parameters such as Reynolds number
for turbulent flows, or the 
Prandtl and Schmidt numbers for mixing problems, 
as they define length and time scales 
of interest in the problem. It is therefore of consequence to use schemes that 
are optimal for the time and length scale posed by the physics of the 
problems.} 
This concern is typically addressed by analyzing the spectral characteristics
of scheme derived above by using, for example,  
von Neumann analysis \cite{hirsch.I} or a modified wavenumber approach
\cite{lele1992}. The result of this analysis is used to determine 
if the scheme derived from an order-of-accuracy consideration
is indeed appropriate to resolve all relevant scales in 
the problem at hand.
Note that this is in general {\it a posteriori} evaluation of the scheme. 

There has been efforts in the literature to devise 
schemes with general properties in terms of
spectral accuracy. Early work on acoustic computations
\cite{TW1992} showed that some desired spectral behavior can be obtained 
by solving an 
optimization problem where the objective function to 
be minimized is some measure of these errors. 
\colk{A number of different applications based on the same
general approach have been presented in the literature
\cite{TW1992,KL1996,WANG2001,Zhuang2002,PONZ2003,BB2004,MTWW2006,TKS2006,Fang2013,ZY2013,PIROZ2007}.
In all these studies specific requirements were put forth
typically based on the physics of interest 
which often resulted in subjective criteria to account for
those specific requirements. Some were limited in scope,
for exmaple, by limiting 
the results to approximating only the first derivative
\cite{TW1992,KL1996,BB2004,Zhuang2002,TKS2006}.}
These optimizations were also limited in the sense that 
the nature of the unknown coefficients (e.g.\ whether they 
are symmetric or antisymmetric) 
is specified a priori 
which resulted in the objective function comprising either only the real part 
or only the imaginary part of the spectral error, depending upon the
order of the derivative 
being computed. However, as we show below, 
that symmetry and antisymmetry 
of the coefficients for even and odd derivatives respectively, 
can be obtained as a {\it consequence}
of the minimization problem without any externally
imposed conditions.
Some others 
\cite{ZY2013} formulated an optimization
based on the maximum norm to minimize spectral errors 
but the algorithm adheres to certain ad-hoc rules  
(such as the nature of 
the coefficients or the relative magnitude of the coefficients in 
a stencil) and presents challenges in finding the global optimum.
%Furthermore, 
%the highest wavenumber accurately resolved results from 
%the specified tolerance and cannot be specified {\it a priori} in 
%the formulation.
Other specialized optimizations have been conducted to obtain
schemes with lower errors that, for example, can resolve very strong gradients 
in fluid flow calculations \cite{MTWW2006,WANG2001,PONZ2003},
though some degree of subjectivity and trial-and-error was 
used in the formulations.
Both \cite{WANG2001,PONZ2003} utilize a two level 
optimization to achieve better 
resolving efficiency and \cite{PONZ2003} also incorporated a weight 
function to emphasize on the relevant scales. Another approach is 
presented in \cite{Fang2013} where the unknown coefficient 
is selected according to the level of dissipation required 
without carrying out any formal optimization. 
A general conclusion one can arrive at from all these studies, 
is that more sensible choices
than commonly made can lead to numerical schemes that
\colk{can outperform standard finite differences 
\cite{lele1992,TW1992}}.
Here we support this idea and show that 
the precise meaning of this metric has a clear impact 
on the scheme obtained. 
\colk{Another observation is that formulations are typically tailored 
with specific applications and constraints in mind
due to the different requirements dictated by the physics of 
interest.
It is thus not surprising that there seems to be no 
general rigorous mathematical
framework under which these particular cases can be derived.}
Here, in a first step, we provide such a framework along 
with some rigorous results on the nature of the error that result 
from optimization formulations. \colk{As we will show, the 
framework can incorporate the different requirements needed 
for different cases and conditions.}

A third critical aspect when considering numerical schemes is
their stability. Obviously, to be usable, a scheme must be numerically
stable when utilized to solve an ODE or PDE. Again, this
has traditionally been {\it a posteriori} undertaking:
after selecting 
a scheme of a given order, with a desired spectral accuracy, one
would check if the scheme is stable or not and under which conditions.
This has been the case for standard or optimized schemes 
\cite{TW1992,AZ2003}.
\colk{This is also the case for the two-step optimization of 
\cite{HARAS1994} in which an optimal spatial scheme is obtained 
first, followed by an optimization to get a stable time marching scheme.
This also highlights the importance 
of considering the relation 
between the space and time discretization to increase 
the computational efficiency of the spatial operators.} 
\colk{In \cite{PIROZ2007}, 
the spatial and the temporal schemes are decoupled and optimized 
separately in order to achieve maximum resolving efficiency for both the operators.
The optimization of the temporal scheme is also subject to the stability constraint. 
}
Schemes that remain stable for a broader range in the appropriate 
parameter space (time step size, grid spacing, non-dimensional groups 
such as CFL number, etc.) are typically preferred as simulations with 
larger time steps are computationally less expensive. \colk{While the 
time step and grid spacing are conventionally subject to the stability 
and resolution requirements, 
\cite{PIROZ2007} shows that optimal values for these that 
minimize computational cost for some error level can also be obtained.}

An overriding question is, thus, whether it is possible to 
find optimal schemes of given order, that are stable and 
that minimize the spectral error in a suitably defined manner
typically informed by the physical characteristics of the
problem being solved. This is the main motivation of the present work.
The mathematical framework in which this can be achieved 
reduces to an optimization 
problem with equality and inequality constraints which can be solved, 
under certain conditions, analytically.
\colk{The importance of this work is that it allows us to express
physically meangingful desired properties and 
constraints into a unified %and rigorous 
mathematical 
framework which results in highly-accurate schemes for a particular 
problem of interest.} 
Another important aspect of the proposed 
framework is that it also exposes explicitly tradeoffs that can be 
profitably used in specific circumstances. For example, we 
show that it is possible to construct explicit schemes
that can remain stable for very large time steps 
(even an order of magnitude larger than equivalent 
standard schemes) 
when constraints on accuracy at some scales can be relaxed.
We can also show that unlike traditional finite differences, 
one can design ``spectrally flat'' schemes which present a
more homogeneous accuracy distribution across wavenumbers.
\colk{Furthermore, because of optimality in spectral properties, 
we also show that the resulting schemes present better
performance in terms of important physical properties like its
dispersion relation, and group and phase velocity \cite{TKS2003,
TKS2006, TKS2014}.}

The rest of the paper is organized as follows. In \rsec{framework}
we present the fundamental ideas behind the framework
which include order of accuracy, spectral accuracy 
and stability. In \rsec{simulations} we present numerical
results using optimal schemes constructed using the tools presented 
here. We present results on accuracy, stability and tradeoffs that 
can be explicitly incorporated into the design of finite differences.
Final remarks and conclusions are discussed in \rsec{conclusions}.
Proof of two important theorems are included as appendices.

%%%%%%%%%%%%%%%%%%%%%%%%%%%%%%%%%%%%%%
\section{The framework for deriving finite differences} 
\label{sec:framework}

As described above, in this work we construct a
framework to derive 
finite differences in which three important 
design characteristics, namely 
order of accuracy, spectral accuracy,
and stability, are combined into a 
rigorous mathematical framework.
We now present each in turn.

%\subsubsection{Even derivatives}
%\subsubsection{Odd derivatives}

%\subsection{Stability}

\subsection{Order of accuracy}
A generalization of the approximation in \req{du} to the $d$-th
\colk{spatial}
derivative is given by
\begin{align}
f^{(d)}_i	= \frac{1}{(\dx)^d}\sum_{m=-\ns}^{\ns} a_m f_{i+m},
\eqnlabel{explicit}
\end{align}
where, as before, we have $\ns$ points on either side of the $i$-th 
grid point where the derivative is sought.
The stencil size is then $S\equiv 2\ns+1$.
A Taylor series for a term on the
right-hand-side of \eqn{explicit} 
can be written as $$f_{i+m}=f_i + (m\dx)f_i' + (m\dx)^2f_i''/2!+ \dots.$$
Upon constructing the entire sum in \eqn{explicit}, 
a $(p+1)$-th order approximation of the $d$-th derivative
requires that  
the term with the $d$-th derivative be equal to $d!$ and that the rest
of the terms up to order $p$
be zero. 
After some algebra these constraints can be written as 
%\begin{align}
%	&\sum_{m=-\ns}^{\ns} a_m = 0, \eqnlabel{askew}\\
%&\sum_{m=-\ns}^{\ns} ma_m  = 0,
%\sum_{m=-\ns}^{\ns} m^2a_m  = 0,
%\cdots, \sum_{m=-\ns}^{\ns} m^{d-1}a_m  = 0,\\
%&\sum_{m=-\ns}^{\ns} m^da_m  = d!,\\
%&\sum_{m=-\ns}^{\ns} m^{d+1}a_m  = 0, \cdots, 
%\sum_{m=-\ns}^{\ns} m^{d+p}a_m  = 0 \ ,
%\end{align}

\begin{equation}
\sum_{m=-\ns}^{\ns} m^q a_m  = \begin{cases}
0  & q\ne d, \\
d! & q=d,
\end{cases}
\end{equation}
for $q\le d+p$,
or more compactly as
\begin{align}
\vo{a}_d^T\vo{X}_d &= \vo{y}_d,	\eqnlabel{exp:orderConstr}
\end{align}
where
\begin{align}
\vo{a}_d^T &:= \begin{bmatrix}a_{-\ns} & a_{-\ns+1} & \cdots a_{\ns-1} & a_{\ns}\end{bmatrix},\\
\vo{m}^T &:= \begin{bmatrix} -\ns & -\ns+1 & \cdots & \ns-1 & \ns \end{bmatrix},\\
\vo{X}_d &:= \begin{bmatrix}
\vo{1}_{S\times 1} & \vo{m} & \cdots &
\vo{m}^{d-1} & \vo{m}^d & \cdots &\vo{m}^{d+p}
\end{bmatrix}, \eqnlabel{def_Xd} \\
\vo{y}_d &:= \begin{bmatrix}
\vo{0}_{1\times d} & d! & \vo{0}_{1\times p}
\end{bmatrix}.
\end{align}
with $\vo{1}_{S\times 1}$ is an $S\times 1$ vector with ones as its 
elements and $\vo{0}_{1\times d}$ a $1\times d$ vector with zeros as 
its elements. The vectors $\vo{m}^n$ in \eqn{def_Xd} are 
defined as vectors composed 
of each element of $\vo{m}$ raised to the power $n$.

Equation \eqn{exp:orderConstr} is the linear system that,
for a given stencil size $\ns$, results in a finite difference scheme 
of order $p+1$. 
%\comment{[DD: do we always have a unique solution to \eqn{exp:orderConstr}? ]}
%\raktim{The linear system in  \eqn{exp:orderConstr} has a solution if and only if $$\vo{y}_d(\vo{I}-\vo{X}_d^+\vo{X}_d) = 0,$$ where $^+$ denotes pseudo-inverse. This means that $\vo{y}_d$ should be in the null-space of $(\vo{I}-\vo{X}_d^+\vo{X}_d)$. We next investigate conditions when this is true. TBD.}
This approach, though presented in different forms across the
literature, forms the basis for standard 
derivation of finite differences when $S=d+p$.
In this case, the number of unknown coefficients in 
$\vo{a}_d$ equals the number of terms that need to be 
eliminated 
to maintain a certain order of accuracy.
If on the other hand, $S>d+p$ then the solution to \eqn{exp:orderConstr}
is not unique. The remaining degrees of freedom 
can then be used to, e.g., assure spectral accuracy.
This is presented next.

\subsection{Spectral accuracy}

In order to understand the behavior of discrete differentiation 
operators at different scales or frequencies, 
it is common to evaluate them utilizing a 
Fourier representation.
For simplicity consider a single mode in a 
\colk{spatial} 
discrete Fourier series:
\begin{equation}
f(x) = \hat{f}e^{jkx} 
\end{equation}
where $\hat{f}$ is the Fourier coefficient of the mode
at wavenumber $k$, and $j=\sqrt{-1}$.
Clearly, the exact $d$-th derivative is given by
\begin{equation}
f^{(d)}(x) = (jk)^d f(x). \eqnlabel{spectral}
\end{equation}

Now consider a discrete approximation of the derivative of the 
form \eqn{explicit}. 
Since,
\begin{align}
f_i &:= f(x_i) = \hat{f}e^{jkx_i} \eqnlabel{singleFmode},\\
f_{i+m} & := f(x_i + m\Delta x) = \hat{f}e^{jkx_i}e^{jkm\Delta x} =
f_ie^{jkm\Delta x}
\end{align}
equation \eqn{explicit} becomes 
\begin{align}
f_i^{(d)} &= \left(\frac{1}{(\dx)^d} \sum_m a_m e^{jkm\Delta x}\right) f_i \eqnlabel{explicit:spectral}
\end{align}
in terms of this single Fourier mode.

%Substituting $\kdx:=k\Delta x$, $\kdx\in[0,\pi]$, we can write the error in the modified wave number as
Comparison between the numerical approximation \eqn{explicit:spectral}
and exact differentiation \eqn{spectral} suggests a natural way to 
define the error at a given waveumber as
\begin{equation}
	e(\kdx) = \sum_m a_m e^{jm\kdx} - (j\kdx)^d,
	\eqnlabel{e_w}
\end{equation}
where $\kdx:=k\Delta x$
is a convenient normalized wavenumber in the interval 
$[0,\pi]$.
More compactly, this can be expressed in matrix form as 
\begin{equation}
	e(\kdx) = \left(\vo{C}^T(\kdx) + j\vo{S}^T(\kdx)\right)\vo{a}_d - (j\kdx)^d,
\end{equation}
where
\begin{align}
\vo{C}(\kdx) := \begin{bmatrix} \cos(-\ns \kdx)\\\vdots \\ \cos(-\kdx) \\ 1 \\ \cos{\kdx} \\ \vdots \\\cos(\ns \kdx) \end{bmatrix}, \text{ and } \vo{S}(\kdx) := \begin{bmatrix} \sin(-\ns \kdx)\\\vdots \\ \sin(-\kdx) \\ 0 \\ \sin(\kdx) \\ \vdots \\\sin(\ns \kdx) \end{bmatrix}. \eqnlabel{SinCos:def}
\end{align}

Equation \eqn{e_w} can also be written 
in terms of the so-called modified wavenumber 
$(j\ksdx)^d=\sum_m a_m e^{jm\kdx} = \left(\vo{C}^T(\kdx) +
j\vo{S}^T(\kdx)\right)\ad$ as 
\begin{equation}
e(\kdx)=(j\ksdx)^d-(j\kdx)^d = j^d(\ksdx^d-\kdx^d) .
\eqnlabel{e_modw}
\end{equation}
Clearly, the difference between $\ksdx^d$ and $\kdx^d$
provides a measure of the spectral error at wavenumber 
$\kdx$. The ratio of the modified wavenumber to the actual
wavenumber is then
\begin{equation}
G_d := {\ksdx^d\over \kdx^d} = {e(\kdx)\over (j \kdx)^d} + 1 ,
\eqnlabel{modw}
\end{equation}
which is a complementary measure of error across wavenumbers and 
will be used later on when comparing different schemes.

A global figure of merit to assess how accurately the scheme captures 
spectral content can be defined as the weighted 
$\mathcal{L}_2$ norm of the error $e(\kdx)$:
\begin{equation}
\|e(\kdx)\|^2_{\mathcal{L}_2} := \int_0^\pi \gamma(\kdx)e^\ast(\kdx)e(\kdx)d\kdx =: \left\langle e^\ast(\kdx)e(\kdx) \right\rangle .
\eqnlabel{sp_error}
\end{equation}
Here $e^\ast(\kdx)$ is the complex conjugate of $e(\kdx)$, and $\gamma(\kdx)$ is a
weighting function introduced to provide control 
over which wavenumbers are to be more accurately resolved. 
The selection of $\gamma(\kdx)$
would depend, in general, on the physical characteristics 
of the system of interest. For example, for PDEs with multi-scale broadband
solutions a natural choice would be a constant $\gamma(\kdx)$
over the range of $\kdx$ of interest and zero elsewhere. 
For a system with two well defined wavenumber bands, on the other hand, 
one can define $\gamma(\kdx)$ presenting relatively large values 
around those bands but negligible values everywhere else.
Examples on the impact of this choice 
will be provided in \rsec{simulations}.

%We thus seek numerical schemes that minimize the 
%appropriately defined spectral error \eqn{sp_error}.
We can now frame the problem of the derivation of 
finite difference schemes combining 
spectral resolution and order of accuracy.
Formally, 
our goal is to determine $\vo{a}_d$ such that
$\|e(\kdx)\|^2_{\mathcal{L}_2}$ is minimized, subject to 
a given order of
accuracy defined by \eqn{exp:orderConstr}, i.e.
\begin{align}
\min_{\vo{a}_d\in\real^{2\ns+1}} \|e(\kdx)\|^2_{\mathcal{L}_2}, \text{ subject to \eqn{exp:orderConstr}.}
\eqnlabel{order_spect}
\end{align}

Equation \eqn{order_spect} provides, then, the unifying formalism 
to find  the 
coefficients in \eqn{explicit} that both provides a given order
of accuracy and minimizes error in spectral space \cite{TW1992}.
This formulation, written in different ways,
has been used extensively as pointed out
in the introduction, but here is presented in a very general form.
Note that if the number of unknowns ($S=2\ns +1$) is 
equal to the number of terms to be removed from the truncation 
error to achieve a given order, then \eqn{exp:orderConstr}
has a unique solution and no optimization is possible.
This is the case of standard finite difference schemes.
If, on the other hand, the stencil
size makes the number of unknowns greater that those needed to 
achieve a given order, 
the system will utilize those degrees of freedom
to minimize $\|e(\kdx)\|_{\mathcal{L}_2}^2$. %\eqn{order_spect}.

We note here that in standard derivations of finite differences
based only on order-of-accuracy considerations, the spectral 
behavior is found a posteriori. Spectral resolution is then
coupled to (dependent on) order of accuracy.
While typically increasing the formal 
order of accuracy of the scheme, leads to a better spectral 
resolution, standard techniques provide no mechanism to
constrain the spectral behavior of the resulting schemes.
In the approach presented above, on the other hand,
specifications on spectral accuracy are independent 
of order of accuracy.
Thus, we see that
by {\it coupling the two mathematical systems} 
into a unified formulation, 
we effectively 
{\it decouple the requirements} for order of accuracy and 
spectral resolution. 
%In standard derivations, one
%uses the information from neighboring points to obtain
%the highest possible order of accuracy, and then evaluates
%spectral resolution a posteriori. Here, 
The result is that when the stencil size is increased, 
the extra degrees of freedom can be used 
to either increase the order of accuracy further or
improve specific spectral behavior. 
%While the formulation 
%presented here is general, 
In \rsec{simulations} 
we present specific examples of optimized 
low-order schemes that are shown to 
have better spectral resolution 
than standard high-order finite differences.

To develop the theory further, it is convenient 
to distinguish between odd-order and even-order derivatives which 
leads to different behavior in Fourier space. 
In both cases an analytical solution can be found and 
is presented next.
 
 \subsubsection{Even derivatives}\label{sec:even}
 For $d=2q$, $q=\{1,2,\cdots\}$, $e(\kdx)$ becomes
 \begin{align}
 	\nonumber e(\kdx) &= \left(\vo{C}^T(\kdx)\vo{a}_d-(-1)^q\kdx^{d}\right) + j\vo{S}(\kdx)^T\vo{a}_d.
 \end{align}  
 
The $\mathcal{L}_2$ norm of the error is therefore,
\begin{align}
\|e(\kdx)\|^2_{\mathcal{L}_2} 
  &:= \int_0^\pi \gamma(\kdx)
  \left[\left(\vo{C}^T(\kdx)\vo{a}_d-(-1)^q\kdx^{d}\right)^2 +
  \left(\vo{S}(\kdx)^T\vo{a}_d\right)^2\right] d\kdx
    \eqnlabel{error_SC},\\
& = \vo{a}_d^T\underbrace{\left(\left\langle  \gamma(\kdx)\vo{C}(\kdx)\vo{C}^T(\kdx) \right\rangle + \left\langle \gamma(\kdx)\vo{S}(\kdx)\vo{S}^T(\kdx)\right\rangle\right)}_{\vo{Q}_d}\vo{a}_d - 2\vo{a}_d^T\underbrace{\left\langle \gamma(\kdx) (-1)^q\kdx^d\vo{C}(\kdx)\right\rangle}_{\vo{r}_d} + \left\langle \gamma(\kdx) \kdx^{2d} \right\rangle ,
\end{align}
and the optimization problem \eqn{order_spect} can be written as 
\begin{equation}
\min_{\vo{a}_d\in\real^{2\ns+1}} 
(\vo{a}_d^T\vo{Q}_d\vo{a}_d - 2\vo{a}_d^T\vo{r}_d),
\text{ subject to \eqn{exp:orderConstr}},
\eqnlabel{cost:a}
\end{equation}
 where the constant term $\left\langle \kdx^{2d} \right\rangle$ is ignored in the
 minimization. The system \eqn{cost:a} is a quadratic programming problem,
 with linear equality constraints. We can solve this problem analytically,
 which is determined from the 
 Karush-Kuhn-Tucker (KKT) condition
 \cite{bertsekas1999nonlinear}. The optimal solution for the 
 coefficients $\vo{a}_d$, which will be denoted with an asterisk 
 ($\vo{a}_d^\ast$), satisfies the following KKT condition
 \begin{align}
 \underbrace{\begin{bmatrix} \vo{Q}_d & \vo{X}_d \\
 \vo{X}_d^T & \vo{0}_{(d+p+1)\times(d+p+1)}\end{bmatrix}}_{:=\vo{K}} \begin{pmatrix}\vo{a}_d^\ast\\ \boldsymbol{\lambda}_d^\ast\end{pmatrix} = 
 \begin{pmatrix} \vo{r}_d \\ \vo{y}_d^T\end{pmatrix}, \eqnlabel{KKT1}
 \end{align}
 where $\boldsymbol{\lambda}_d\in\real^{d+p+1}$ is the Lagrange multiplier associated with the constraint in \eqn{exp:orderConstr}. %Let $\vo{Z}\in\real^{N\times (N-d-p-1)}$ be the matrix whose columns span $\vo{Ker}(\vo{X}_d^T)$, i.e. $\vo{X}_d^T\vo{Z}=0$.
Since $\vo{X}_d^T\in\real^{(d+p+1)\times S}$ has full row rank, then
the KKT matrix  $\vo{K}$ is nonsingular and \eqn{KKT1} has unique
solution $(\vo{a}_d^\ast,\boldsymbol{\lambda}^\ast_d)$ given by
 \begin{align}
  \begin{pmatrix}\vo{a}_d^\ast\\ \boldsymbol{\lambda}_d^\ast\end{pmatrix} = \begin{bmatrix} \vo{Q}_d & \vo{X}_d \\
 \vo{X}_d^T & \vo{0}_{(d+p+1)\times(d+p+1)}\end{bmatrix}^{-1} \begin{pmatrix} \vo{r}_d \\ \vo{y}_d^T\end{pmatrix}.\eqnlabel{optimal:Explicit}
 \end{align}

Since everything on the right-hand-side is known, 
\eqn{optimal:Explicit} provides 
the coefficients $\vo{a}_d^\ast$ 
for the finite difference approximation 
\eqn{explicit} of order
 $p+1$ which minimize the spectral error \eqn{sp_error}.
 
As an example consider a second-order scheme for a
second derivative. In this case, a 3-point stencil 
($S=2\ns+1=3$) yields a unique solution to \eqn{exp:orderConstr},
namely the common approximation 
$\vo{a}_2^\ast=\begin{bmatrix}1 & -2 & 1 \end{bmatrix}^T$ 
or
$f'_i \approx (f_{i-1}-2f_i+f_{i+1})/\dx^2$.
If one now retains a second-order approximation but 
increases the stencil size $S$, the additional degrees of 
freedom are utilized to reduce, through the minimization
process, the spectral error in the approximation.
For illustration purposes assume we would like to 
resolve as accurately as possible all wavenumbers in the range 
$\kdx\in[0,2.5]$. In this case, one can naturally choose  
$\gamma(\kdx)=1$ for $\kdx\in[0,2.5]$ and $\gamma(\kdx)=0$ otherwise.
Upon solving \eqn{optimal:Explicit},
we obtain the coefficients of the resulting schemes which are
shown in \rfig{second_optim_ad} and \rtab{second_optim_ad}. 

\begin{figure}[h!]
\includegraphics[trim={3cm 0 3cm 0},clip,width=\textwidth]{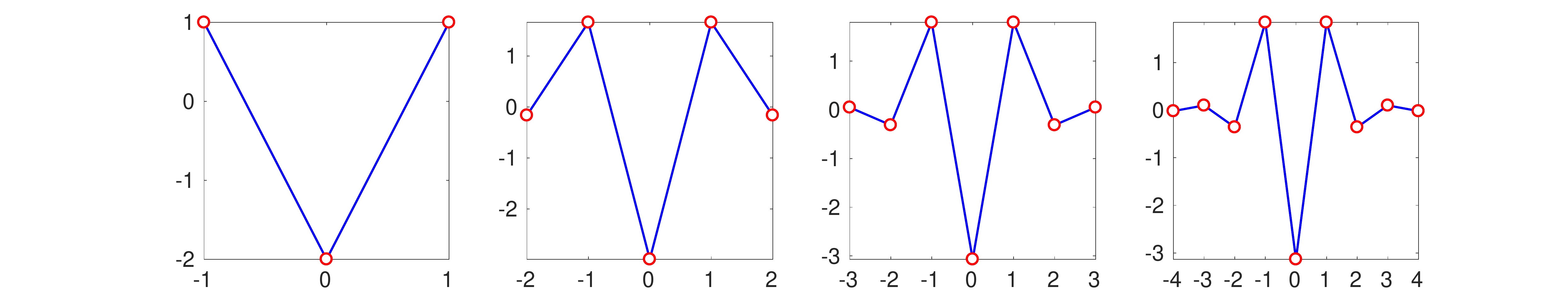}	
\begin{picture}(0,0)
        \put(10,60){\rotatebox{0} {$\vo{a}_2^*$}}
        \put(55,0){$-\ns:\ns$}
	\put(165,0){$-\ns:\ns$}
	\put(275,0){$-\ns:\ns$}
	\put(385,0){$-\ns:\ns$}
        \end{picture}
\caption{Optimal coefficients for various stencil sizes 
$S=2\ns+1$, for a second derivative with second order accuracy. 
The coefficients were
obtained using $\gamma(\kdx)=1$, for $\kdx\in[0,2.5]$ and 
$\gamma(\kdx) = 0$ otherwise.
}
\label{fig:second_optim_ad}
\end{figure}
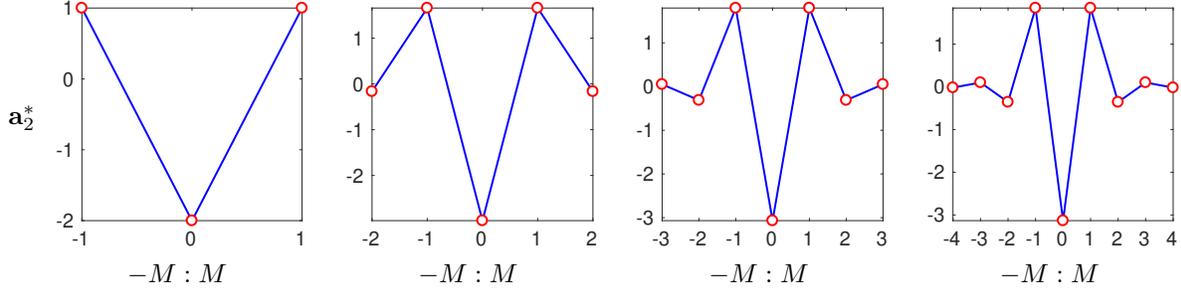

\begin{table}[h!]
\begin{center}
%{\def\arraystretch{1}
\setlength{\tabcolsep}{15pt}
\begin{tabular}{|c|c|c|c|c|}
\hline
     & $\ns=1$ &$\ns=2$ & $\ns=3$& $\ns=4$\\[0.5ex]
\hline 
       $a_0^*$ & -2 & -2.986945912146335 & -3.067324780469417 & -3.132525936497260 \\[0.5ex]
       $a_1^*$ &  1 & 1.657963941430890 & 1.795865984254199 & 1.843958787844204 \\[0.5ex]
       $a_2^*$ &    & -0.164490985357722 & -0.312793272384242 & -0.357929955982910 \\[0.5ex]
       $a_3^*$ &    & & 0.050589678364752 & 0.099426449444277 \\[0.5ex]
       $a_4^*$ &    & & & -0.019192313056941 \\[0.5ex]
\hline
\end{tabular}%}
\end{center}
\caption{Numerical coefficients for the schemes 
in \rfig{second_optim_ad}. Note that since the resulting 
schemes are symmetric around the central grid point,
only coefficients on one side are shown.
}
\label{tab:second_optim_ad}
\end{table}

%\label{fig:second_optim_ad}
%\end{figure}

The spectral accuracy of these schemes is shown in \rfig{second_optim}
where we compare $\ksdx^2$ (magenta) to $\kdx^2$ (dashed black)
as in \eqn{e_modw}: the difference between these two curves
correspond the error representing the derivative
at a given frequency. The standard second order scheme, $\ns=1$, 
is shown with the red line for comparison.
In the top panels of the figure,
we clearly see that as we increase $\ns$ a better
representation of the exact derivative is achieved in
the range $\kdx\in[0,2.5]$ as expected. There 
is also a stark difference between the standard and the optimized 
schemes which becomes more prominent as $\kdx$ or $\ns$ increases. 

\begin{figure}[h!]
\includegraphics[trim={4cm 0 4cm 0},clip,width=\textwidth]{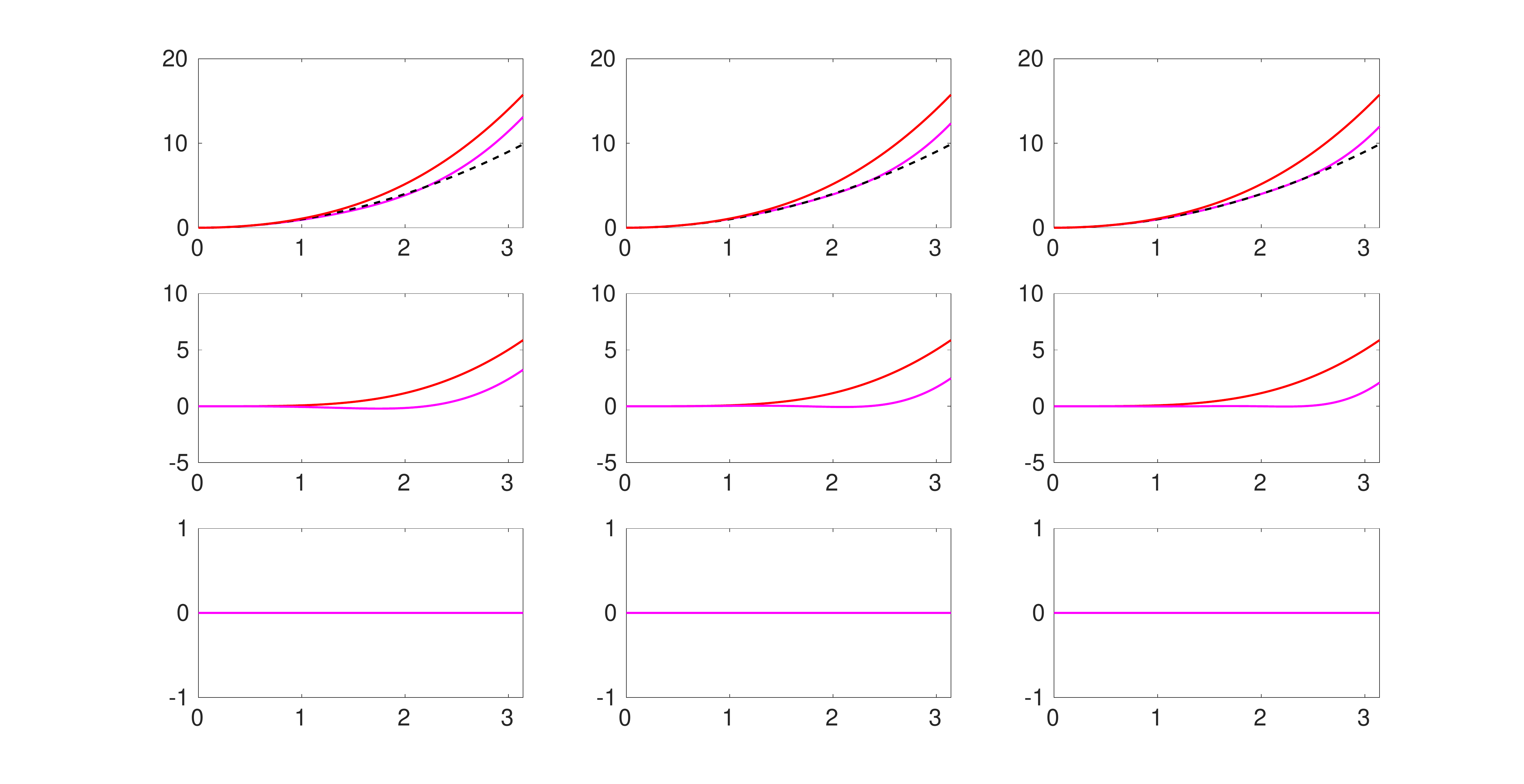}	
	\begin{picture}(0,0)
        \put(70,285){$M=2$}
        \put(230,285){$M=3$}
        \put(385,285){$M=4$}
        \put(3,235){\rotatebox{90} {$\ksdx^2$}}
        \put(3,150){\rotatebox{90} {$\Re[e(\kdx)]$}}
        \put(3,60){\rotatebox{90} {$\Im[e(\kdx)]$}}
        \put(80,25){$\kdx$}
        \put(240,25){$\kdx$}
        \put(395,25){$\kdx$}
        \end{picture}
\caption{Spectral accuracy for the schemes presented 
in \rfig{second_optim_ad} and \rtab{second_optim_ad}.
Top, middle and bottom rows show the modified wavenumber,
the real part, and the imaginary part of the spectral error, 
respectively.
Left, middle and right columns correspond to $M=2$, $3$, 
and $4$, respectively. Dashed black line: 
exact differentiation ($\ksdx^2=\kdx^2$).
Magenta solid line: $\ksdx$ for optimized schemes.
In all plots, red solid line corresponds to 
standard second-order scheme ($M=1$) for 
comparison.
}
\label{fig:second_optim}
\end{figure}
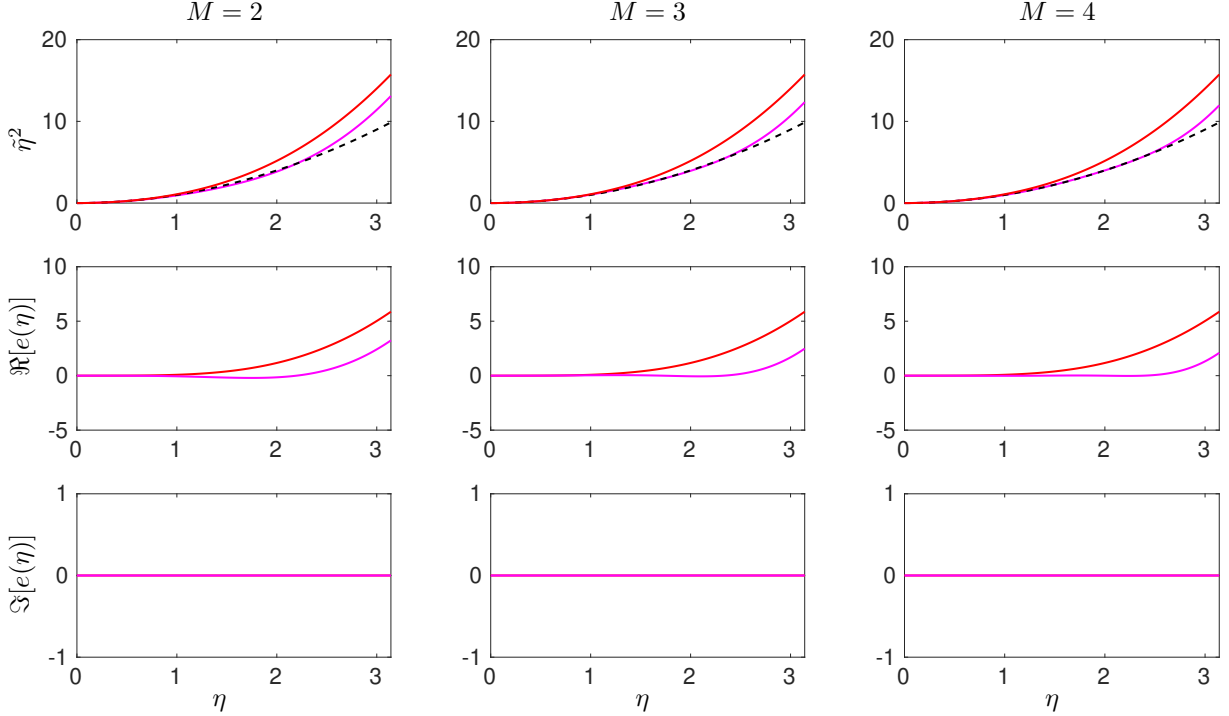

Further information about the nature of the error can 
be obtained by analyzing the real and imaginary components
of $e(\kdx)$. This is shown in the middle and bottom panels
of \rfig{second_optim}. While there are significant errors in the 
real component $\Re[e(\kdx)]$ for
$\kdx\in[0,2.5]$ when $\ns=1$, they decrease
substantially when $\ns$ is increased,
as expected.
The errors outside of the support of $\gamma(\kdx)$ are marginally 
affected since those are not
modes for which the minimization process seeks 
optimal solutions. The imaginary part $\Im[e(\kdx)]$
on the other hand presents zero error.
This fact can indeed be proven and is 
presented here as,
\begin{lem}
\label{lem:imagZero}
Finite difference approximations \eqn{explicit}
obtained from \eqn{order_spect}
with $d$ even, yield a 
spectral error \eqn{e_w} which satisfies $\Im[e(\kdx)] = 0$.
\end{lem}
\begin{proof}
See \ref{sec:app_imag_zero}.
\end{proof}

This result has implications in terms of the type of error expected
(dispersive versus dissipative) 
when the scheme is used in a fully  discretized PDE. 
In particular, from \eqn{modw} we have $G_2=\ksdx^2/\kdx^2 =
(-e(\kdx)/\kdx^2+1)$
which implies that the modified wavenumber will not have a phase
error since $\Im[e(\kdx)]=0$. 
Only dissipative errors are thus expected.

We note here that the error has also been defined using an 
additional arbitrary coefficient 
\cite{MTWW2006, WANG2001} 
%used in between $0$ and $1$ 
to emphasize the relative importance of the 
dissipative and dispersive errors. 
That is, one modifies \eqn{error_SC} to 
$\|e(\kdx)\|^2_{\mathcal{L}_2} 
  = \int_0^\pi  \gamma(\kdx)
  [\sigma(\vo{C}^T(\kdx)\vo{a}_d-(-1)^q\kdx^{d})^2 +
  (1-\sigma)(\vo{S}(\kdx)^T\vo{a}_d)^2] d\kdx $
  where $\sigma\in[0,1]$.
However, from the proof of 
lemma \ref{lem:imagZero}, the minimum value of the error
is obtained when one of the integrals is
identically zero.
Therefore the parameter $\sigma$ will only appear 
as a constant and will have no effect on the optimal coefficients. 
The parameter $\sigma$, thus, is of relevance only 
when the grid is biased in one direction.
In this scenario, neither of the integrals in 
\eqn{error_SC} will be $0$, resulting in error with both real and imaginary
part. 
The optimal solution will then be a minimum of the summation of 
dispersive and dissipative errors.

It is also of interest to understand how errors change with 
$S$ (or $\ns$). This is so
because increasing $N$ (number of 
grid points) or $S$ (stencil size) leads to more computations 
and, thus, more computationally intensive simulations.
For standard schemes, when one increases the number of points in 
the stencil ($S$), 
one also increases the order of accuracy. In particular,
the error for a given $S$ (or $\ns$) is proportional
to $\dx^{2\ns}$ 
where the proportionality constant also decreases 
with $\ns$.
While there is no easy way to represent these constants in
closed form, from their structure one may then still 
expects an approximately exponential decrease in the error 
as the stencil size increases. This is indeed
what we see in \rfig{explicitL2error2} where we show 
the error $\| e(\kdx)\|_{\mathcal{L}_2}^2$ as a function 
of $\ns$ (black symbols).
For optimized schemes of fixed 
order, on the other hand, it is not obvious a priori the rate of convergence 
with the number of grid points used in the stencil.
However, we can readily evaluate the error numerically 
to assess this convergence rate. 
This is shown in \rfig{explicitL2error2} where we include 
the result of such numerical 
calculations with the second-order optimized schemes (blue symbols) in 
\rfig{second_optim_ad} as $M$ is increased.
It can be seen that the error also decreases 
approximately exponentially with $\ns$. 
While for $M=1$ the two curves coincide as expected 
since both approaches lead to the same scheme, optimized
schemes (blue) present a better convergence rate as
the stencil size increases.

\begin{figure}[h!]
\begin{center}
\includegraphics[width=0.5\textwidth]{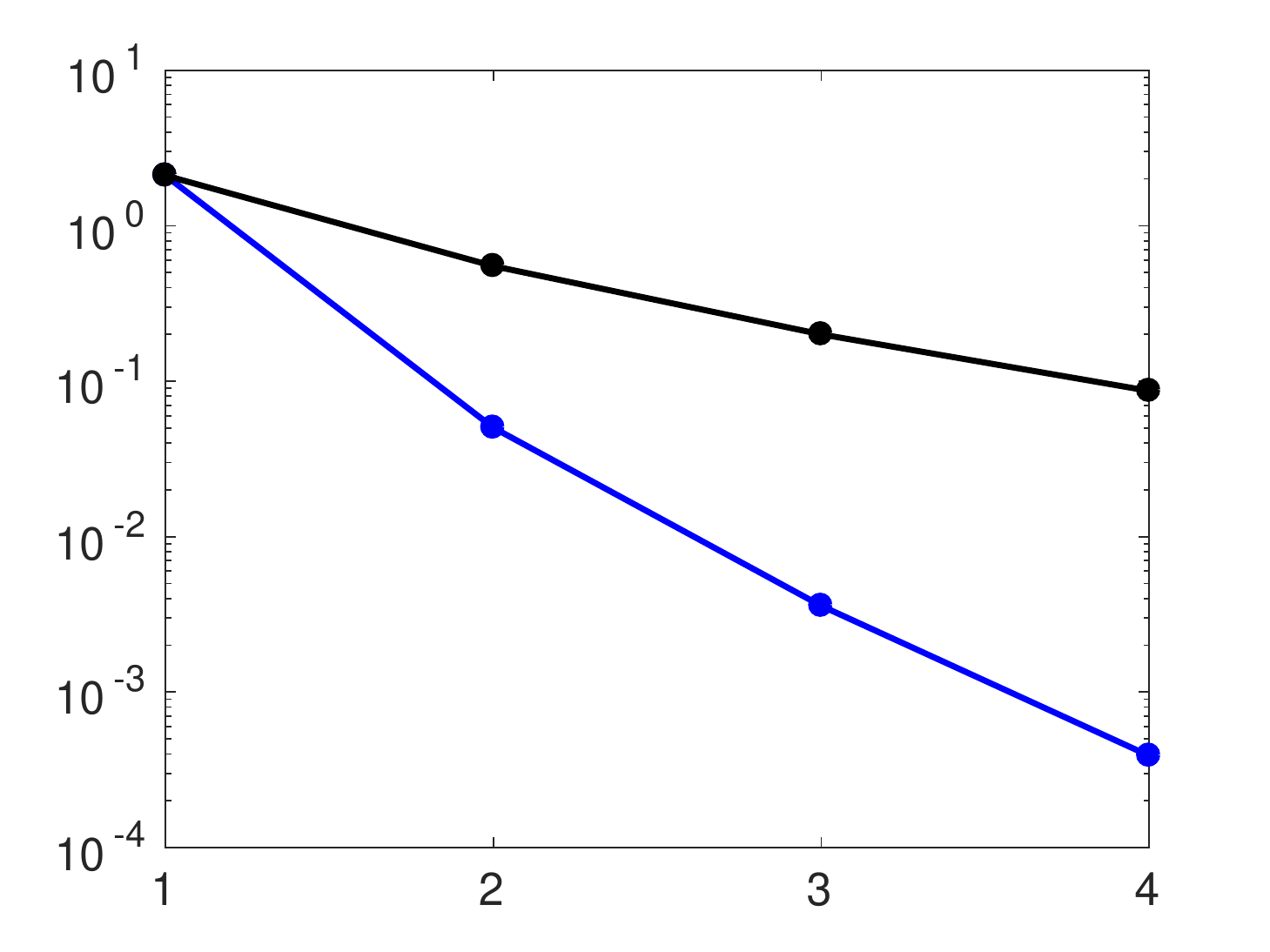}	
\begin{picture}(0,0)
        \put(-245,80){\rotatebox{90} {$\| e(\kdx)\|_{\mathcal{L}_2}^2$}}
        \put(-130,-5){$\ns$}
       \end{picture}
\caption{Optimal spectral error for various stencil size $\ns$,
approximating second derivative. The 
blue line corresponds to the second order 
optimized schemes presented in 
\rfig{second_optim_ad} and \rtab{second_optim_ad}. The 
black line corresponds to the standard schemes with 
stencil size $\ns$ and $2\ns$ order of accuracy.
}
        \label{fig:explicitL2error2}
\end{center}
\end{figure}
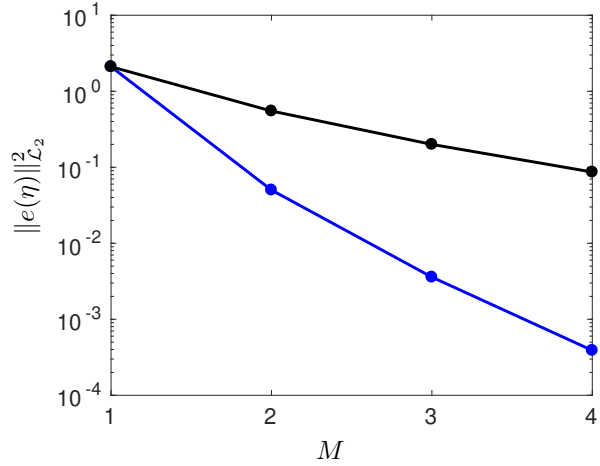

\subsubsection{Odd derivatives}
 For $d=2q+1$, $q=\{0,1,\cdots\}$, $e(\kdx)$ becomes
 \begin{align}
 	\nonumber e(\kdx) &= \vo{C}^T(\kdx)\vo{a}_d + j\left(\vo{S}^T(\kdx)\vo{a}_d -(-1)^q \kdx^{d}\right).
 \end{align}  
 
The $\mathcal{L}_2$-norm of the error for this case is
\begin{align}
\|e(\kdx)\|^2_{\mathcal{L}_2} &:= \int_0^\pi \gamma(\kdx) \left[\left(\vo{C}^T(\kdx)\vo{a}_d\right)^2 + \left(\vo{S}^T(\kdx)\vo{a}_d -(-1)^q \kdx^{d}\right)^2\right]d\kdx,\\
&= \vo{a}_d^T\underbrace{\left(\left\langle \gamma(\kdx) \vo{C}(\kdx)\vo{C}^T(\kdx)\right\rangle + \left\langle  \gamma(\kdx) \vo{S}(\kdx)\vo{S}^T(\kdx)\right\rangle\right)}_{\vo{Q}_d}\vo{a}_d -2\vo{a}_d^T\underbrace{\left\langle \gamma(\kdx) (-1)^q\kdx^d\vo{S}(\kdx)\right\rangle}_{\vo{r}_d} + \left\langle  \gamma(\kdx) \kdx^{2d}\right\rangle. \eqnlabel{rd:odd}
\end{align}
The optimization problem is the same as \eqn{cost:a} with $\vo{r}_d$
given by \eqn{rd:odd}, and the optimal solution is given by
\eqn{optimal:Explicit}.

As an example, we consider here again a second-order approximation but
of the first derivative for increasing values of $\ns$. The resulting
schemes are shown in \rfig{first_optim_ad} and their spectral 
behavior is shown in \rfig{first_optim}. 
Similar conclusions to the second derivative example shown above
are observed. As $\ns$ increases the wavenumbers where $\gamma(\kdx)$
is non-zero are increasingly well resolved which is seen 
as $\ksdx$ becoming closer to $\kdx$ in the top panels of \rfig{first_optim}.
We can also see that the real component of the spectral error (middle panels)
is zero. This is due to the following lemma.
\begin{lem}
Finite difference approximations \eqn{explicit}
obtained from \eqn{order_spect}
with $d$ odd, yield a 
spectral error \eqn{e_w} which satisfies $\Re[e(\kdx)] = 0$.
\label{lem:realZero}
\end{lem}
\begin{proof}
See \ref{sec:app_real_zero}.
\end{proof}

%In the case of odd derivatives, 
%it can be shown  the real part of the error $\Re[e(\kdx)]$ 
%is indeed zero (\ref{sec:app_real_zero})
%which is confirmed by the numerical results in the middle panels 
%or \rfig{first_optim}. 
Since $G_1 = e(\kdx)/(j\kdx)+1$, and 
$e(\kdx)$ has only an imaginary component, then $G_1$ is real. 
As in the example before, errors are then expected to be 
\colk{dispersive} 
in nature though the exact nature of the error would depend on
the PDE in which such a scheme is used. 
As before, we found the global error to 
decrease exponentially with $\ns$ as seen in \rfig{first_optim_L2}
and faster than non-optimized schemes.

%\begin{cor}
%Optimal solution $\vo{a}_d^\ast$ for odd $d$ is skew symmetric about the $(\ns+1)^\text{th}$ element.
%\end{cor}
%\begin{proof}
%Since the components of $\vo{C}(\kdx)$ are symmetric about the $(\ns+1)^\text{th}$ element, $\vo{a}_d^\ast$ must be skew symmetric about the $(\ns+1)^\text{th}$ element so that $\vo{C}(\kdx)^T\vo{a}_d^\ast = 0, \forall\,w$.
%\end{proof}

\begin{figure}[h!]
\includegraphics[trim={3cm 0 3cm 0},clip,width=\textwidth]{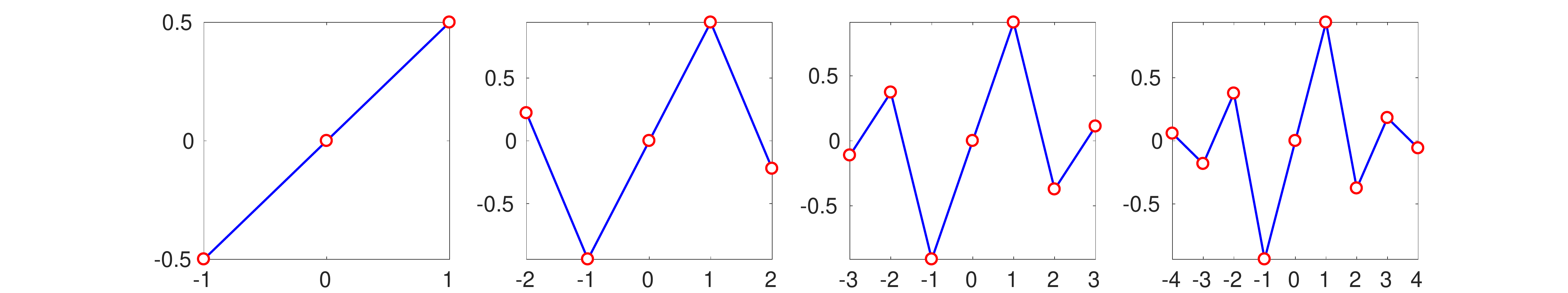}
\begin{picture}(0,0)
        \put(10,60){\rotatebox{0} {$\vo{a}_1^*$}}
        \put(55,0){$-\ns:\ns$}
	\put(165,0){$-\ns:\ns$}
	\put(275,0){$-\ns:\ns$}
	\put(385,0){$-\ns:\ns$}
        \end{picture}

\caption{Optimal coefficients for various stencil size $\ns$, approximating
first derivative with second order accuracy. The coefficients were obtained
using $\gamma(\kdx)=1$, for  $\kdx\in[0,2.5]$; and $\gamma(\kdx) = 0$ otherwise.
%\comment{The optimal schemes for $\ns=1$, $2$, $3$ and $4$ are 
%$\vo{a}_1^\ast=\protect\begin{bmatrix}1 & -2 & 1 \protect\end{bmatrix}^T$,
%$\vo{a}_1^\ast=\protect\begin{bmatrix}x & x & x & x & x\protect\end{bmatrix}^T$, 
%$\vo{a}_1^\ast=\protect\begin{bmatrix} x & x & x & x & x & x & x\protect\end{bmatrix}^T$, 
%and 
%$\vo{a}_1^\ast=\protect\begin{bmatrix}  x & x & x & x & x & x & x & x & x \protect\end{bmatrix}^T$, respectively. }
}
\label{fig:first_optim_ad}
\end{figure}
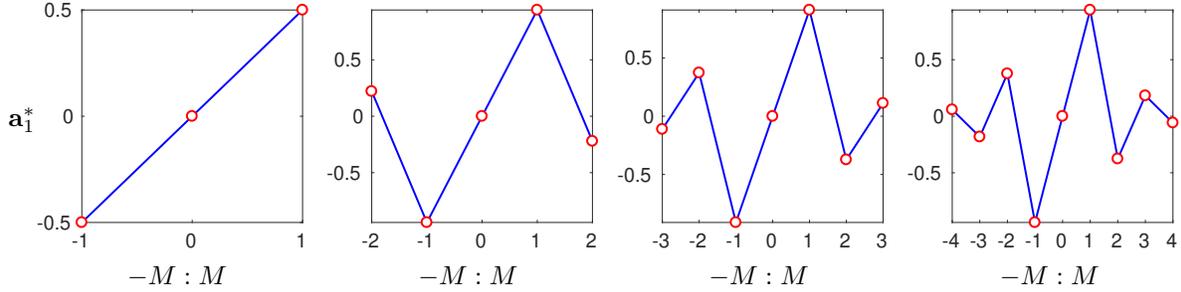

\begin{table}[h!]
\begin{center}
%{\def\arraystretch{1}
\setlength{\tabcolsep}{15pt}
\begin{tabular}{|c|c|c|c|c|}
\hline
     & $\ns=1$ &$\ns=2$ & $\ns=3$& $\ns=4$\\[0.5ex]
\hline 
%    $a_0^*$ & -2 & -2.986945912146335 & -3.067324780469417 & -3.132525936497260 \\[0.5ex]
       $a_1^*$ &  0.50 &0.941502204636976  &0.911624839168511 &  0.939273151104227 \\[0.5ex]
       $a_2^*$ &    &-0.220751102318488  & -0.372951233396604 & -0.376375957228243 \\[0.5ex]
       $a_3^*$ &    & & 0.111425875874899 & 0.182092697439389 \\[0.5ex]
       $a_4^*$ &    & & & -0.058199832241477 \\[0.5ex]
\hline
\end{tabular}%}
\end{center}
\caption{Numerical coefficients for 
the schemes in \rfig{first_optim_ad}. Note that since the 
resulting schemes are anti-symmetric $(a_{-\ns}^*=-a_{\ns}^*)$
 around the central grid point, 
only coefficients on one side are shown.}
\label{tab:first_optim_ad}
\end{table}

%\begin{figure}[h!]
%\includegraphics[width=\textwidth]{matlab/explicitWaveEquation/firstDerivative-SpectralError-Oscillatory}
%\caption{Spectral accuracy for various stencil size $\ns$, approximating first derivative with second order accuracy. The coefficients were obtained using $\gamma(\kdx)=1, w\in[0,\pi]$. \comment{Explain oscillations}.}
%\end{figure}

\begin{figure}[h!]
\includegraphics[trim={4cm 0 4cm 0},clip,width=\textwidth]{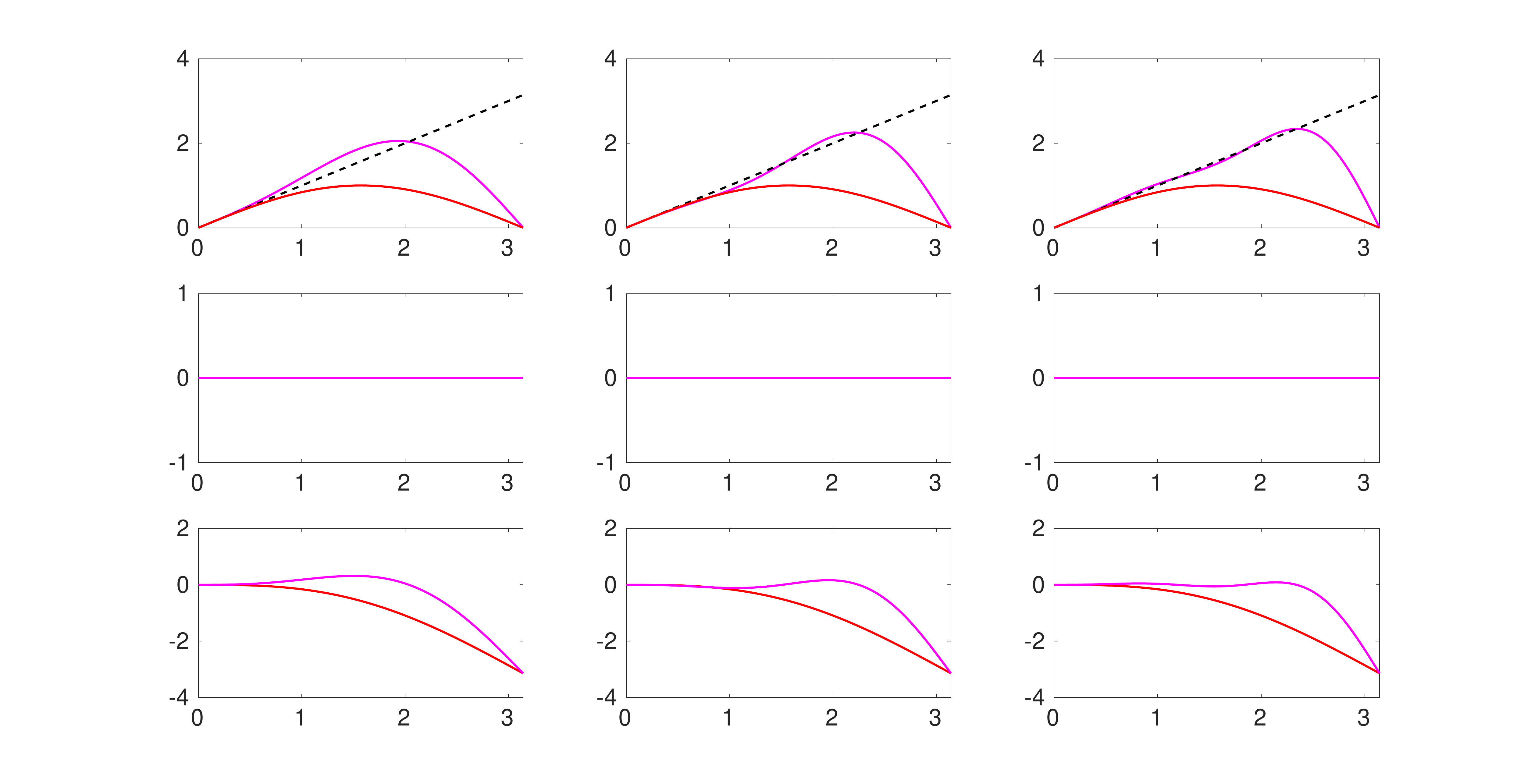}
	\begin{picture}(0,0)
        \put(70,285){$M=2$}
        \put(230,285){$M=3$}
        \put(385,285){$M=4$}
        \put(3,235){\rotatebox{90} {$\ksdx$}}
        \put(3,150){\rotatebox{90} {$\Re[e(\kdx)]$}}
        \put(3,60){\rotatebox{90} {$\Im[e(\kdx)]$}}
        \put(80,25){$\kdx$}
        \put(240,25){$\kdx$}
        \put(395,25){$\kdx$}
        \end{picture}

\caption{
Spectral accuracy for the schemes presented 
in \rfig{first_optim_ad} and \rtab{first_optim_ad}.
Top, middle and bottom rows show the modified wavenumber,
the real part, and the imaginary part of the spectral error, 
respectively.
Left, middle and right columns correspond to $M=2$, $3$, 
and $4$, respectively. Dashed black line: 
exact differentiation ($\ksdx=\kdx$).
Magenta solid line: $\ksdx$ for optimized schemes.
In all plots, red solid line corresponds to 
standard second-order scheme ($M=1$) for 
comparison.}
\label{fig:first_optim}
\end{figure}

\begin{figure}[h!]
\begin{center}
\includegraphics[width=0.5\textwidth]{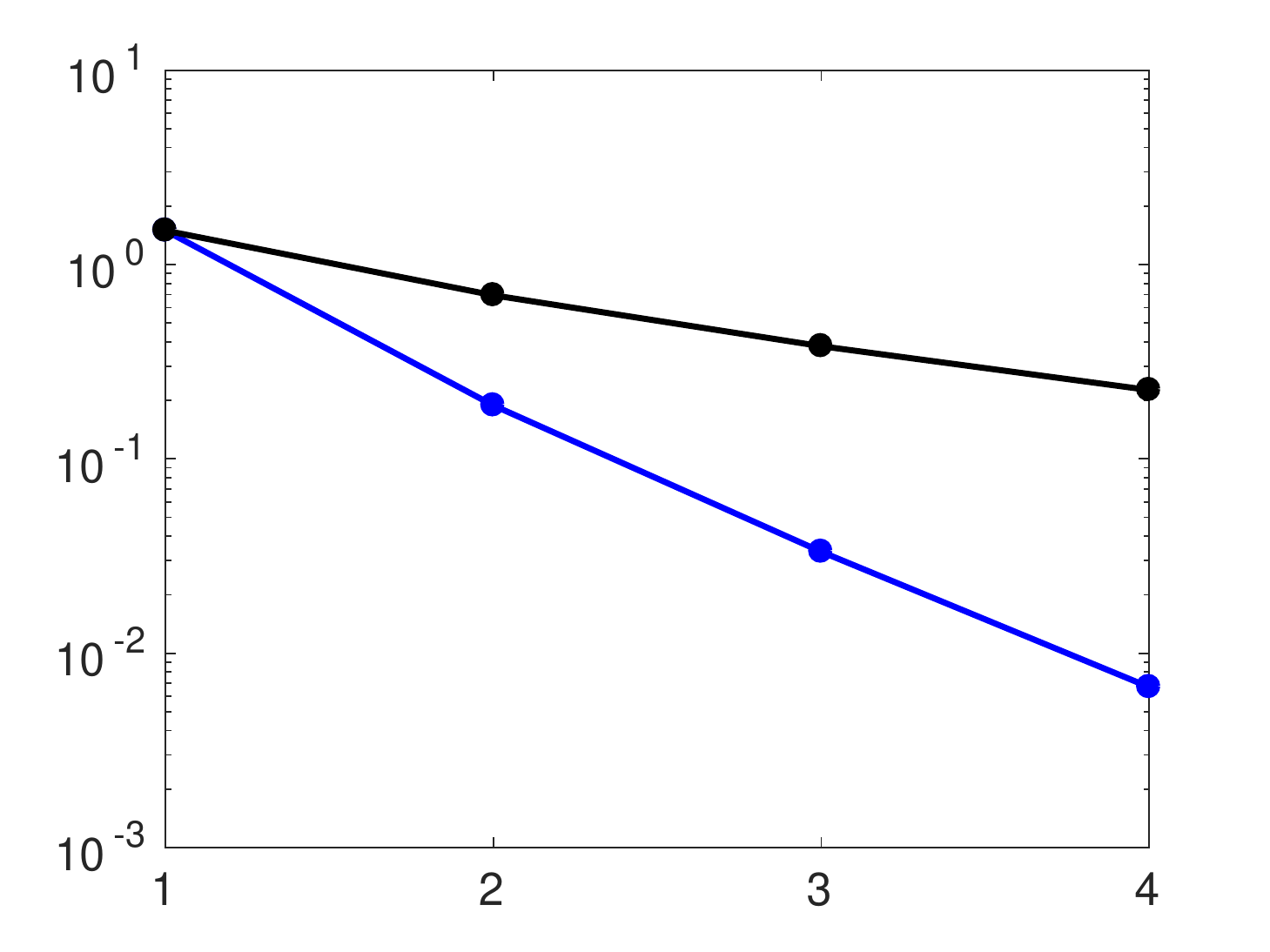}
\begin{picture}(0,0)
        \put(-245,80){\rotatebox{90} {$\| e(\kdx)\|_{\mathcal{L}_2}^2$}}
        \put(-130,-5){$\ns$}
       \end{picture}

\caption{Optimal spectral error for various stencil size $\ns$,
approximating first derivative. The 
blue line corresponds to the second order 
optimized schemes presented in 
\rfig{first_optim_ad} and \rtab{first_optim_ad}. The 
black line corresponds to the standard schemes with 
stencil size $\ns$ and $2\ns$ order of accuracy.
}
\label{fig:first_optim_L2}
\end{center}
\end{figure}

% \comment{Show that the real part of $e(\kdx)$ is zero for this case.}
\subsubsection{Effect of $\gamma(\kdx)$} \label{sec:gamma}
An important feature of the framework presented here 
for optimized schemes is that
depending upon the physics of the problem, 
one can choose what scales need to be properly resolved.
By setting $\gamma(\kdx)$ in \eqn{sp_error}
 as unity, equal weight is given to all wavenumbers while 
solving the minimization problem.
The formulation, however, is more general and 
allow us to change $\gamma(\kdx)$ to, for example,
resolve a subset of wavenumbers more accurately 
than others.
To illustrate this, consider an eight-order standard scheme ($M=4$)
for the second derivative
whose modified wavenumber can be readily computed analytically.
The relative error at each wavenumber, 
which with \eqn{modw} can be written as 
$|\ksdx^2-\kdx^2|/\kdx^2 = G_2-1$,
is shown in \rfig{gam} as a black line.
We can see that standard schemes have very low errors at
low $\eta$ but become progressively worse at high $\eta$.
In fact, there are about 15 orders of magnitude difference 
between the error incurred at low and high wavenumbers. 
This situation may present some challenges when these schemes 
are used to resolve multiscale problems where all wavenumbers
contribute to the dynamics.
In contrast, the optimized scheme with $\gamma(\kdx)=1$ in $\eta\in[0,2.5]$
shown in magenta in the figure,
results in a much flatter error in spectral space. 
The oscillatory nature of the error is due to the following. 
The optimized schemes minimize a global measure of the error based on 
the $\mathcal{L}_2$ norm of difference between $\ksdx^d$
and $\kdx^d$. Pointwise however, the resulting scheme produces a $\ksdx$ 
that can be above or below $\kdx$; only the appropriate integral is minimized.
An easy-to-see example of this behavior is seen in \rfig{first_optim} 
for $M=3$ (top-middle panel).
The crossing points between $\ksdx$ and $\kdx$ leads then 
to zero error which correspond to the down peaks in \rfig{gam} for 
the corresponding scheme.

\begin{figure}[h]
 \centering
 \subfigure{\includegraphics[width=0.49\textwidth]{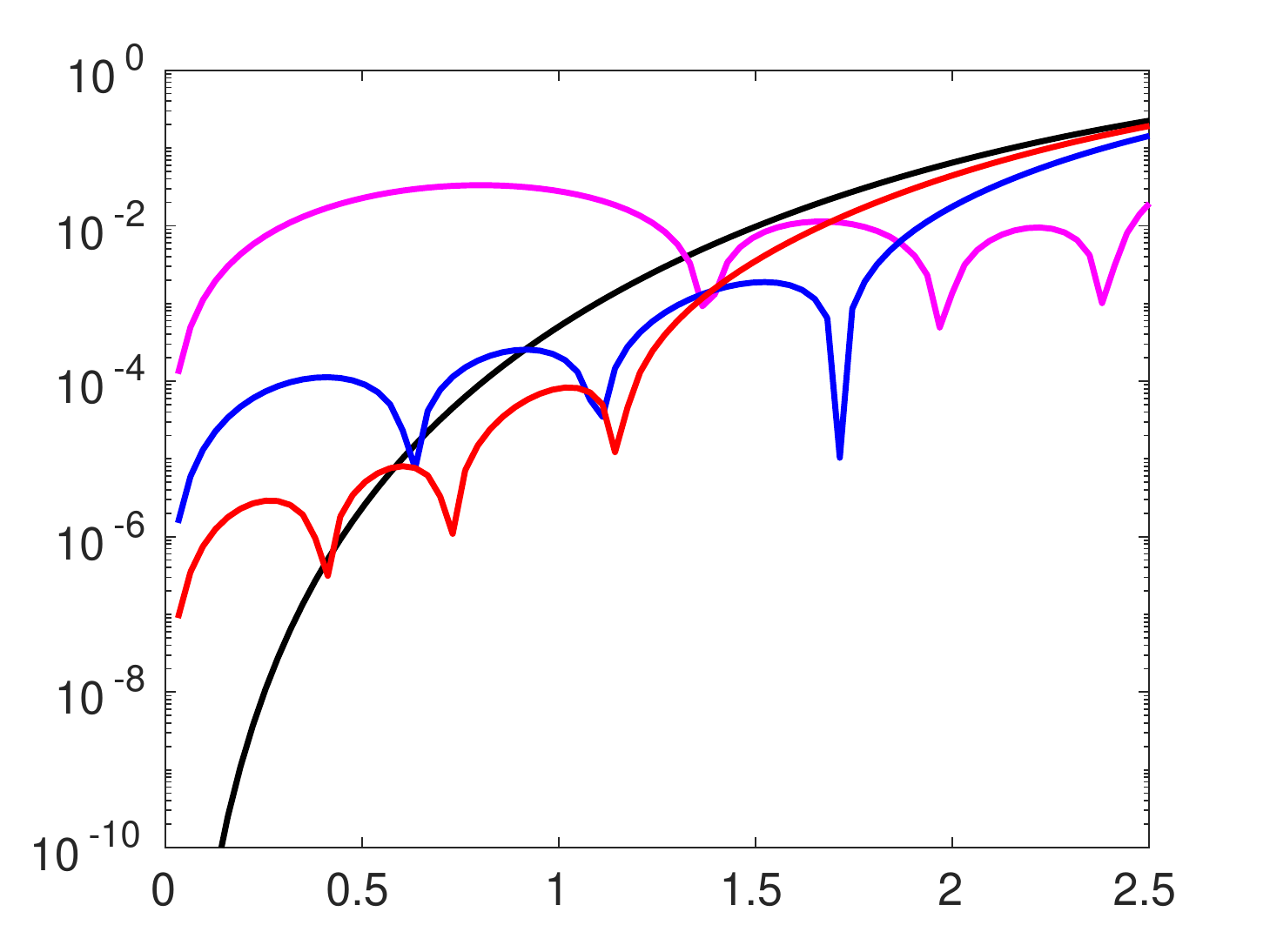}}
 \begin{picture}(0,0)
       \put(-121,-4){$\kdx$}
       \put(-240,74){\rotatebox{90} {$\frac{|\ksdx^2-\kdx^2|}{\kdx^2}$}}
 \end{picture}
 \caption{
 \figlabel{gam}
	Spectral error for various $\gamma(\kdx)$ for
	second derivative. Standard eighth order scheme (black),
	 $\gamma(\kdx) = 1$ (magenta), $\gamma(\kdx) = exp(-\kdx/0.1)$ (blue) and 
	 $\gamma(\kdx) = exp(-\kdx/0.06)$ (red) for $\kdx \in [0,2.5]$
	 and $0$ otherwise. }
\end{figure}

The other curves in \Fig{gam} illustrate the effect of changing $\gamma(\kdx)$ 
on the spectral error. For this illustration we use $\gamma(\eta)=e^{-\eta/\psi}$
where $\psi$ is a constant that 
controls how quickly the weighting function drops to zero near
the origin.
A similar form has been used in some applications 
\cite{MTWW2006, PONZ2003} though for 
some $\psi$.
As $\gamma(\kdx)$ becomes steeper (small $\psi$), the error for lower 
and higher wavenumbers decreases and increases, respectively. 
The optimized schemes then get closer to the strongly non-uniform 
distribution of error in spectral space of standard schemes.

Clearly one can use, for example, 
a banded $\gamma(\kdx)$, such that it is non-zero only for some intermmediate 
wavenumbers. Or one may need to resolve to distinct bands in wavenumber 
space which can easily be accommodated by an appropriate choice of 
$\gamma(\kdx)$. 
The choice of $\gamma(\kdx)$ 
provides a high flexibility to select how the available 
information (the value of the function at $S$ grid points)
to resolve only the scales relevant to the problem being solved.
%One can also select a complex $\gamma(\kdx)$ 
%such that the resultant scheme will be biased and 
%both real and imaginary components 
%of $e(\kdx)$ will be non-zero. 

\subsection{Stability}

As discussed in the introduction, it is common to first 
choose a scheme of given order and then
verify that spectral accuracy is acceptable. That scheme is
then used to discretize a PDE. However, before this 
numerical arrangement can be utilized, one needs to 
determine whether the fully discretized PDE is stable.
Different methods to assess stability of a given 
discretized equation have been discussed 
extensively \cite{hirsch.I}.
Our objective here is not simply to determine whether the schemes developed
above are stable or not. Rather is to incorporate the stability
requirement into a unified formulation.
%In previous sections, we coupled order requirements with 
%spectral resolution. Here we further include stability 
%requirements. 
This framework would thus 
provide, for a given $\ns$, a stable scheme of given order 
with the best spectral resolution possible for wavenumbers
of interest. %This is done next.

For this, consider the general linear partial differential equation
\begin{align}
{\partial f\over\partial t} =
\sum_{d=1}^D \beta_d{\partial^d f\over \partial x^d}. \eqnlabel{pde}
\end{align}
discretized over the entire domain with 
$N:= 2\ns_\text{max}+1$ grid points.
As before, spatial derivatives are approximated using 
a stencil of size $S=2\ns+1$, where $\ns$ can take 
values from $1$ to $\ns_\text{max}$.
The
$d^\text{th}$ derivative at the $i^\text{th}$ grid point is
parameterized by
$\vo{a}_{i,d}\in\real^{S}$. Define $\vo{A}_d\in\real^{N\times S}$ to be
the vertical stacking of $\vo{a}^T_{i,d}$, i.e.
\begin{align}
	 \vo{A}_d := \begin{bmatrix}{\vo{a}}_{1,d}^T \\ \vdots \\
	 {\vo{a}}_{N,d}^T \end{bmatrix}.
	 \eqnlabel{Ad}
\end{align}
The order accuracy constraint \eqn{exp:orderConstr} 
for every $\vo{a}_{i,d}$, can then be
compactly written as 
\begin{align}
\vo{A}_d \vo{X}_d = \vo{Y}_d,\eqnlabel{domainOrder}
\end{align}
for $d=1,\cdots,D$; and $\vo{Y}_d:=\vo{1}_{N\times 1}\otimes \vo{y}_d$.

%For $\mathcal{L}_2$ optimal spectral error, 
The cost function to be minimized is
the sum of the spectral error \eqn{sp_error} at all locations $i$:
\begin{align}
\sum_{i=1}^{N} \Ltwo{e_i(\kdx)}^2 &:=
\sum_{i=1}^{N}\sum_{d=1}^{D}\left(\vo{a}_{i,d}^T\vo{Q}_d\vo{a}_{i,d}
-2\vo{a}^T_{i,d}\vo{r}_d \right) + N\left\langle \kdx^{2d}\right\rangle, \nonumber\\
& = \sum_{d=1}^{D} \vo{v}_d^T(\vo{I}_N\otimes\vo{Q}_d)\vo{v}_d - 2\vo{v}_d^T(\vo{1}_{N\times 1}\otimes \vo{r}_d),
\eqnlabel{domain_ew}
\end{align}
where $\vo{v}_d:=\vo{vec}(\Ad^T)$, and $\vo{vec}(\cdots)$ vectorizes a
matrix by vertically stacking the columns.
The minimization of \eqn{domain_ew} subjected to 
\eqn{domainOrder} leads to optimal schemes with a given
order of accuracy. To include stability in the formulation,
we first define the vectors
\begin{align}
\vo{F} := \begin{pmatrix}
% 	f(x_1) \\ \vdots \\ f(x_{N}) 
 	f_1 \\ \vdots \\ f_N
 \end{pmatrix}, 
 \text{ and } 
 \vo{F}^{(d)} := \begin{pmatrix}
% 	\left.\frac{\partial^d f}{\partial x^d}\right|_{x_1} \\ \vdots \\ 
% 	\left.\frac{\partial^d f}{\partial x^d}\right|_{x_{N}}
 	f_1^{(d)} \\ \vdots \\ 
 	f_N^{(d)} 
 \end{pmatrix}. 	
\end{align}

Using these definitions, the finite-difference approximation for 
\eqn{pde} for all the grid points can be written compactly as 
\begin{align}
	\vo{F}^{(d)} =  \frac{1}{(\dx)^d} \Adp \vo{F},
	\eqnlabel{fdDomain}
\end{align}
where the matrix $\Adp$ contains the unknown coefficients
$\vo{a}_{i,d}$ 
arranged appropriately for the 
correct computation of the $d^\text{th}$ derivative at $i^\text{th}$
location. 
The matrix $\Adp$ can be constructed using shift operators the 
details of which are
included in  \ref{sec:app_adp}. The important element here
is that the spatial derivative is a linear operator acting on
the value of the function $f$ at all grid points.
\subsubsection{Stability of Semi-discrete Scheme}
With the spatial discretization from \eqn{fdDomain}, 
\eqn{pde} 
%with continuous time and discrete space,
%is the following ordinary differential equation
can be written as
 \begin{align}
 \dot{\vo{F}} = \left(\sum_{d=1}^{D} \frac{1}{(\Delta x)^d}\beta_d
 \Adp\right)\vo{F}, \eqnlabel{ode}
 \end{align}
 whose analytical solution can be readily obtained as
 \begin{align}
\vo{F}(t) := \exp \left(\sum_{d=1}^{D} \frac{1}{(\Delta x)^d}\beta_d t
\Adp\right) \vo{F}_0, \eqnlabel{sd:soln}
 \end{align}
where $\vo{F}_0$ is the initial condition. 
Clearly the stability of \eqn{sd:soln} is governed by the coefficients
$\Adp$. 

%In order to understand the stability of the different discretizations
%we first look at the semi-discrete system, that is \eqn{pde}
%with continuous temporal but discrete spatial operators. 
In order to understand the conditions under which the scheme defined by
$\Adp$ is stable 
consider a single Fourier mode as in \eqn{singleFmode}.
The approximate $d^\text{th}$ derivative can be conveniently  written 
in terms of the modified wavenumber 
as $(j\mw)^d \hat{f}$ and the original PDE becomes
\begin{equation}
{d\hat{f}\over dt} = \sum_{d=1}^D \beta_d(j\mw)^d \hat{f} .
\end{equation}
The solution, with $\hat{f}_0$ being the value of $\hat{f}$ at $t=0$, 
is given by
\begin{equation}
{\hat{f} \over \hat{f}_0}= \exp \left[ \sum_{d=1}^D \beta_d(j\mw)^d t \right] .
\end{equation}
If the solution to the original PDE is non-increasing in time,
then the discretization is considered stable if no Fourier mode
grows in time. This is clearly satisfied if 
\begin{equation}
\Re\left\{ \sum_{d=1}^D \beta_d(j\mw)^d \right\} \le 0
\eqnlabel{realStability}
\end{equation}

From lemma \ref{lem:imagZero} and \ref{lem:realZero}, 
we can conclude that
\begin{align*}
(j\ksdx)^d = (\vo{C}^T(\kdx) + j\vo{S}^T(\kdx))\ad = \left\{\begin{array}{r} j\vo{S}^T(\kdx)\ad, \text{ for odd derivative},\\
\vo{C}(\kdx)^T\ad, \text{ for even derivative},\end{array}\right.
\end{align*}
because $\vo{C}^T(\kdx)\ad = 0$ for odd derivatives and $\vo{S}^T(\kdx)\ad =
0$ for even derivatives. 
Then,
\begin{align}
 \sum_{d=1}^D \beta_d(j\mw)^d   &= j\beta_1\vo{S}^T(\kdx)\vo{a}_1 + \beta_2\vo{C}^T(\kdx)\vo{a}_2 + j\beta_3\vo{S}^T(\kdx)\vo{a}_3 + \beta_4\vo{C}^T(\kdx)\vo{a}_4 + \cdots,\nonumber\\
 & = \left(  \beta_2\vo{C}^T(\kdx)\vo{a}_2 + \beta_4\vo{C}^T(\kdx)\vo{a}_4 + \cdots \right) + j\left(\beta_1\vo{S}^T(\kdx)\vo{a}_1 + \beta_3\vo{S}^T(\kdx)\vo{a}_3 + \cdots \right).
\end{align}
Therefore, \eqn{realStability} implies
\begin{equation}
  \beta_2\vo{C}^T(\kdx)\vo{a}_2 + \beta_4\vo{C}^T(\kdx)\vo{a}_4 + \cdots \leq 0
\eqnlabel{realStability1}
\end{equation}
which provides a general constraint to assure stability of the 
semidiscrete system.

%\raktim{ Rewrite this: We next show that, if the original PDE is non-increasing in time, i.e. stable,
%\eqn{realStability1} is implicitly satisfied. PDEs are non-increasing in time}
%if $\beta_{2q} := (-1)^{q+1}\gamma_{2q}^2$ for some $\gamma_{2q}\in\real$  
%\comment{[Can we cite this?][DD: this is a necessary condition, right? The sign
%of some of the terms in the PDE case may be positive as long as the sum is negative. Here we are imposing every term to be ``stable''.]}. \raktim{RB: You are correct ... we then have to define what we mean by stable PDEs in terms of $\beta_{2q}$. This will impose some condition on $\beta_{2q}$ that can be used for deriving the stability constraint in the optimization. Is the following correct?: for non increasing solution $\Re\left\{ \sum_{d=1}^D \beta_d(jw)^d \right\} \le 0$, i.e. original PDE will satisfy the constraint with actual wave number?}

Clearly stability for the semi-discrete system depends only on the even derivatives. 
In that case $d$ can be written as $2q$ for $q=1,2,\cdots$,
and $(j\ksdx)^{d} =
(-1)^q\ksdx^{2q}$. Since the optimization of spectral error guarantees
$(j\ksdx)^{d} = \vo{C}^T(\kdx)\vo{a}_{2q}$, we have
$\vo{C}^T(\kdx)\vo{a}_{2q}
= (-1)^q\ksdx^{2q}$. Consequently, the sign of $\vo{C}^T(\kdx)\vo{a}_{2q}$
alternates with $q$, i.e.
\begin{align*}
\text{for } q = 1, \;&    \vo{C}^T(\kdx)\vo{a}_2 = -\ksdx^2 <= 0, \\
\text{for } q = 2, \;&    \vo{C}^T(\kdx)\vo{a}_4 = \ksdx^4 >= 0, \\
\text{for } q = 3, \;&    \vo{C}^T(\kdx)\vo{a}_6 = -\ksdx^6 <= 0, 
\end{align*}
and so on. Therefore, if the coefficients of the original PDE can be 
written as $\beta_{2q} := (-1)^{q+1}\gamma_{2q}^2$ for some
$\gamma_{2q}$, then 
\eqn{realStability1} is implicitly satisfied as it reduces to a sum of negative numbers.
In general, if $\beta_{2q}$ does not have sign as required by $(-1)^{q+1}\gamma_{2q}^2$, 
then, while that particular derivative is unstable, 
the system is still stable if 
\begin{equation}
-\beta_2 \ksdx^2 +\beta_4 \ksdx^4 -\beta_6 \ksdx^6 + \cdots \le 0.
\end{equation}

\subsubsection{Stability of Fully-Discrete Scheme}
While semi-discrete analyses provide some information about
the temporal behavior of the discretized spatial derivatives
in the PDEs, one is ultimately interested in the stability 
of the fully discretized system.
For concreteness assume that the time discretization of \eqn{pde}
is done using a \textit{forward difference}. With 
\eqn{fdDomain} for the spatial discretization 
we write \eqn{pde} as:

\begin{align}
	\vo{F}^{k+1} = \left(\vo{I}_{N} +  \sum_d \frac{\Delta t}{(\Delta x)^d}\beta_d \Adp\right)\vo{F}^k.	
	\eqnlabel{discDyn}
\end{align}
This is a linear discrete-time system in $\vo{F}$, where the system
matrix is linearly dependent on the stencil coefficients $\vo{A}_d$,
discretization parameters $\dt$ and $\dx$, 
and the coefficients $\beta_d$. In what follows 
we assume $\dx$ and $\beta_d$ are
given which is a typical situation in simulations of physical systems.
The objective then is to determine $\dt$ and $\vo{A}_d$ so that stability 
is achieved subjected to specifications in order of accuracy and spectral 
resolution.

In general,
stability is guaranteed if the spectral radius of the evolution 
matrix is bounded by unity
\cite{hirsch.I}:
\begin{equation}
\lambda_\text{max} \left(\vo{I}_{N} +  \sum_d
\frac{\Delta t}{(\Delta x)^d}\beta_d \Adp\right)\leq 1 .
\end{equation}
Because the spectral radius is bounded by matrix norms,
we guarantee stability by bounding the $2$-norm, i.e.
\begin{equation}
\left\|\vo{I}_{N} +  \sum_d \frac{\Delta t}{(\Delta x)^d}\beta_d
\Adp\right\|_2 \leq 1 .
\eqnlabel{2norm_stab}
\end{equation}
Using Schur complements, this inequality is equivalent to
\begin{align}
\begin{bmatrix} \vo{I}_{N} & \left(\vo{I}_{N} +  \sum_d \frac{\Delta t}{(\Delta x)^d}\beta_d \Adp\right)^T\\
\left(\vo{I} +  \sum_d \frac{\Delta t}{(\Delta x)^d}\beta_d \Adp\right) & \vo{I}_{N}
\end{bmatrix} \geq 0. \eqnlabel{explicit:Stability}
\end{align}
This (inequality) constraint coupled with the order of accuracy (equality)
constraint \eqn{domainOrder} completes the set of constraints for the 
optimization problem aimed at minimizing the spectral error given by 
\eqn{domain_ew}.

Unfortunately, 
the matrix inequality in \eqn{explicit:Stability} has products of $\dt$ and
$\vo{A}_d$, which makes the problem non convex. 
Thus, more general techniques than convex optimization have to be employed 
to solve the problem.
%We circumvent this using two approaches. 
Here we present two approaches. 
In the first approach, we determine $\vo{A}_d$ analytically that
minimizes the spectral error and then maximize $\dt$ for which stability is
achieved. In the second approach, we assume a value for $\dt$ and determine
$\vo{A}_d$ which minimizes spectral error and guarantees stability. Both these
approaches are discussed in detail below.

\subsubsection{Given spectrally optimal $\vo{A}_d$, maximize $\dt$ and
guarantee stability simultaneously}\label{sec:approach1}
In this approach, we solve for the optimal $\vo{a}_{i,d}$ analytically
that,
for a given order, minimizes the spectral error and then maximize $\dt$
for which stability is guaranteed. 
Optimization of the spectral error with
given order of accuracy is given by previous formulation as
\begin{align*}
&\min_{\{\vo{A}_d\}_{d=1}^{D}}  \;\; \sum_{d=1}^{D} \vo{v}_d^T(\vo{I}_N\otimes\vo{Q}_d)\vo{v}_d - 2\vo{v}_d^T(\vo{1}_{N\times 1}\otimes \vo{r}_d),
\end{align*}
subject to
\begin{align*}
&\vo{A}_d\vo{X}_d  = \vo{Y}_d, \text{ for } d = 1,\cdots,D,
\end{align*}
where $\vo{v}_d:=\vo{vec}(\Ad^T)$.  

We observe that the cost and constraint functions are separable with respect to $d$, and thus can be independently optimized using,
\begin{align}
&\min_{\vo{A}_d}  \;\;  \sum_{d=1}^{D} \vo{v}_d^T(\vo{I}_N\otimes\vo{Q}_d)\vo{v}_d - 2\vo{v}_d^T(\vo{1}_{N\times 1}\otimes \vo{r}_d), \text{ subject to } \vo{A}_d\vo{X}_d  = \vo{Y}_d, \eqnlabel{explicit:optimD}
\end{align}
for $d = 1,\cdots,D$. The linear constraint $\Ad\vo{X}_d = \vo{Y}_d$ can be written as
$$
\left(\vo{I}_{N}\otimes \vo{X}_d^T\right)\vo{v}_d = \vo{Y}_d^T.
$$

Therefore, the optimization problem in \eqn{explicit:optimD} can be written
as
\begin{align}
\min_{\vo{v}_d}\;\; 
& 
  \vo{v}_d^T\left(\vo{I}_{N}\otimes \vo{Q}_d\right)\vo{v}_d -
  2\vo{v}_d^T\left(\vo{1}_{N\times 1}\otimes \vo{r}_d\right), 
  \\
 & \text{ subject to }
\left(\vo{I}_{N}\otimes \vo{X}_d^T\right)\vo{v}_d =  \vo{Y}_d^T ,
\end{align}
which has the analytical solution
\begin{align}
\begin{pmatrix}
\vo{v}_d\\\boldsymbol{\Lambda}_d 
\end{pmatrix}^\ast = \begin{bmatrix} 
\vo{I}_{N}\otimes \vo{Q}_d & \vo{I}_{N}\otimes \vo{X}_d\\
\vo{I}_{N}\otimes \vo{X}_d^T\ & \vo{I}_{N}\otimes\vo{0}
\end{bmatrix}^{-1}
\begin{pmatrix}
\vo{1}_{N\times 1}\otimes \vo{r}_d \\ \vo{Y}_d^T
\end{pmatrix},
\eqnlabel{ana_sol_approach1}
\end{align}
where $\boldsymbol{\Lambda}_d$ is the vector of Lagrange multipliers
associated with the constraints. 
Finally, we can recover $\vo{A}_d^\ast$,
the coefficients of 
the scheme that minimizes the spectral error and is stable,
from $\vo{v}_d^\ast$.
While intuitively it is clear that this solution should be 
the same at all grid points in the domain, here we present,
for completion, a proof as:

\begin{lem} \label{lem_invariance}
Optimal explicit symmetrical finite-difference approximation is invariant
of the grid point location in the domain.
\end{lem}
\begin{proof}
See \ref{sec:app_invariance}.
\end{proof}

From Lemma \ref{lem_invariance}, we write the complete system as 
a stacking of local solutions:
\begin{align}
\vo{A}_d^\ast = \vo{1}_{N\times 1}\otimes {\vo{a}^{\ast}_d}^T. 
\eqnlabel{explicit:Adstar}
\end{align}
Therefore, the maximum $\dt$ for which stability is guaranteed is
obtained by solving the optimization problem \eqn{explicit:Stability},
that is
\begin{equation}
\max_{\dt} \text{ subject to } 
 \begin{bmatrix} \vo{I}_{N} & \left(\vo{I} +  \sum_d \frac{\Delta t}{(\Delta x)^d}\beta_d \Adp\right)^T\\
\left(\vo{I} +  \sum_d \frac{\Delta t}{(\Delta x)^d}\beta_d \Adp\right) & \vo{I}_{N}
\end{bmatrix} \geq 0, \eqnlabel{explicit:maxTime}
\end{equation}
where $\Adp$ is determined using \eqn{linearAd} and $\vo{A}_d^\ast$ from
\eqn{explicit:Adstar}. The optimization in \eqn{explicit:maxTime} is a
convex optimization problem in $\dt$ and can be efficiently solved
numerically using software such as \texttt{cvx}\cite{grant2008cvx}.
This formulation can then be used to obtain stability limits
of the fully-discrete system.

As an example, consider an advection-diffusion equation, that 
is \eqn{pde} with $D=2$.
For given values of $\beta_1$, $\beta_2>0$, and $\dx$, one can solve
the optimization problem to obtain the largest $\dt$ for 
which a previously obtained optimal $\vo{A}_d^\ast$ remain stable.
Note that the stability of the scheme will depend only 
on the non-dimensional
parameters $\dt\beta_d/\dx^d$. For $D=2$, these are commonly 
called convective and diffusive CFL, that is, $r_c:=\dt\beta_1/\dx$ and 
$r_d:=\dt\beta_2/\dx^2$, respectively.
Stability regions can then be obtained by 
e.g.\ sweeping values of $\dx$ and plot results in 
terms of $r_c$ and $r_d$.

This is what we show in 
\Fig{explicit:rc_rd} for the advection-diffusion equation for 
different values of $\ns$. 
For $\ns=1$ the numerically evaluated 
stability region is the same as the analytical form found 
in textbooks for the standard second-order 3-point stencil
for first and second derivatives, namely
$r_c^2\le 2 r_d\le 1$ \citep{hirsch.I}.
As $\ns$ is increased
the stability region becomes smaller.
%. Bisection method was used to obtain this 
%plot for finding the maximum $\Delta t$ for which the optimum 
%coefficients guarentee stability. For larger stencil size $\ns$, 
%smaller $\Delta t$ (lying within the area bounded 
%by the lines in \Fig{explicit:rc_rd}) has to be used to ensure stability.  
%Although $\beta_d$=[10,1] was used for \Fig{explicit:rc_rd} but the
%stability region 
%is independent of these values. 
This is also generally consistent with standard schemes for which
as the order (and stencil size) 
increases the stability region shrinks. In both cases,
this is related to the decreasing dissipation at high wavenumbers
as spectral resolution improves which could trigger instabilities.

In summary, the formulation here
allows us to obtain, for optimal schemes of given order and
maximum spectral resolution, 
the largest step size which guarantees stability 
or, more generally, regions of stability.

\begin{figure}[h!]
\begin{center}
\includegraphics[width=0.5\textwidth]{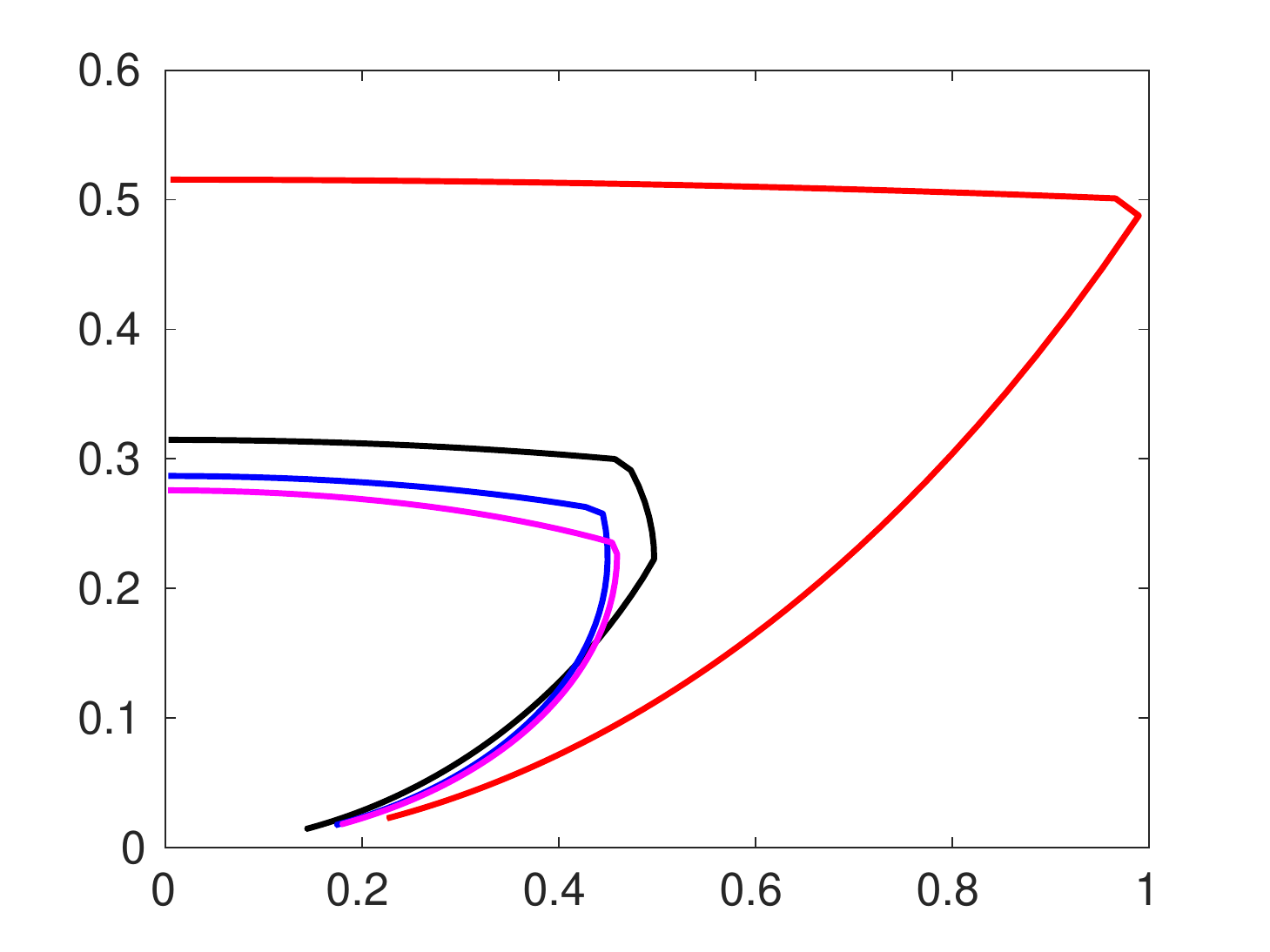}
        \begin{picture}(0,0)
        \put(-240,90){\rotatebox{90} {$r_d$}}
        \put(-120,0){$r_c$}
       \end{picture}
\caption{Stability region for the advection-diffusion 
equation with optimal coefficients 
$\vo{a}_1^*$ and $\vo{a}_2^*$ for different stencil sizes: $\ns=1$ 
(red), $\ns=2$ (black), $\ns=3$ (blue), $\ns=4$ (magenta).
}
\figlabel{explicit:rc_rd}
\end{center}
\end{figure}
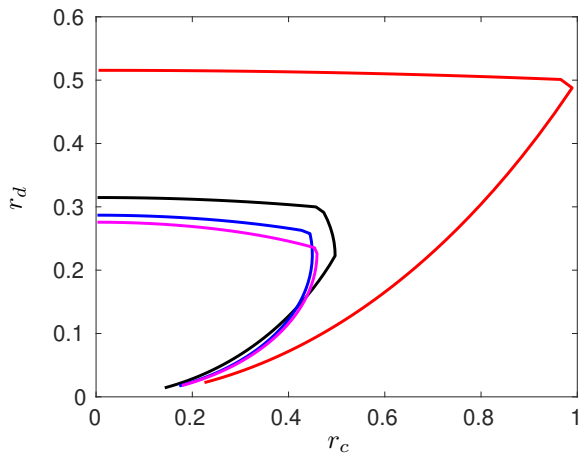
\subsubsection{Given $\dt$, optimize spectral error and guarantee
stability simultanously}\label{sec:approach2}
In this approach, for a given $\dt$ 
we solve for $\vo{A}_d$ that simultaneously minimizes spectral error,
achieves given order accuracy, and guarantees stability. 
\colk{This approach is of practical relevance as the time-step $\dt$ 
would depend 
on the fastest physical process in the problem. For example, when 
chemical reactions are present in a flow, the 
Damkohler number (the ratio of flow time scales to chemistry time scales) 
can be very high which means the time 
scales of reactions are considerably smaller than flow time scales
which forces $\dt$ 
to be small for an accurate solution.}
%and thus it is important to 
%use schemes that are spectrally optimal for $\dt$ determined by the 
%physics of the problem.}  

The formulation here is done by combining the results in previous sections
to write the following optimization problem,

\begin{equation}\left.
\begin{aligned}
&\min_{\{\vo{A}_d\}_{d=1}^{D}}  \;\; \sum_{d=1}^D  \sum_{i=1}^{N}\left(\vo{a}_{i,d}^T\vo{Q}_d\vo{a}_{i,d} -2\vo{a}^T_{i,d}\vo{r}_d\right), \\
\text{subject to} &\\
&\vo{A}_d\vo{X}_d  = \vo{Y}_d, \text{ for } d = 1,\cdots,D;\\[2mm]
& \begin{bmatrix} \vo{I}_{N} & \left(\vo{I} +  \sum_d \frac{\Delta t}{(\Delta x)^d}\beta_d \Adp\right)^T\\
\left(\vo{I} +  \sum_d \frac{\Delta t}{(\Delta x)^d}\beta_d \Adp\right) & \vo{I}_{N}
\end{bmatrix} \geq 0.	 
\end{aligned}\;\;\right\} \eqnlabel{explicitSynth}
\end{equation}
Equation \eqn{explicitSynth} is a convex optimization problem
\cite{ben2001lectures} in variables ${\{\vo{A}_d\}_{d=1}^{D}}$, with
$\vo{A}_d\in\real^{N\times S}$, and can be efficiently solved numerically
using \texttt{cvx} \cite{grant2008cvx}. The time step $\dt$ can be
maximized by
iteratively solving \eqn{explicitSynth} with increasing $\dt$ until the
problem
becomes unfeasible. Since the spectra of a matrix is continuous in terms of the
elements, the maximum $\dt$ can be  determined using a bisection algorithm
\cite{stoer2013introduction}.

In \Fig{explicit:Err_Stab} we show the result of such a 
computation again for an advection-diffusion equation discretized with 
second-order approximations.
In particular, we show contours of the sum of the 
spectral error $\|e(\kdx)\|^2_{\mathcal{L}_2}$ in the first and second derivatives
for different values of $M$ normalized by the value for $M=1$. 
Colored areas represent regions where schemes 
are stable. For $M=1$, the stable region is well known as
discussed above and presents values of 1.0 due to the normalization
chosen. 
Interestingly, as $M$ increases the stability area
increases instead of decreasing as in \Fig{explicit:rc_rd}.
This illustrates an interesting aspect of the formulation. 
By providing additional degrees of freedom (bigger stencil size) but
fixing the formal order of accuracy, \eqn{explicitSynth} 
minimizes spectral error with the condition that the scheme be stable
when some fixed $\dt$ is used.
As can be seen, for $M=4$, schemes can remain stable for very large 
CFL numbers (an order of magnitude larger than for standard schemes)
though with larger spectral error.
In other words, the framework allows one to trade 
off error with stability.
If one is interested in stationary states, for example, one can use
the optimal scheme with $M=4$ with very large time steps
in a fully explicit arrangement to solve transients. After 
the desired steady state is attained one can 
reduce $\dt$ with corresponding optimized schemes
to reduce errors. In fact, by comparing 
part (a) and (d) in the figure, we can see that for 
regions where $M=1$ is stable, $M=4$ provides errors which 
could be two orders of magnitude smaller. 
Conversely, for similar accuracy, schemes with $M=4$ can 
be used with much larger $\dt$ than a scheme with $M=1$.
Note that this approach could prove beneficial at
very large levels of parallelism where implicit 
schemes in time, while providing good stability 
characteristics for large time steps, become challenging  
due to the necessity to invert large matrices.

% It is also expected that there will be a tradeoff
% between $\dt$ and spectral accuracy ($\|e(\kdx)\|^2_{\mathcal{L}_2}$). 
% \Fig{explicit:Err_Stab} illustrates this tradeoff between stability 
% and spectral accuracy. The plot is obtained by solving \eqn{explicitSynth} 
% and maximizing $\dt$ using bisection algorithm. The colors are
% proportional 
% to spectral accuracy for which stability was achieved. The spectral error 
% is normalized by the error obtained for the standard second 
% order scheme. We see that, for a 
% fixed $\dx$, it is possible to get stable explicit schemes with large $\dt$; 
% but at the expense of spectral accuracy. 

\begin{figure}[h!]
\includegraphics[trim={3.5cm 0 0 0},clip,width=\textwidth]{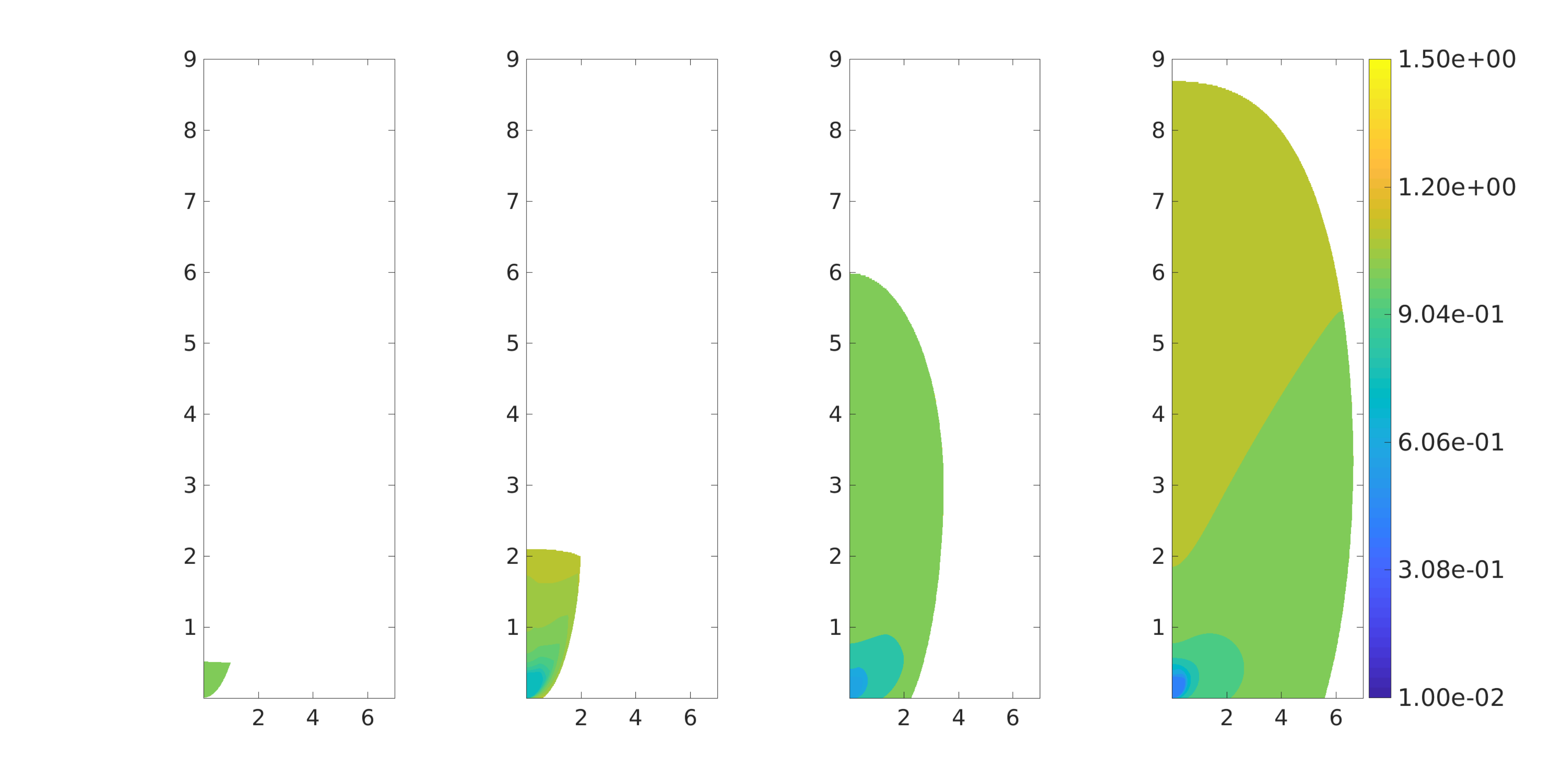}
\begin{picture}(0,0)
\put(75,235){\scriptsize $(a)$}
\put(180,235){\scriptsize $(b)$}
\put(285,235){\scriptsize $(c)$}
\put(390,235){\scriptsize $(d)$}
\put(10,140){$r_d$}
\put(60,18){$r_c$}
\put(165,18){$r_c$}
\put(270,18){$r_c$}
\put(375,18){$r_c$}
\end{picture}
\caption{Tradeoff between stability and spectral accuracy with respect to
$(r_d,r_c)$, for $\beta_1=10$ and
$\beta_2=1$. 
Contours represent sum of spectral errors $\|e(\kdx)\|^2_{\mathcal{L}_2}$
for first and second derivatives using optimal coefficients obtained
as a solution of \eqn{explicitSynth} with fixed $\dt$ for the 
advection-diffusion equation.
Errors are normalized by the spectral error 
obtained for $M=1$. Different 
stencil sizes shown in (a) $M=1$, (b) $M=2$, (c) $M=3$, and (d) $M=4$.}
\figlabel{explicit:Err_Stab}
\end{figure}
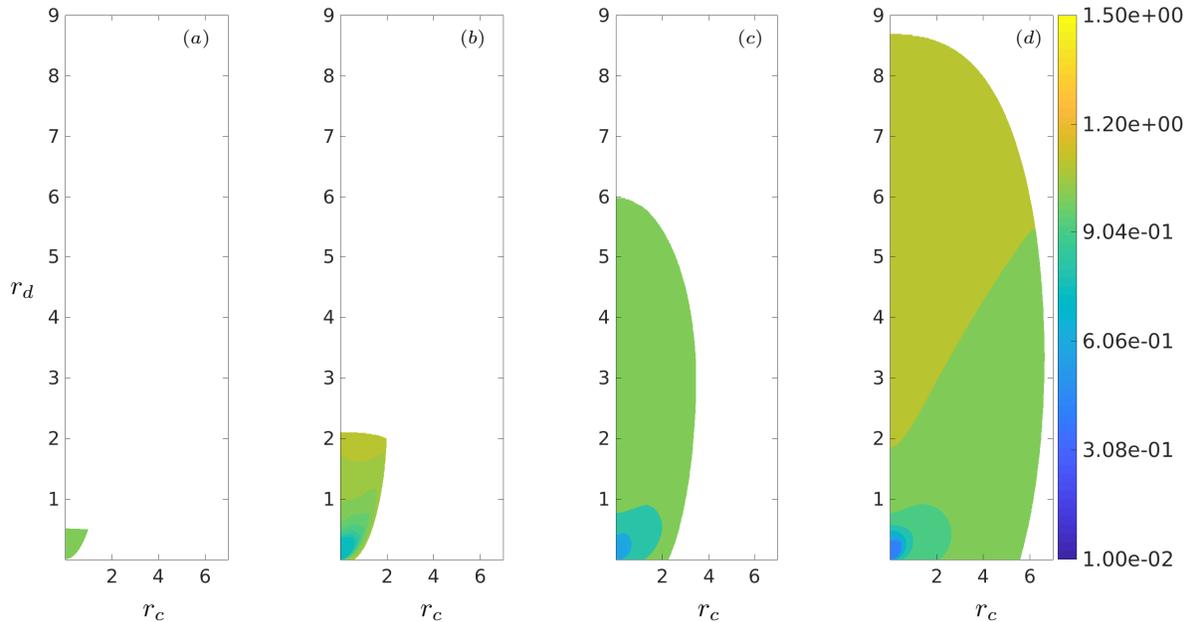

It is interesting to observe how 
optimal coefficients change to maintain stability 
for very large time steps. Since biased (upwind) schemes 
tend to be more dissipative and stable, one would expect 
that gains in stability are mediated by losses in accuracy 
due to biasing. This is, in fact, what is observed.
In order to quantify deviations from symmetry 
($a_{i}=a_{-i}$,
where $a_i$s are the coefficients in \eqn{explicit}
that are obtained from $\vo{a}_d$) for 
even derivatives or anti-symmetry 
($a_{i}=-a_{-i}$) for odd
derivatives we define the metric:
\begin{equation}
A=\sum_{i=1}^{M}\frac{|a_{-i}+(-1)^{d+1}a_{i}|}{|a_{-i}|+|a_{i}|},
\eqnlabel{defA}
\end{equation}
for the $d$-th derivative.
Clearly $A=0$ when the coefficients are 
symmetric (for even derivatives) or anti-symmetric 
(for odd derivatives).
Larger values of $A$ are associated with increasingly 
biased approximations.
The denominator in \eqn{defA} is included such that the contribution 
from all coefficients are of the same order. This is needed because
coefficients tend to decrease in magnitude with 
distance from the point where derivatives are computed.

In \rfig{skew}(a) we plot $A$ for the first (a) and 
second (b) derivatives, respectively for $M=2$. 
As before, colored 
areas correspond to stable conditions.
The black line in \Fig{skew}(a) corresponds to the stability region 
given in \Fig{explicit:rc_rd} for $M=2$.
When we are within this boundary, the value of $A$ for the 
first derivative is very small 
which suggests that all extra degrees  
of freedom are used to minimize the spectral error for which
symmetric stencils are generally better.
However, as we move outside of this area, these 
extra degrees of freedom are needed to satisfy the stability 
constraint. By losing symmetry or anti-symmetry error increases 
(as seen in \Fig{explicit:Err_Stab}) but stability improves.
In fact, asymmetry for the first derivative implies that 
the error has both real and imaginary parts as seen in 
\rfig{errM2}. The biased nature of 
coefficients leads not only to dispersion errors (as assured 
by lemma \ref{lem:realZero} with symmetric coefficients) but also 
to dissipation errors as $\Re[e(\eta)]\ne 0$. 
Both dissiption and dispersion error 
increase as we increase the value of $r_c$ and $r_d$.

A different behavior is observed for the second derivative
whose contours of $A$ are shown in \Fig{skew}(b).
Here we see that $A$ is very small for any $r_c$ and $r_d$
within the limits of stability. The implication of this is 
that stability is essentially governed by the first derivative,
which may be intuitive. It is indeed common to stabilize
fluid flow simulation codes by special treatment of the convective
terms while leaving diffusive terms approximated by standard 
central differences. 
What our results show, beyond heuristic 
stabilization considerations is that, indeed, biasing
convective terms and keeping central schemes for diffusion 
is the spectrally optimal way of achieving stability. 
\colk{They also show that, unlike \cite{TKS2003}, so-called 
anti-diffusion 
%a positive imaginary part of the modified wavenumber 
does not lead to 
instability due to the compensating effect of the other term 
in the equation}.
%\colr{[DD: Komal, did you confirm the imaginary part 
%of the wavenumber for these schemes is positive? Also, we need
%to confirm that what we are saying is correct. This is very counterintuitive, 
%is it not? I expected *more* negative imaginary part instead.]}
For completion, in \Fig{skew}(c) we show 
the change in optimal coefficients for the first derivative 
for the conditions represented 
by different symbols in part (a). We clearly see how 
the scheme changes from a completely anti-symmetric 
configuration (red circles) to a more biased set of 
coefficients. We emphasize this stable configurations
are optimal in spectral space.

\colk{We close this section with two remarks about the stability 
limits computed here.
First, in terms of computational cost we note that 
as a consequence of Lemma \ref{lem_invariance} 
for periodic domains, the optimization problem in \eqn{explicitSynth} 
has a total of $n=Sd$ degrees of freedom, which is independent 
of the total number of grid points $N$. 
Since the computational cost in terms 
of time and memory for state-of-the-art optimization algorithms is $\mathcal{O}(n)$ 
\cite{Zhang2018}, obtaining the numerical scheme 
represents a small fraction of the total 
computational cost. Furthermore, since the time step size is typically
controlled either by $r_c$ or $r_d$, the computational cost of solving
the equation scales as $N^2$ or $N^3$, respectively, making the optimization
cost increasingly smaller as the problem size increases.
Second, 
we note that the 
stability limits shown in \Fig{explicit:Err_Stab} were 
obtained numerically with machine precision. 
Thus, caution needs to be exercised when 
selecting conditions in the $r_c$-$r_d$ plane 
very close to the stability boundaries.}

%In \rfig{skew}(b) the optimal coefficients 
%are plotted for the $r_c,r_d$ values given by the corressponding 
%symbol in \rfig{skew}(a). As the parameter values increase, the coefficients 
%become more staggered and result in larger spectral error. The 
%effect on the even derivative is minimal. \colr{(WHY?)}   
%The advantage is that depending 
%upon the level of accuracy required for the problem in hand, 
%optimized schemes allow the use of a very large $\dt$ while 
%ensuring stability.
%\begin{figure}[h!]
%\includegraphics[width=\textwidth]{figs/err_vs_rd_vs_rc.eps}
%\includegraphics[trim={3.5cm 0 0 0},clip,width=\textwidth]{newfigs/rc_rd_ratio.eps}
%\begin{picture}(0,0)
%\put(10,140){$r_d$}
%\put(65,18){$r_c$}
%\put(170,18){$r_c$}
%\put(273,18){$r_c$}
%\put(378,18){$r_c$}
%\put(45,255){$M=1$}
%\put(150,255){$M=2$}
%\put(253,255){$M=3$}
%\put(358,255){$M=4$}
%\end{picture}
%\caption{Tradeoff between stability and spectral accuracy with respect to
%$(r_d,r_c)$, for $\beta_1=10$ and
%$\beta_2=1$. \colr{[DD: Komal - need to describe here what we are showing
%and what the four subplots
%are?]} Normalized spectral error inside the stability region for advection-diffusion 
%equation for optimal coefficents obtained by fixing $\dt$ for different stencil size.}
%\figlabel{explicit:Err_Stab}
%\end{figure}

\begin{figure}[h!]
\begin{center}
\subfigure{\includegraphics[trim={0cm 0cm 0.1cm 0cm},clip,width=0.45\textwidth]{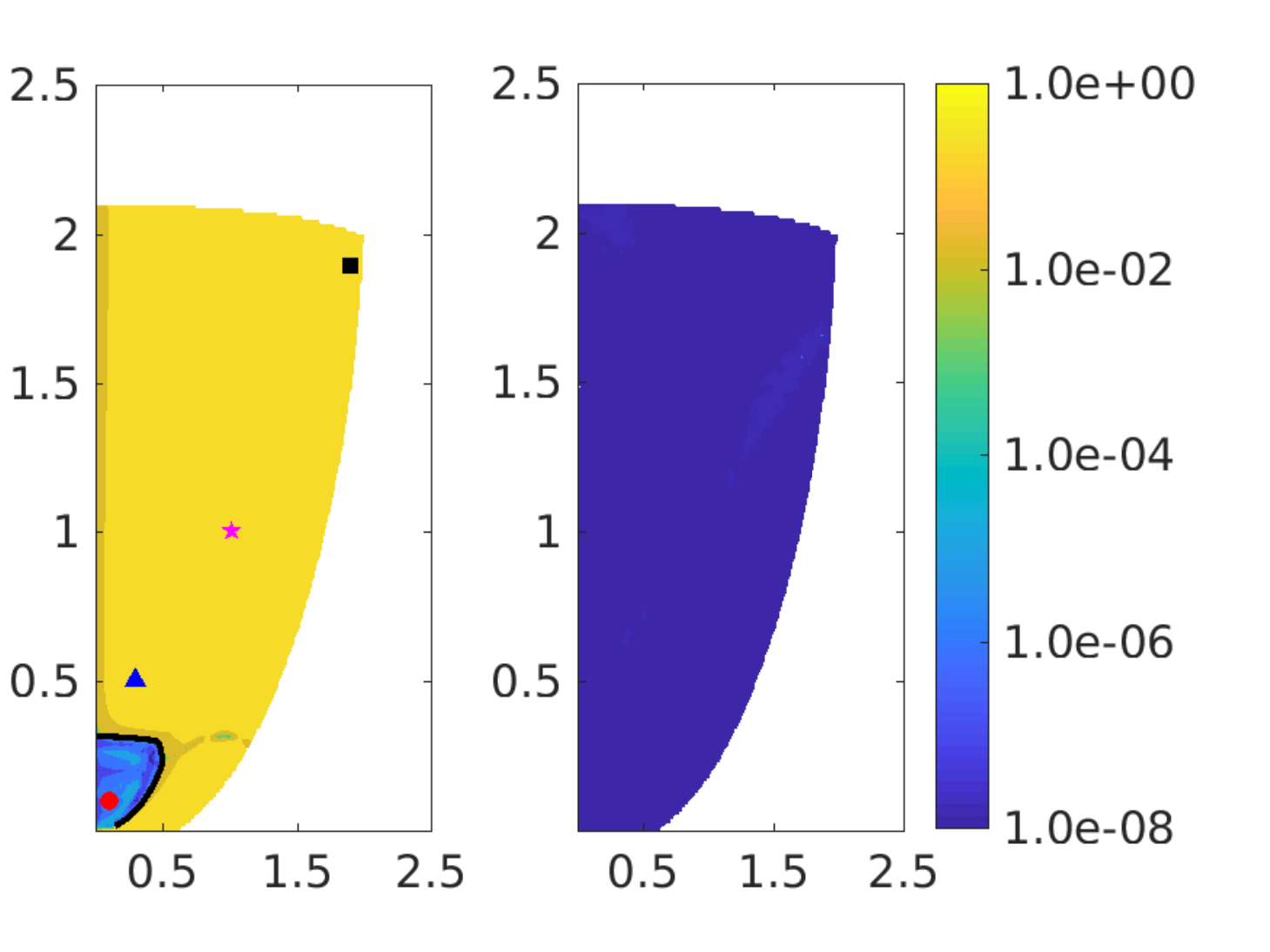}}
%\hspace{0.5cm}
\subfigure{\includegraphics[width=0.45\textwidth]{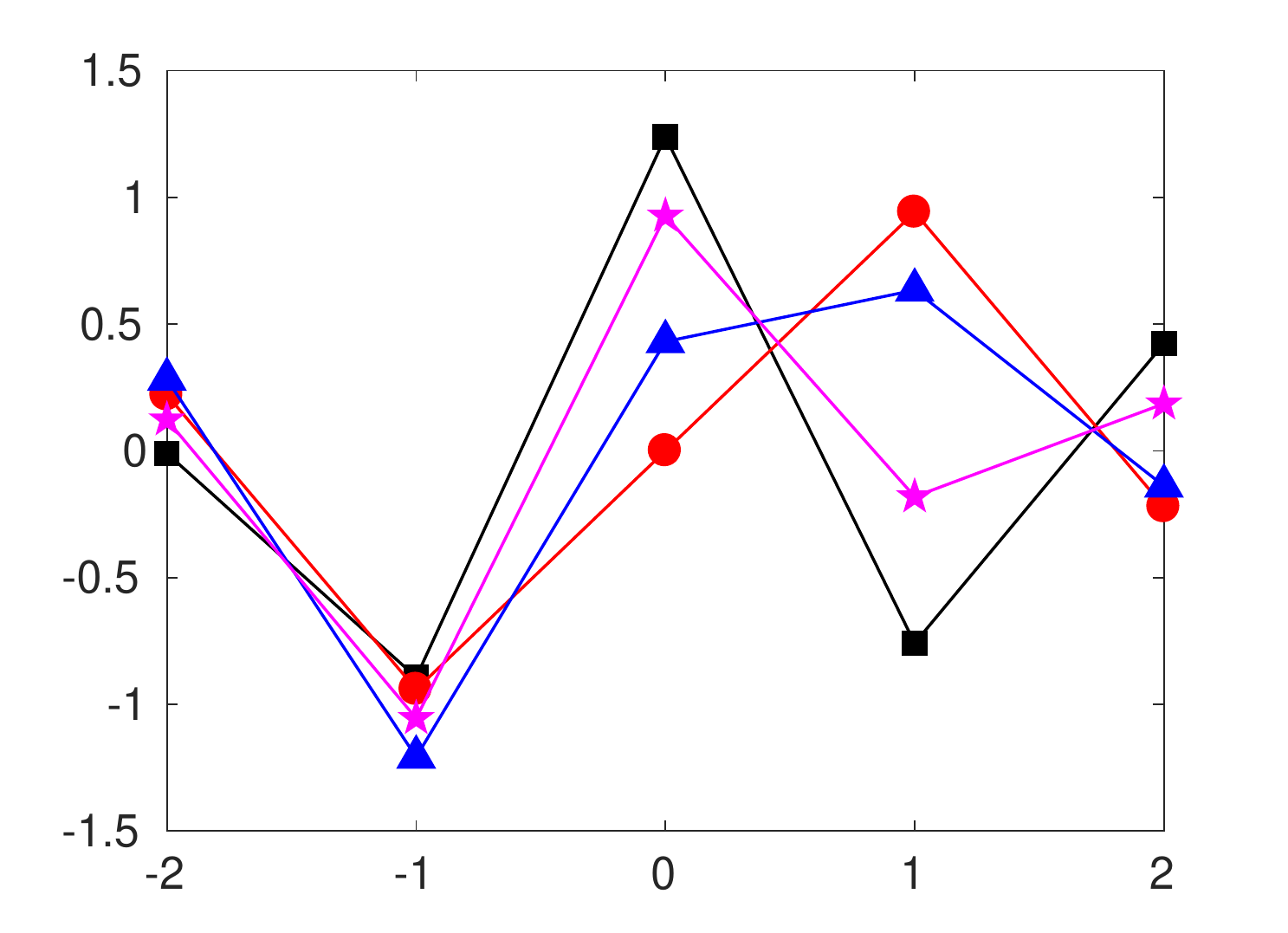}}
\begin{picture}(0,0)
\put(-390,-5){$r_c$}
\put(-310,-5){$r_c$}
\put(-450,80){$r_d$}
\put(-128,-5){$-M:M$}
\put(-40,135){\scriptsize $(c)$}
\put(-370,135){\scriptsize $(a)$}
\put(-290,135){\scriptsize $(b)$}
\put(-215,50){\rotatebox{90}{Elements of $\vo{a}_1$}}
\end{picture}
\caption{Contours of measure $A$ for $M=2$ for (a) first derivative 
(b) second derivative. (c) Optimal coefficients for 
different values of $r_c$ and $r_d$ as marked, with the 
same symbols, in (a).
%for $\beta_1=10$ and $\beta_2=1$.
}
\figlabel{skew}
\end{center}
\end{figure}
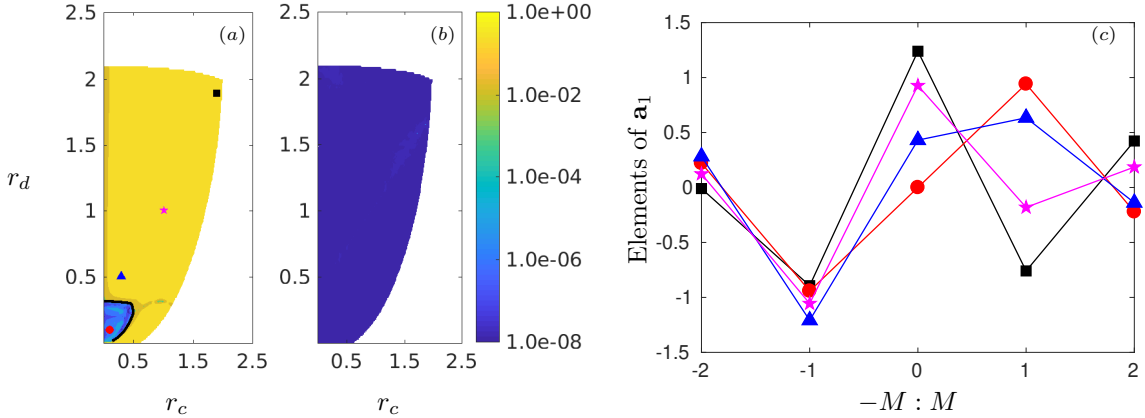

\begin{figure}[h!]
\begin{center}
\subfigure{\includegraphics[width=0.45\textwidth]{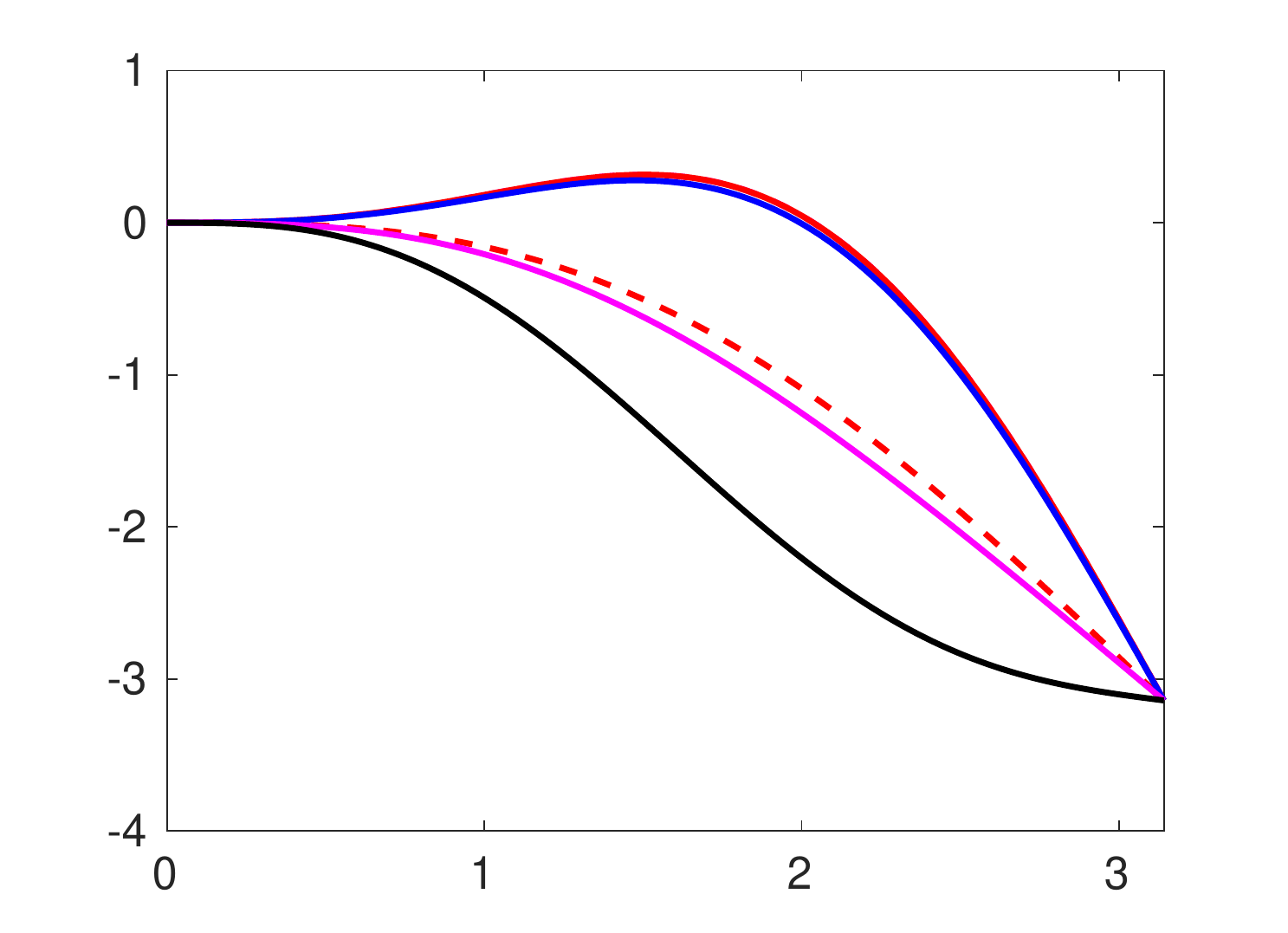}}
%\hspace{0.5cm}
\subfigure{\includegraphics[width=0.45\textwidth]{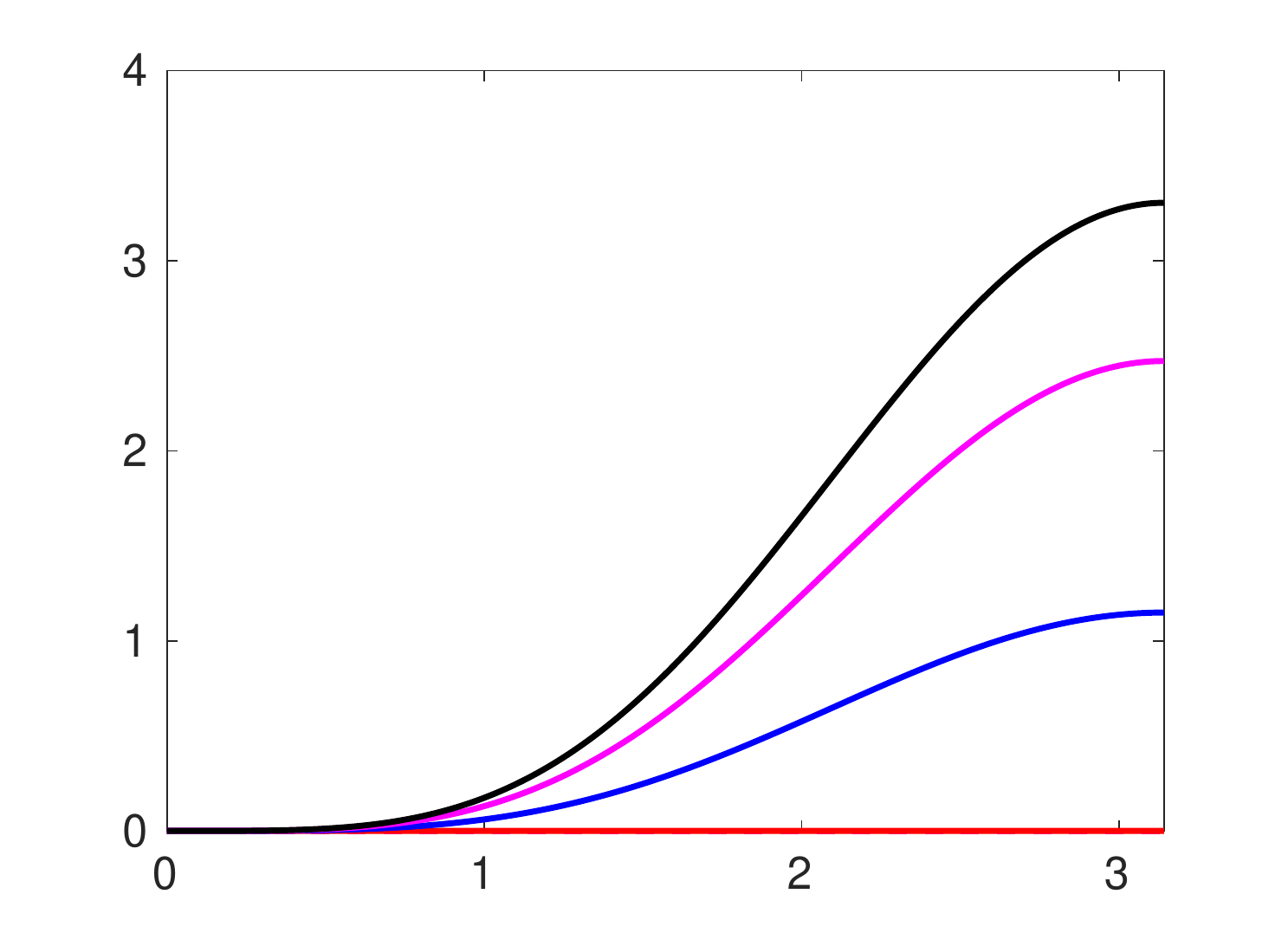}}
\begin{picture}(0,0)
\put(-325,-5){$\eta$}
\put(-110,-5){$\eta$}
\put(-40,132){\scriptsize $(b)$}
\put(-255,132){\scriptsize $(a)$}
\put(-215,65){\rotatebox{90} {$\Re[e(\kdx)]$}}
\put(-440,65){\rotatebox{90} {$\Im[e(\kdx)]$}}
\end{picture}
\caption{\colk{Dispersion error (a) and
dissipation error (b) for the first derivative with 
$M=2$ at different values of $r_c$ and $r_d$ as marked
with the same color in \rfig{skew}(a).}
The dashed 
red line corresponds to standard second order 
scheme with $M=1$ and the dashed magenta line (hidden) corresponds 
to optimized scheme obtained for $M=2$ without stability constraint.
}
\figlabel{errM2}
\end{center}
\end{figure}
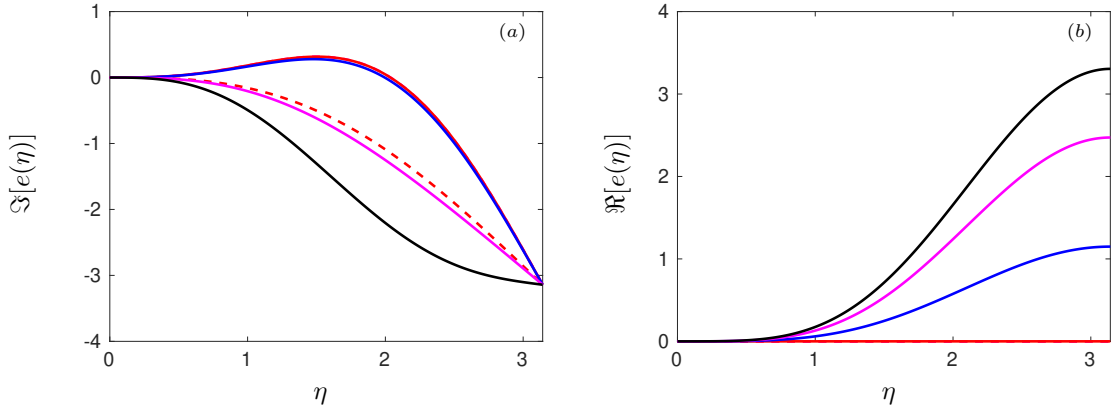

\subsubsection{Higher-order temporal schemes}

In the previous section we illustrated how stability can 
be incorporated into the framework using a first-order forward 
temporal discretization. However, in practice higher order 
schemes are typically used. 
Thus, we generalize 
\eqn{discDyn} to include higher order multi-step temporal schemes, such 
as Adams-Bashforth.
For this we define a vector $\underline{\vo{F}}$ 
as a stack of the solution at $L$ consecutive time levels,  
\begin{equation}
\underline{\vo{F}}^{k+1}=\begin{bmatrix} \vo{F}^{k+1} & \vo{F}^k& \dots 
& \vo{F}^{k-L+1}\end{bmatrix}^T .
\end{equation}
 The evolution equation can then be written as 
\begin{equation}
\underline{\vo{F}}^{k+1}=\left(\vo{A}_t+\sum_d\frac{\Delta t}{(\Delta x)^d}
\beta_d\vo{B}_t\right)\underline{\vo{F}}^k,
\eqnlabel{mse}
\end{equation}
where
\begin{equation}
\vo{A}_t:=\begin{bmatrix} \vo{A}_{1}
 \\
\vo{A}_{2}\end{bmatrix}
~~~\text{and}~~~
\vo{B}_t:=\begin{bmatrix} \vo{b}^T
\otimes \vo{A}_d^{\vo{\Phi}} \\
\vo{0}_{NL \times N(L+1)} \end{bmatrix},
\end{equation}
\begin{equation}
\vo{A}_1:= \vo{e}_{1 \times (L+1)}
\otimes \vo{I}_{N \times N} ~~~\text{and}~~~
\vo{A}_2:=\begin{bmatrix}\vo{I}_{NL \times NL} & \vo{0}_{NL \times N}\end{bmatrix},
\end{equation}
with $\vo{e}=\begin{bmatrix} 1&0&\dots&0\end{bmatrix}$.
The vector $\vo{b}^T$ is the vector of the coefficients of the
temporal scheme that uses $L+1$ time levels.
By examining \eqn{mse} 
it is clear that multistep methods have the same structure as
the example in \rsec{approach2} where, for a given $\dt$,
the unknowns are contained in $\vo{B}_t$.
Thus, the stability condition can then be written as
\begin{equation}
\lambda_{max}\left(
\vo{A}_t+\sum_d\frac{\Delta t}{(\Delta x)^d}
\beta_d\vo{B}_t
\right)\le1
\end{equation}
Following the same procedure as before, 
%Since the spectral radius is bounded by matrix norms,
%the 
this condition, and thus stability,
is guaranteed by bounding the $2$-norm, i.e.
\begin{equation}
\left\|\vo{A}_{t} +  \sum_d \frac{\Delta t}{(\Delta x)^d}\beta_d
\vo{B}_t\right\|_2 \leq 1 .
\eqnlabel{ms-2norm}
\end{equation}
which is the generalization for the inequality constraint 
\eqn{2norm_stab}.

Thus, the unified generalized formulation for 
arbitrary multi-step temporal discretizations
consists of 
an optimization problem that minimizes the spectral error given by
\eqn{domain_ew} subjected to 
the order of accuracy (equality)
constraint \eqn{domainOrder}, and 
the stability (inequality) constraint 
\eqn{ms-2norm}.

The two approaches discussed in \rsecs{approach1}{approach2},
can then be applied to the generalized multi-step schemes 
as well. That is, we
can compute maximum $\Delta t$ for given $\vo{A}_d^{\vo{\Phi}}$
or we can compute $\vo{A}_d^{\vo{\Phi}}$ by fixing $\Delta t$. 
We do note, however, that 
bounding the spectral radius by the $2$-$norm$ can be 
unnecessarily restrictive for some multistep schemes especially
for long temporal stencils which leads to 
long vectors resulting from the stacking 
of increasingly large number of time levels.

%%%%%%%%%%%%%%%%%%%%%%%%%%%%%%%%%%%%%%
%\section{Pade differentiation}
%\subsection{Order of accuracy}
%\subsection{Spectral accuracy}
%\subsubsection{Even derivatives}
%\subsubsection{Odd derivatives}
%\subsection{Stability}

%%%%%%%%%%%%%%%%%%%%%%%%%%%%%%%%%%%%%%%%%%
\section{Numerical results} \label{sec:simulations}

In order to test the theoretical results from previous 
sections we conducted several test on model PDEs of increasing 
complexity.  The focus would be in a comparison between the schemes developed here
and standard schemes.
In particular, we will compare our optimized schemes 
against 
\begin{enumerate}[(a)]
\item standard scheme of the same order 
\item standard scheme with the same stencil size 
\end{enumerate}
While (a) allows us to assess how the additional degrees 
of freedom are used to increase spectral accuracy and/or
maintain stability, 
(b) provides comparison between two schemes with 
the same computational cost in computing spatial derivatives 
since both schemes use the same stencil size. 
As pointed out later on, however, optimal schemes of lower 
order may indeed provide a computational advantage 
when they are coupled with a temporal scheme of matching order
to solve a PDE.

We begin our analysis with the diffusion equation and the linear
advection-diffusion equation for which
exact analytical solutions are known and the error in 
numerical solutions can be evaluated accurately. We then turn to the non-linear
advection-diffusion (Burgers)
equation which is a widely used proxy to study important 
features of fluid flow motion 
governed by the Navier-Stokes equations. \colk{We will 
conclude the numerical section by a brief analysis of the wave equation
as well as a discussion of the effect on dispersion relations.}

For short, we will refer to the standard and optimized 
finite difference schemes as SFD and OFD, respectively in what 
follows.

\subsection{Diffusion equation}
Consider the equation: 
\begin{equation}
{\partial u\over \partial t} = \alpha {\partial^2 u \over \partial x^2},
\eqnlabel{diffeqn}
\end{equation}
where $\alpha$ is the diffusivity.
Since this equation 
is linear, different Fourier modes do not interact and 
wavenumbers present in the solution are only due to them 
being present in the initial conditions.
The dissipative action of the second derivative causes the 
decay of amplitude of all modes with time. This decay becomes more 
prominent as the wavenumber increases.
Equation \eqn{diffeqn} is solved 
in a periodic domain of length $L=2\pi$. The initial condition is a 
superimposition of sinusoidal waves,
\begin{equation}
u(x,0)=\sum_k A(k) \sin(k x + \phi_k),
\eqnlabel{diffinit}
\end{equation}
where $k$ denotes the wavenumber, 
$\phi_{k}$ is
a random phase angle corresponding to each wavenumber, and
$A(k)$ is the amplitude of each mode taken here to be represented as
a power law of the form $A(k)=A(1)k^\sigma$. The value of the 
exponent was chosen to be $\sigma={-1/6}$ which, by being small,
corresponds to a shallow spectrum representative of a 
solution with a wide range of energetic modes. The reason for this 
is to critically assess the ability of schemes to represent accurately 
a wide range of scales.

The analytical 
solution of \eqn{diffeqn} is known:
\begin{equation}
u_a(x,t)=\sum e^{-\alpha k ^2 t}A(k) \sin(k x + \phi_k)
\eqnlabel{diffsola}
\end{equation}

For the semi-discrete analysis, we discretize \eqn{diffeqn} using 
an optimized second order scheme (OFD2) with $\ns=4$ in space
whose coefficients can be 
found in \rtabs{second_optim_ad}{first_optim_ad}.
 The numerical results so obtained are compared with standard second order 
 (SFD2) and standard eighth order (SFD8) in space.
For the fully discrete system, we match the order of 
accuracy of time and space discretization.
Time and space step sizes ($\dt$, $\dx$) are related through 
a diffusive CFL condition
($r_{d}=\alpha \Delta t/ \Delta x^2$).
Thus, we use a forward first order discretization in time
for OFD2 and SFD2,
and a fourth-order five-stage Runge-Kutta scheme for SFD8.

%When the decay of a wave is less than the analytical decay for that wave 
%then the dissipation is less or the spectral energy is more 
%than its actual value. This difference between the analytical spectral energy and 
%the spectral energy of the numerically evolved system gives a measure of 
%the dissipation error. 
To assess the error across scales, we compute the relative difference
between the energy at individual 
Fourier modes (at a given wavenumber) of the numerical 
and analytical solutions, that is 
$|\hat{u}(\eta)^2-\hat{u}_a(\eta)^2|/|\hat{u}_a(\eta)|^2$
where as before $\eta=k\dx$.
In \Fig{spec_err} we 
show this error for the semi-discrete system (solid lines)
at a normalized time of 
$t \alpha k_0^2 \approx0.002$ where $k_0=1$ is the 
lowest wavenumber in the simulation.
This small normalized time was selected to ensure that
the solution has evolved enough to 
present measurable errors while, at the same time,
energy in high wavenumbers is not completely dissipated. 
In terms of the highest wavenumber in the simulation
he normalized time is $t\alpha k_{max}^2\approx 20$.
The error for SFD2 is larger than OFD2
and grows drastically with increasing wavenumber, illustrating 
the inability of low-order standard schemes to capture rapid 
spatial fluctuations. However, this is clearly not the case of 
optimized second-order schemes.
For SFD8, the error is very small at low wavenumbers but also increases
significantly at high wavenumbers. In fact, we see that 
the error spans
more than fifteen orders of magnitude showing a dramatic disparity in 
resolution capabilities for multiscale problems. 

The optimized scheme OFD2, on the other hand, shows a more uniform 
error distribution across the entire wavenumber
 space shown. In fact, 
 OFD2 presents much smaller errors than even the
 eighth order SFD8 at high wavenumbers. 
By fixing the formal order the scheme to two, 
the additional degrees of freedom provided are used to increase 
the resolution capabilities of increasingly large regions of 
wavenumber space.
The fact that OFD2 is more {\it spectrally flat} than the 
other schemes stems for 
the choice we have made for the function $\gamma(\eta)$ 
in \eqn{sp_error}
as unity in the region of interest $\eta\in [0,2.5]$ and
zero otherwise.
This assigns equal weights
to all wavenumbers in that interval and, thus, leads to spectrally flat schemes.
Such performance is highly desired when one wants to study 
multiscale physical processes like turbulence. 
As we have shown in \rsec{gamma}, however, the formulation allows for 
non-uniform weights which can be used to obtain better resolution in one or 
more arbitrary regions of the wavenumber space.

In the figure we also include the fully-discrete system integrated 
with a very small time step ($r_{d}=0.0005$) as a dashed line. 
This is seen to be very similar to the semi-dicrete case. 
However, if we increase the time-step, that is, increase $r_{d}$, then 
the error increases and the time discretization errors may dominate 
the solution. Note that this is true both for optimized and standard schemes.

\begin{figure}[h]
 \centering
 \subfigure{\includegraphics[width=0.49\textwidth]{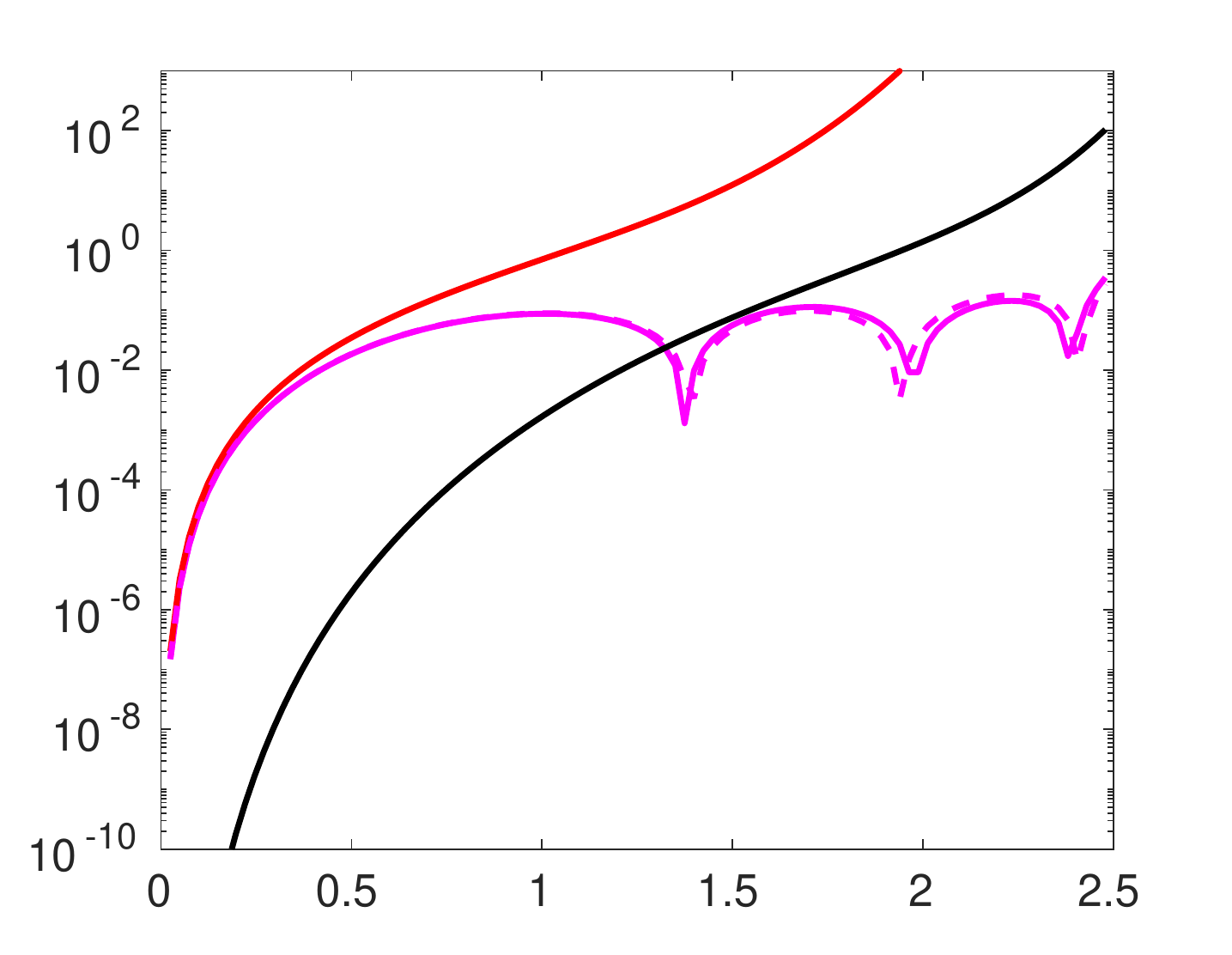}}
 \begin{picture}(0,0)
       \put(-115,-4){$\kdx$}
       \put(-240,64){\rotatebox{90} {$\frac{|\hat{u}(\kdx)^2-\hat{u}_a(\kdx)^2|}{(\hat{u}_a(\kdx))^2}$}}
 \end{picture}
 \caption{
 \figlabel{spec_err}
Normalized error in spectral energy for diffusion equation. 
Solid lines correspond to semi-discrete integration 
with different numerical schemes: 
SFD2 (red), SFD8 (black), OFD2 (solid magenta). 
Dashed magenta line is OFD2 for the fully-discrete system 
using forward first order in time with $r_d=0.0005$.
 }
\end{figure}
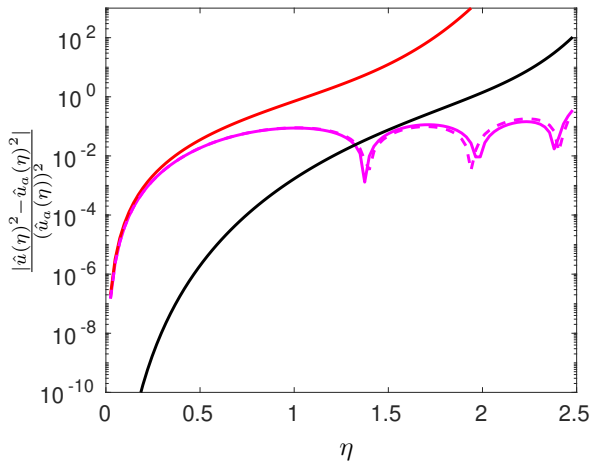

\subsection{Linear advection-diffusion equation}

We now consider, next in complexity, the linear advection-diffusion equation which has both the 
first and second derivatives in space
which, as shown in previous sections, typically present different types of errors:
\begin{equation}
{\partial u\over \partial t} +c{\partial u \over \partial x} = \alpha {\partial^2 u \over \partial x^2},
\eqnlabel{adleqn}
\end{equation}
In this case, Fourier modes are convected at the velocity $c$ and dissipated 
at a rate determined by the diffusivity $\alpha$.
Here we use the same initial condition \eqn{diffinit} 
as in the previous section. 
The analytical solution for \eqn{adleqn} is given by 
\begin{equation}
u_a(x,t)=\sum e^{-\alpha k^2 t}A(k) \sin(k (x - c t)+ \phi_k)
\eqnlabel{adlsola}
\end{equation}

Because of the presence of both first and second derivatives 
we expect the numerical solution to be affected by 
both dispersion and dissipation errors.
Dispersion error typically due to the first derivative, distorts 
phase relations between different waves and tend to create 
distorted shapes. Dissipation errors, as explained before, 
affect the amplitude of different waves.
%We compare these errors for the optimized schemes 
%with the errors for standard schemes. 
In \Fig{spec-err-adl} we show again the relative 
error of spectral energy which represents 
dissipation error across 
the wavenumber space. The trend is similar to what we observed 
for the diffusion equation. 

\begin{figure}[ht]
  \centering
  \subfigure{\includegraphics[width=0.49\textwidth]{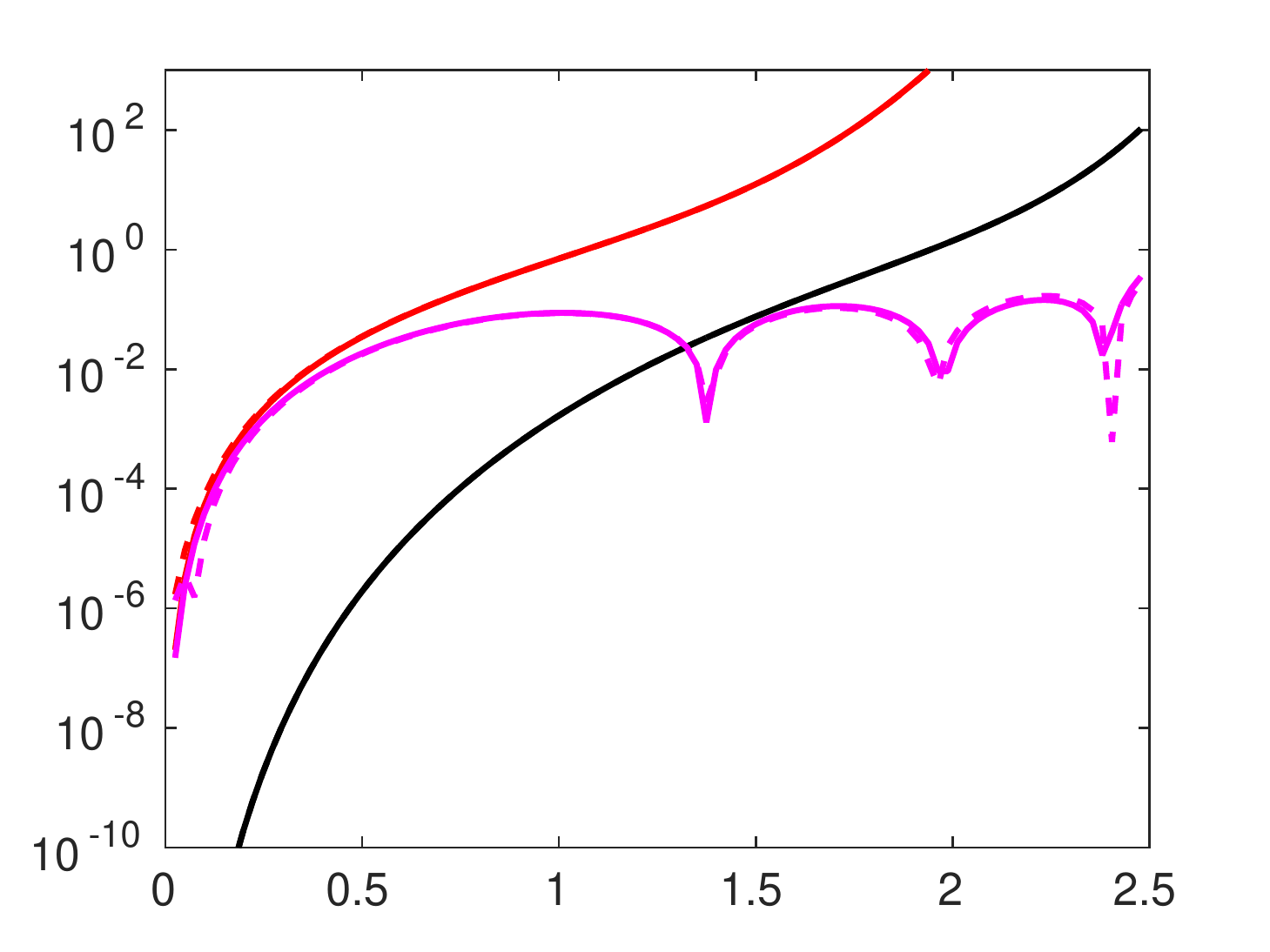}}
  \begin{picture}(0,0)
       \put(-121,-4){$\kdx$}
       \put(-244,58){\rotatebox{90} {$\frac{|\hat{u}(\kdx)^2-\hat{u}_a(\kdx)^2|}{(\hat{u}_a(\kdx))^2}$}}
 \end{picture}
   \caption{
  \figlabel{spec-err-adl}
Normalized error in spectral energy for advection-diffusion equation. 
Different lines correspond to different numerical schemes: 
standard second order in space (red), standard eighth order in space (black), optimized ($\ns$-$4$) second order in space (solid magenta) and 
optimized ($\ns$-$4$) second order in space and forward first order in time (dashed magenta)}

\end{figure}
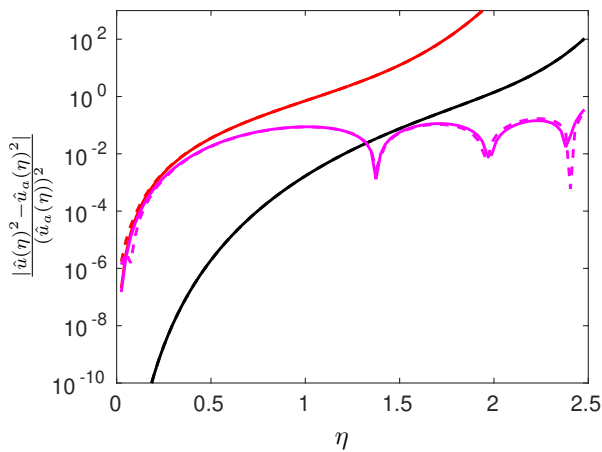

It is also possible to quantify dispersion errors 
by computing the effective propagation speed $c^*$ 
of Fourier modes at wavenumber $k$ which can readily 
shown to be given by $c^*(k)=\arg[\hat{u}_t(k)/\hat{u}_0(k)]/k t$, 
where $\hat{u}_0(k)$ and
$\hat{u}_t(k)$ are the Fourier coefficients at the beginning 
of the simulations and at a time $t$, respectively.
%The such that the 
%normalized time at that time based on the convective velocity $c=1$ is
%$t^*=0.016$}. 
The ratio of the numerical speed $c^*$ to
the actual propagation speed $c$ is a measure of phase (dispersive)
errors.
This ratio is 
unity if there is no phase error which
implies that the numerical solution travels with the same speed 
 as the actual solution.
In \fig{c-err}(a) we show results for a normalized time of 
$c k_0 t=0.1$, or in terms of the highest wavnumber
in the simulation $ck_{max}t\approx 10$.
We see that, for SDF2,
 the numerical wave speed becomes much smaller compared 
 to the actual speed $c$ as the wavenumber increases, resulting 
 in large phase errors. For OFD2 the numerical speed 
 remains much closer to the actual speed, though 
 with some oscillations.
OFD2 is also seen to be better than even 
the eighth order scheme SFD8 for which, as the wavenumber increases,
the numerical speed decreases monotonically leading to very large phase
error.
The normalization of \fig{c-err}(b) highlights this trend in the phase
error. In particular, it shows that, because of the particular choice
of $\gamma(\eta)$, the numerical scheme presents a more spectrally 
flat response.  
%We see that phase error for the optimized scheme is less than the standard 
%second order scheme and does not vary much across the range of wavenumbers. 
%However, 
%for the high-order standard scheme, the phase error varies ten orders of
%magnitude. 
%It is also larger than the optimized scheme at higher wave-numbers. The error
%for 
%the semi-discrete and fully discrete system are comparable for small timestep.

\begin{figure}%[ht]
 \centering
 \subfigure{\includegraphics[width=0.49\textwidth]{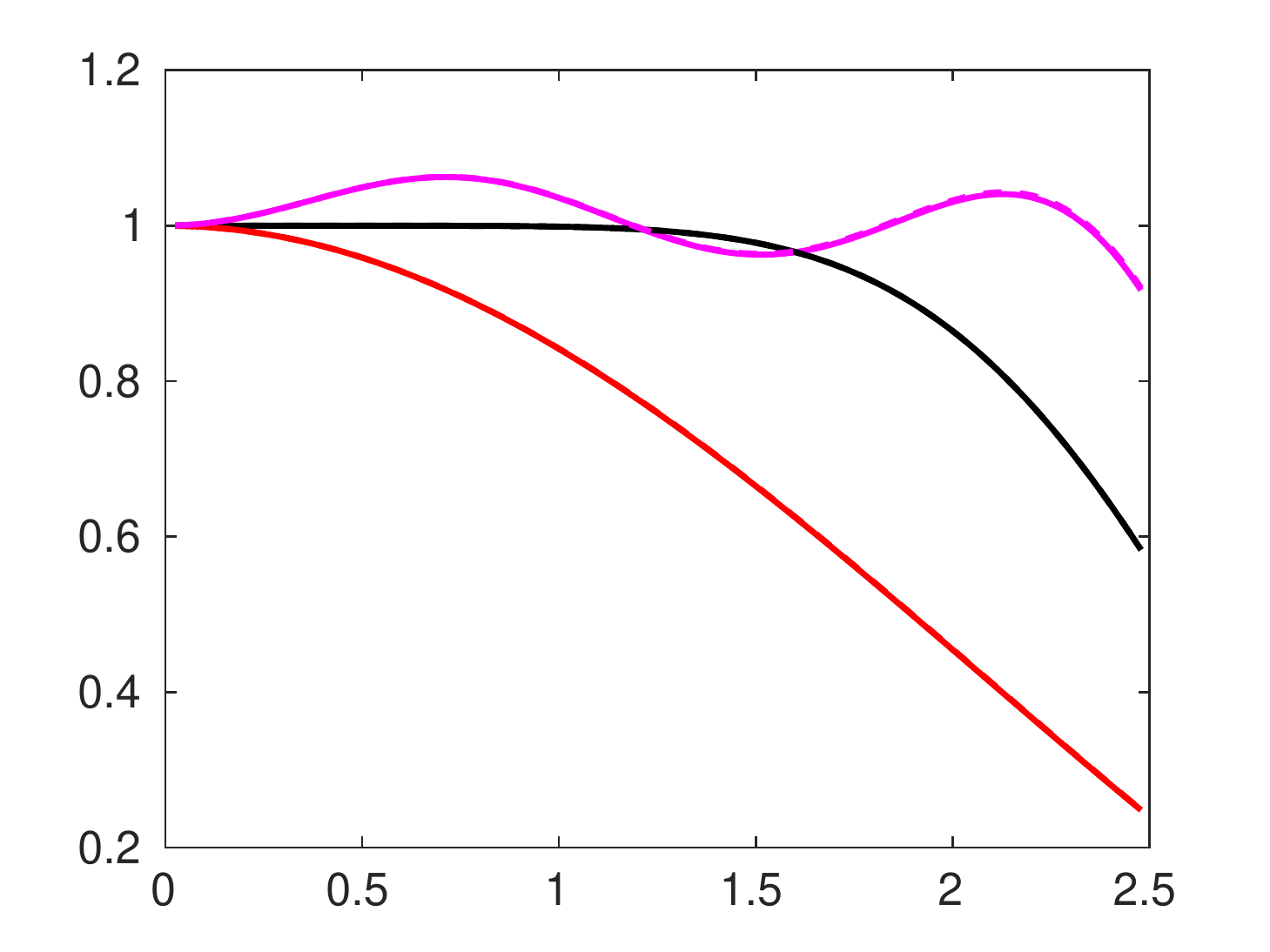}}
 \subfigure{\includegraphics[width=0.49\textwidth]{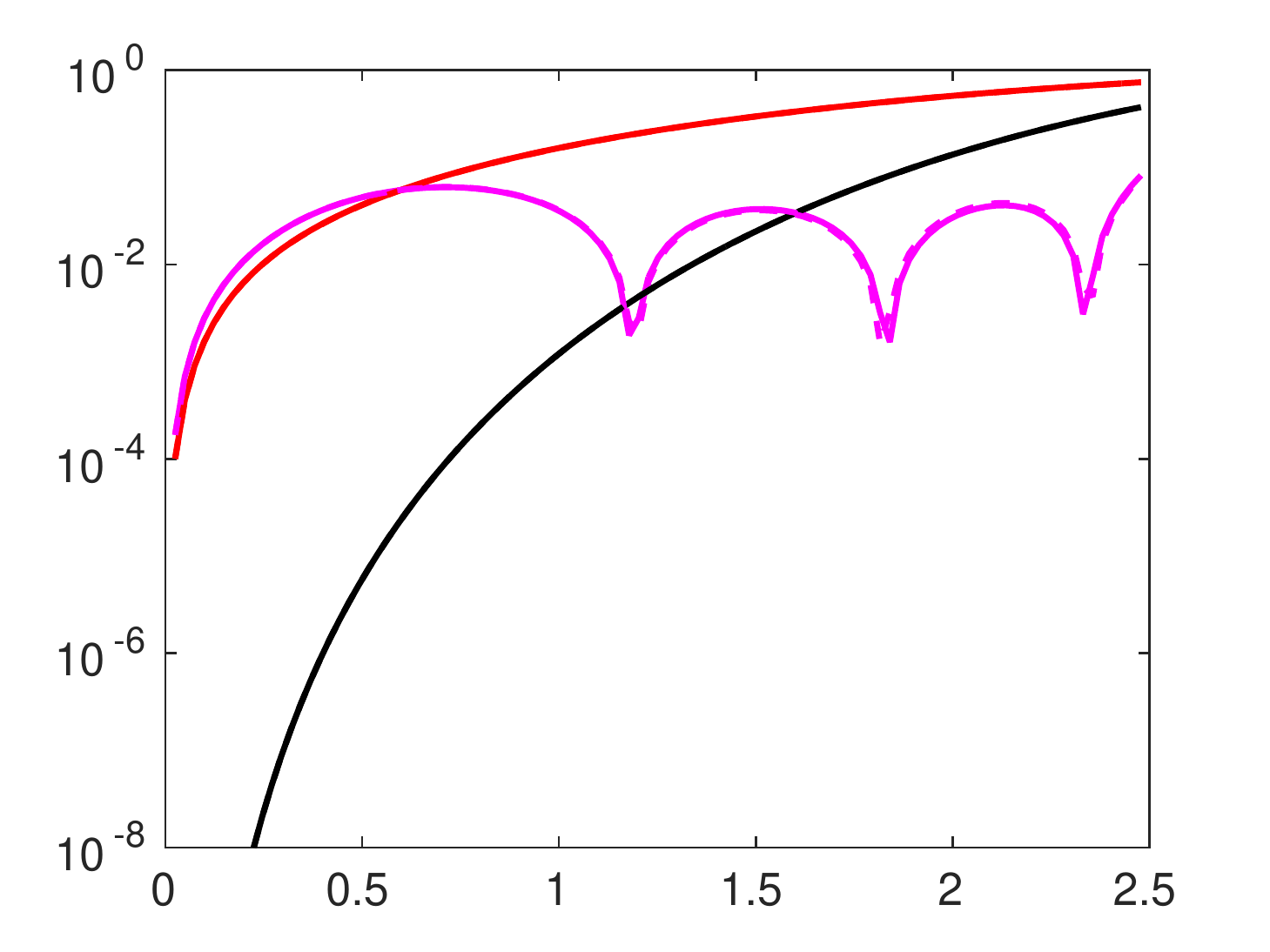}} 
  \begin{picture}(0,0)
	\put(-195,150){\scriptsize{(b)}}
	\put(-430,150){\scriptsize{(a)}}
       \put(-121,-4){$\kdx$}
       \put(-244,64){\rotatebox{90} {$|\frac{c^*}{c}-1|$}}
       \put(-355,-4){$\kdx$}
       \put(-470,80){\rotatebox{90} {$\frac{c^*}{c}$}}
 \end{picture}
\caption{
 \figlabel{c-err}
(a) Ratio of numerical to actual speed. 
(b) Normalized phase error for
advection-diffusion equation. 
Different lines correspond to different numerical schemes: 
SFD2 (red), SFD8 (black), semi-discrete with OFD2 in space
(solid magenta) and 
OFD2 in space and forward first order in time (dashed magenta).}
\end{figure}
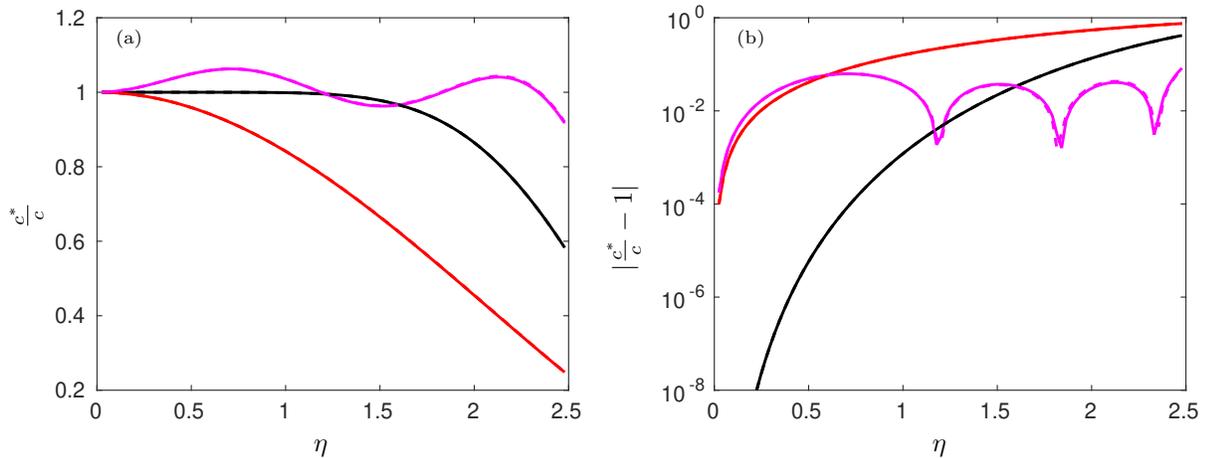

%Instead of looking at the broadband spectrum, we can also focus on
%individual wavenumbers 
%and compare the errors for different schemes for each of them. 
A complementary assessment of the numerical performance of 
these schemes can be obtained by comparing the time evolution 
of different Fourier modes.
%We know that no new wavenumbers are engendered and the original wave
%convects 
%and dissipates energy with time. 
In the fully discrete system, order of accuracy in time is chosen 
to be consistent with the spatial discretization. With a fixed
diffusive CFL, this implies 
%This analysis requires fully discrete system where 
that the second and eight order schemes in space need a first
and fourth order 
temporal discretization, respectively. In the examples
below we use, thus, forward Euler (first order) and RK4
(fourth order) for time integration.

For illustration purposes we consider waves at the extreme
 ends ($k={10,100}$) of the wavenumber interval over which the 
schemes have been optimized, 
and an intermediate wave-number ($k=50$). These correspond
to $\kdx =\{0.245,1.227,2.454\}$.
In \fig{ener_3k} we show the evolution
of energy for these three wavenumbers.
The decay in energy occurs at a faster rate as the wavenumber increases.
At lower $\kdx$, the energy decay is well captured by all three
schemes. 
As we increase $\kdx$, SFD2 over-predicts 
the energy content and fails to capture the actual dissipation.
The disparity between SFD2 and OFD2 increases with increasing 
wavenumber with the latter remaining close to the analytical solution. 
For very high wavenumbers $\kdx=2.454$, 
the second-order optimized scheme 
is in fact visibly closer to the analytical solution 
than a standard eighth-order scheme.
Clearly by using the additional information provided by neighboring
points to increase spectral accuracy instead of formal 
order of accuracy leads to better resolved physics, especially 
at high wavenumbers where dissipation is strongest.
%We see that the optimized scheme better explains the dissipation at the
%extreme
%wave-numbers and captures the physics more accurately, especially at 
%higher wave-numbers.
%The error in each of the waves was computed 
%and plotted in Fig.\ref{err_3k}. We can clearly see the transition form 
%comparable error at low $\kappa$ to to high error at high $\kappa$ for the 
%standard $8^{th}$ order scheme. The claim that for higher $\kappa$ values 
%optimized schemes give much lesser error than the standard schemes is 
%profoundly highlighted in Fig.\ref{err_3k}(c).

%\colr{[DD: this paragraph 
%needs a little bit more explanation. They derived
%schemes for the same PDE? Or it's just spatial derivatives? 
%If so, is it here the best place to present it?]}\colk{ -I use 
%the coefficients given in these papers for spatial derivatives 
%(4th order M=4) 
%and evolved the solution using semi-discrete formulation. }

We have also compared schemes derived with our framework
with those in \cite{BB2004,TW1992}. For example, 
we have computed fourth-order optimized schemes for $\ns=4$ with the same
conditions as those in those references and found essentially
the same coefficients. Thus, our general framework can reproduce 
other particular results in the literature.

%the (dispersion) error obtained for fourth order optimized schemes, with 
%used in optimized schemes derived by 
%were also compared. The 
%results were in close agreement, 
%implying that the parameters in our framework can be changed 
%to obtain high order optimized schemes in the literature. 

%\begin{figure}%[ht]
%  \centering
%  \subfigure{\includegraphics[width=0.31\textwidth]{plots/burgers/linear/ener_errn_time_k10}}
 % \subfigure{\includegraphics[width=0.31\textwidth]{plots/burgers/linear/ener_errn_time_k40}}
  %\subfigure{\includegraphics[width=0.31\textwidth]{plots/burgers/linear/ener_errn_time_k80}}
 % \caption{abc}
 % \label{fig:layout}
%\end{figure}
\begin{figure}%[ht]
 \centering
 \subfigure{\includegraphics[width=0.49\textwidth]{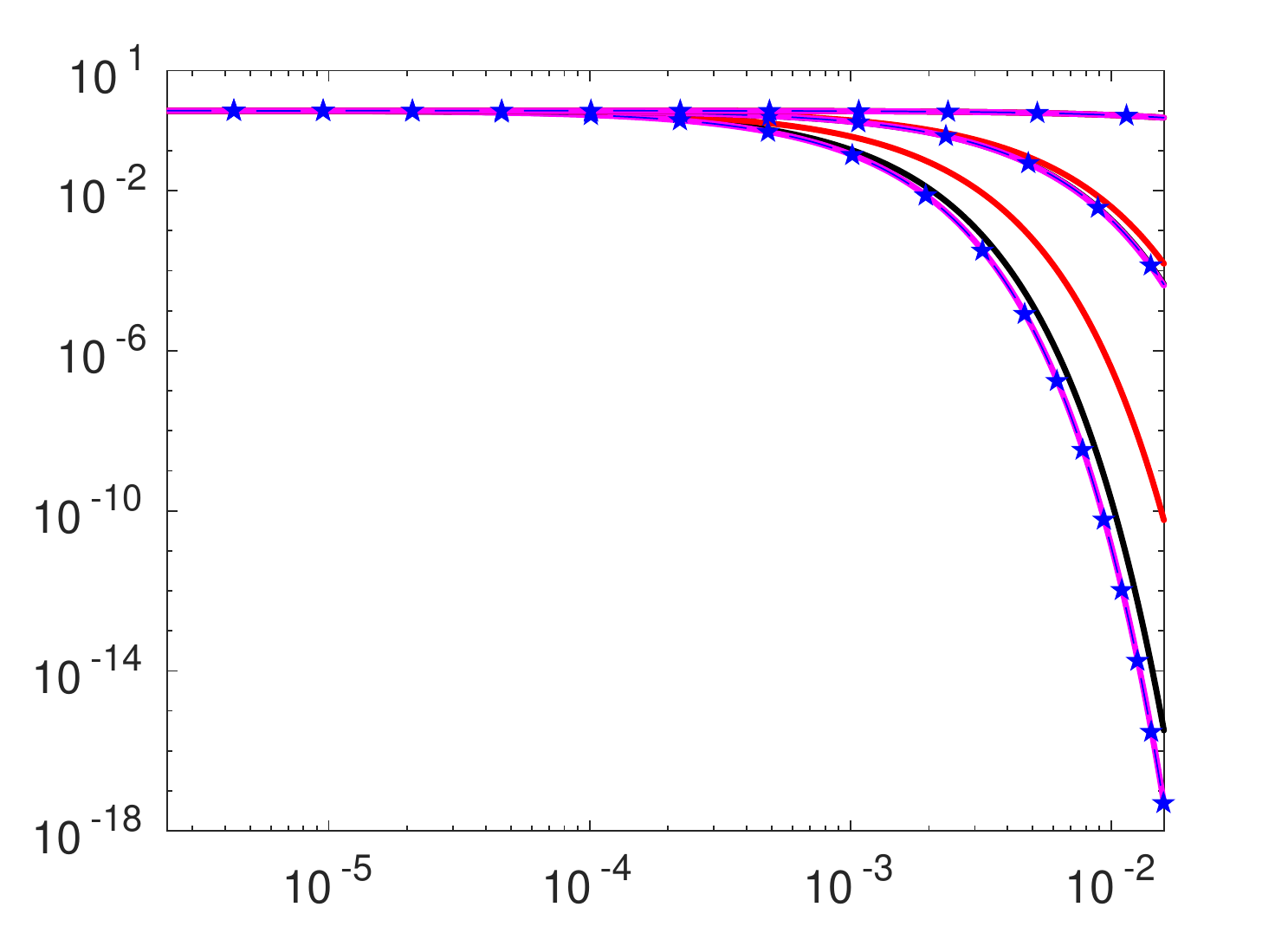}}
\begin{picture}(0,0)
        \put(-121,-6){$t^*$}
	\put(-250,72){\rotatebox{90}
		{$\frac{\langle u(x,t)^2\rangle}{\langle u(x,0)^2\rangle}$}}
	\put(-44,148){\scriptsize 0.245}
	\put(-41,134){\scriptsize 1.227}
	\put(-67,95){\scriptsize 2.454}
\end{picture}
 \caption{
 \figlabel{ener_3k}
Evolution of space-averaged energy normalized by 
the initial space-averaged energy with normalized time $t^*=tc/L$ 
for three different 
wave-numbers ($\kdx =0.245,1.227,2.454$).
Different line styles correspond to different numerical schemes: 
SFD8 with RK4 in time (black), 
OFD2 with forward Euler in time (solid magenta),
SFD2 with forward Euler in time (red).
The blue dashed-star line is the analytical solution.}
\end{figure}
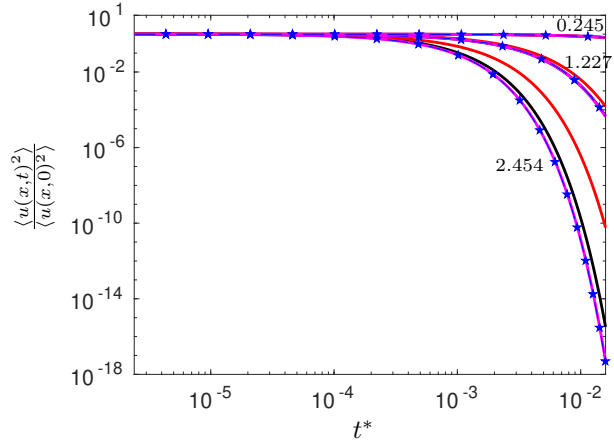

\subsection{Non-linear advection-diffusion equation}

Of fundamental and practical interest is the non-linear
advection-diffusion equation 
as it resembles the one-dimensional version of the
Navier-Stokes equation that governs the motion of fluid flows,
\begin{equation}
{\partial u\over \partial t} + u{\partial u \over \partial x} =
\alpha {\partial^2 u \over \partial x^2}.
\eqnlabel{adnleq}
\end{equation}
Here $u(x,t)$ is the velocity and $\alpha$ is viscosity.
The non-linear term causes interaction between Fourier modes
which redistributes energy 
among the different scales (wavenumbers) in the solution
and produces new scales of motion. 
Because of the absence of a constant input of energy, 
the amplitude of different modes decay with time 
due to the dissipative action which becomes more effective 
at smaller scales.
%With forcing term \textit{f}, we add energy to the small wave-numbers.
%As energy increases, so does dissipation and eventually, there is an 
%equilibrium between production and dissipation and steady-state is achieved.

%In order to compare the accuracy of different solution we use a highly resolved
%numerical solution as the reference solution. The error is then 
%computed with respect to this reference solution.
%which was found to be well resolved at $N=4096$. 
% The initial
%condition is again given by equation \eqn{diffinit}
%For example, grid convergence studies show 
%%that the decay of the space-averaged kinetic energy as a function of time 
%becomes independent of resolution at $N=256$. At $N=4096$ all quantities 
%of interests here are grid converged for the initial conditions used
%for the simulations and is thus used as the reference solution.
%against which relative errors are computed.

%\Fig{ener-adl}(a) shows the steady state energy spectrum
\eqn{adnleq} is solved in a periodic domain of length $L=2\pi$, 
subject to initial conditions given by \eqn{diffinit}. The 
value of the exponent in \eqn{diffinit} was chosen to be $\sigma=-5/6$. 
This exponent, which corresponds to an energy spectrum decaying as 
$k^{-5/3}$ consistent with fully developed turbulence,
ensures that the spectrum is shallow enough to ensure 
high energy content at high wavenumbers while remaining stable.

Although \eqn{adnleq} is non-linear, we can apply 
the Cole-Hopf transformation and find an analytical solution 
to the problem \cite{Hopf, Cole}.
Define a transformation variable $\phi$, such that 
\begin{equation}
u=-2\alpha\frac{1}{\phi}\frac{\partial{\phi}}{\partial{x}},
\eqnlabel{transf}
\end{equation}
Then \eqn{adnleq} reduces to a simple diffusion equation in $\phi$,
%\begin{equation}
%\frac{\partial{\phi}}{\partial{t}}=
%\alpha\frac{\partial^2{\phi}}{\partial{x^2}}.
%\eqnlabel{phieq}
%\end{equation} 
which can be readily solved analytically. The result in terms 
of the primitive variable $u(x,t)$ is  
\begin{equation}
u(x,t)=\frac{\int_{-\infty}^{\infty}
\frac{(x-y)\II}{t} dy}
{\int_{-\infty}^{\infty}
\II dy}.
\eqnlabel{adnl-sol}
\end{equation}
where $\II = e^{-\frac{(x-y)^2}{4 \alpha t}}\phi(y,0).$
and the initial condition for $\phi$, 
is computed from the initial condition for $u$ as
\begin{equation}
\phi(x,0)=e^{\left( 
-\int_0^x \frac{u(y,0)}{2\alpha}dy
\right)}.
\end{equation}
%Once $\phi(x,t)$ is known, from the relation in \eqn{transf}, 
%the actual variable $u(x,t)$ is computed and used as analytical 
%solution of \eqn{adnleq}
While \eqn{adnl-sol} is the exact solution, the integrals involved 
are computed numerically with standard integral techniques 
which were tested for grid convergence.
This is compared with the numerical solution of \eqn{adnleq}
using, as before, SFD2, SFD8, and OFD2. 
The fully discrete system is 
formed with compatible temporal scheme as described in previous section.
We performed grid convergence studies and found that 
e.g.\ the space-averaged 
kinetic energy evolution in time becomes
independent of resolution at $N=256$. 
This is the resolution used for the comparisons that follow.

In \rfig{ener-adnl} we show the energy spectrum obtained 
for these three schemes along with the analytical solution 
\eqn{adnl-sol}, at $t/t_0\approx 0.464$, 
where $t_0=K_0/\epsilon_0$ is a characteristic time scale 
defined by the initial 
energy ($K_0\equiv  \langle u_0^2\rangle/2$, where 
angular brackets denote space averages and a subscript $0$ denotes
initial conditions) and the energy 
dissipation rate 
$(\epsilon_0\equiv \alpha \langle (\partial u_0/\partial x)^2\rangle)$.
We observe that OFD2 and SFD8 agree closely with each 
other and with the analytical solution throughout the 
range of $\eta$. Results for SFD2, however, exhibit clear departures
especially at high wavenumbers.
The ability of OFD2 to capture high wavenumbers accurately 
is expected since, as shown above, this scheme presents a 
more spectrally flat response. 
In fact, this optimized scheme presents better resolution 
than SFD8 at very high waveumbers though the converse is
true at low waveumbers.

\begin{figure}%[ht]
\centering
\includegraphics[width=0.49\textwidth]{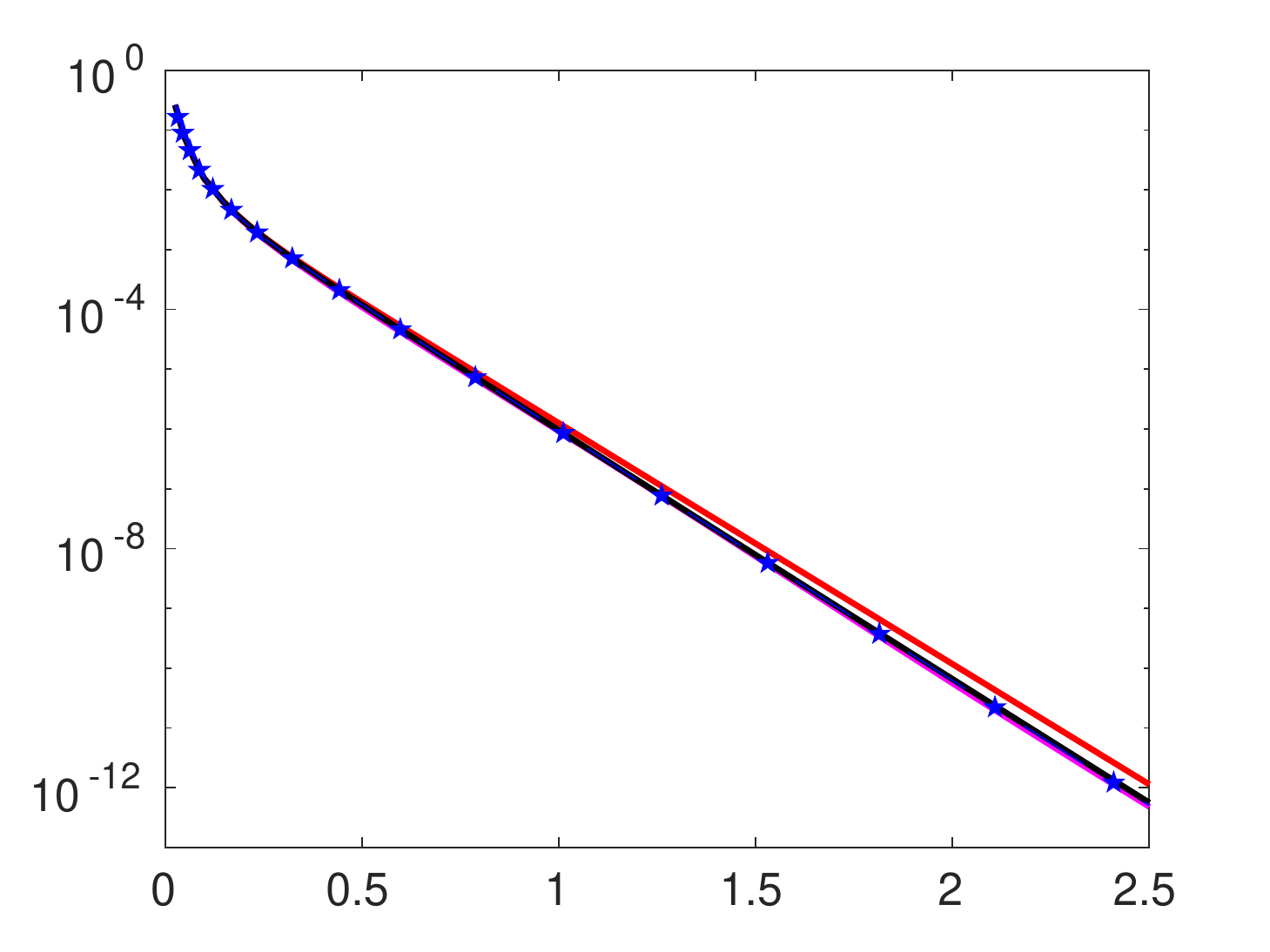}
\begin{picture}(0,0)
	\put(-120,-5){$\kdx$}
        \put(-240,85){\rotatebox{90} {$\hat{u}(\kdx)^2$}}
\end{picture}
\caption{\figlabel{ener-adnl} Energy spectrum 
for the non-linear advection-diffusion equation at 
$t^*=t/t_0\approx0.464$.
Different lines correspond to different numerical schemes: 
SFD8 with RK4 in time (black), 
OFD2 with forward first-order in time (magenta), and
SFD2 with forward first-order in time (red).
The blue dashed-star line is the analytical solution.}
\end{figure}
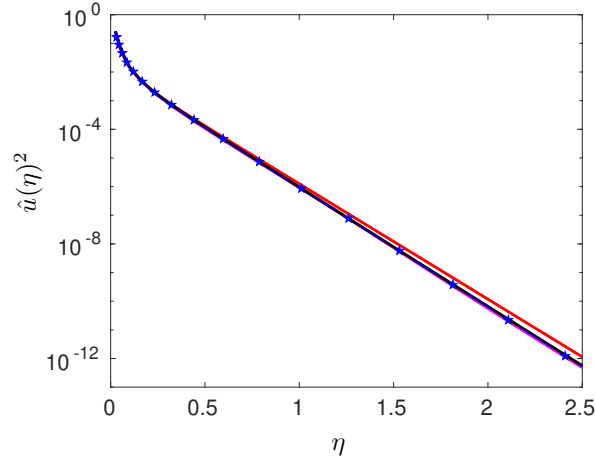

Because of the better small-scale resolution of OFD2,
one would expect it to capture quantities that depend 
sensitively on the small scales more accurately than SFD8 as
well. An example of such a quantity is the dissipation rate 
introduced above which is proportional to the second order moment
of the velocity gradient and thus is dominated by 
high wavenumber activity. 
This is indeed observed in \rfig{disp-adnl}(a) where we show 
the evolution of $\langle\epsilon\rangle$ as the flow decays
along with the analytical value %at $t=0$ 
(dashed-star line)
which can be computed easily 
%from \eqn{diffinit} as
%$\epsilon_0=\alpha 
%%\langle ({\partial u(x,0)}/{\partial x})^2 \rangle
%=\alpha \langle \left( 
%\sum_kk A(k)  cos(kx+\phi_k)\right)^2 \rangle$.
by taking derivative of \eqn{adnl-sol}, 
$
\langle \epsilon \rangle = 
\alpha \langle ({du}/{dx})^2 \rangle=
\alpha \langle 
%\frac{
%\left( 
\big(
	\int_{-\infty }^{\infty } \II dy
	\int_{-\infty }^{\infty }
	\frac{\left(2 \alpha t-(x-y)^2\right)
	\II }{2 \alpha t^2} dy \allowbreak 
	-\int_{-\infty }^{\infty } \frac{ (x-y) \II }{t} dy
 \int_{-\infty }^{\infty } \frac{(y-x) \II }{2 \alpha t} dy
\big) %\right) 
% }{
 \left(\int_{-\infty }^{\infty } \II dy\right)^{-2}
 %}
 \rangle.$%$
%\begin{equation}
%\langle \epsilon \rangle = 
%\alpha \bigg\langle \left(\frac{du}{dx}\right)^2 \bigg\rangle=
%\alpha \Bigg\langle \left(\frac{\int_{-\infty }^{\infty } \II  dy
%\int_{-\infty }^{\infty } \left(2 \alpha t-(x-y)^2\right) \II dy %}
%-\int_{-\infty }^{\infty } (x-y) \II  dy 
% \int_{-\infty }^{\infty } (y-x) \II dy}
%{2 \alpha t^2\left(\int_{-\infty }^{\infty }\II dy\right)^2} \right)^2\Bigg\rangle.
%\eqnlabel{adnl-dudx}
%\end{equation} 
We can clearly see that initially OFD2 is very close to 
the analytical value followed by SFD8 and SFD2.
As time evolves, diffusive effects damp high wavenumbers 
faster then small scales and the main contribution to dissipation 
moves to lower wavenumbers where the three schemes present similar 
resolution capabilities. 
The same conclusion holds for the more challenging 
higher order moments. This is seen in 
\rfig{disp-adnl}(b).
where we show the fourth-order moment of velocity gradients. 
Again OFD2 is more accurate than second and eight order 
standard schemes.

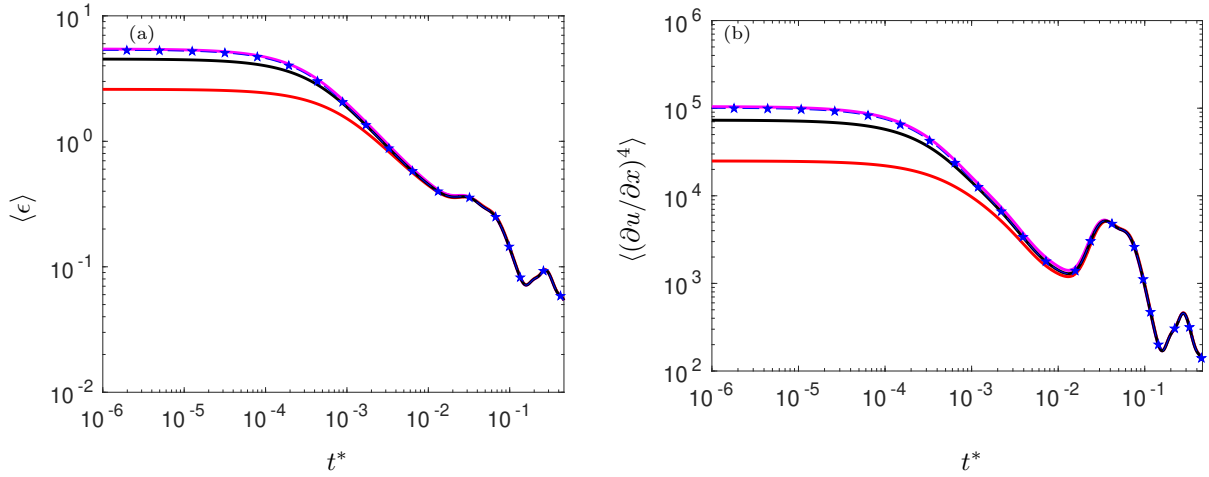
\begin{figure}%[ht]
 \centering
%\subfigure{\includegraphics[width=0.5\textwidth]{new-plots/burgers/non-linear/redo/ener.eps}}
\subfigure{\includegraphics[width=0.48\textwidth]{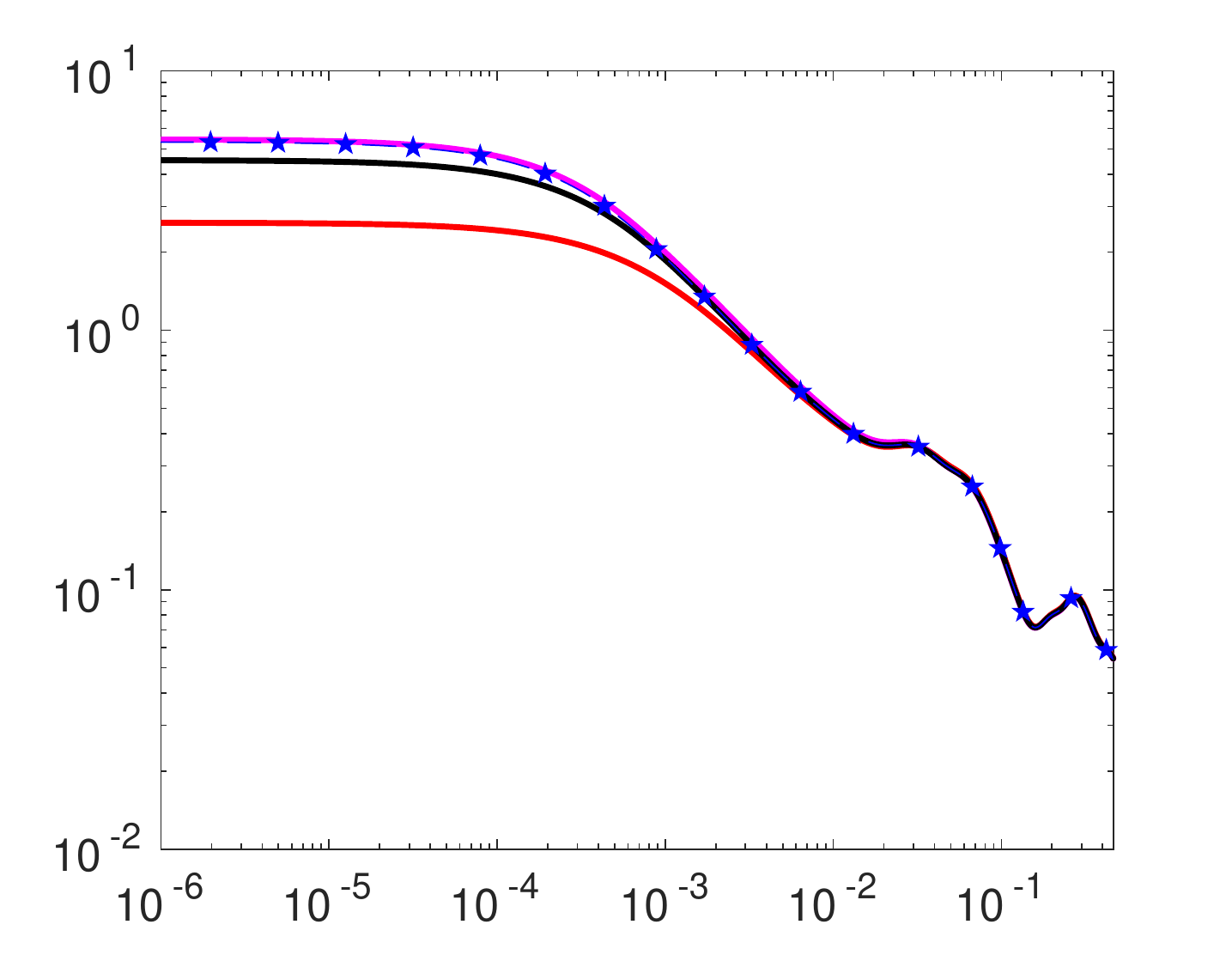}}
\subfigure{\includegraphics[width=0.51\textwidth]{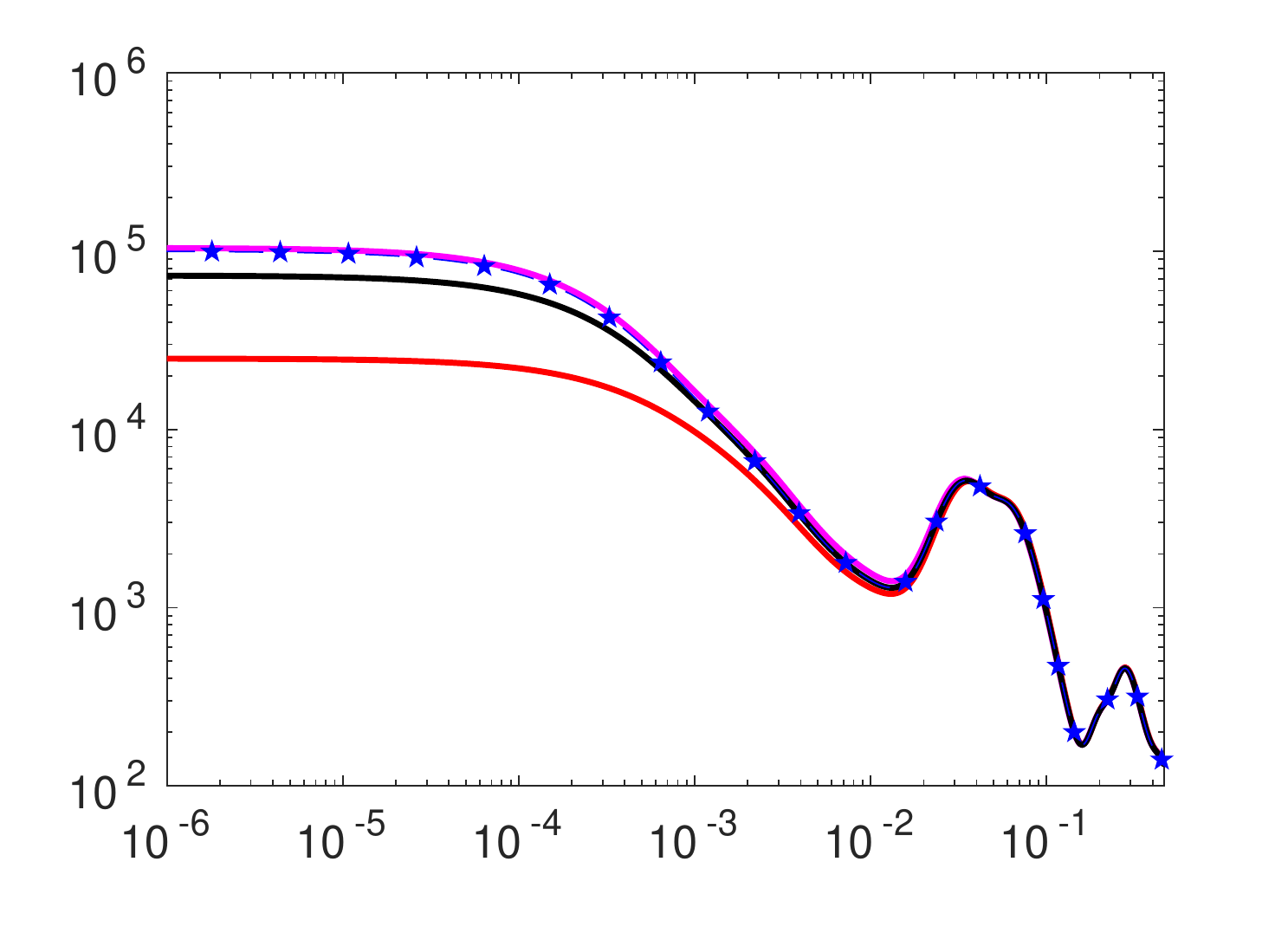}}
  \begin{picture}(0,0)
	\put(-195,165){\scriptsize{(a)}}
	\put(30,165){\scriptsize{(b)}}
        \put(-120,2){$t^*$}
        \put(120,2){$t^*$}
        \put(-240,95){\rotatebox{90} {$\langle\epsilon\rangle$}}
        \put(-10,80){\rotatebox{90} {$\langle(\partial u/\partial x)^4\rangle$}}
\end{picture}
\caption{\figlabel{disp-adnl}(a) Evolution of space-averaged
dissipation (b) Evolution of space-averaged
fourth order moment of velocity gradient with normalized time $(t^*=t/t_0)$ .
Different lines correspond to different numerical schemes: 
SFD8 with RK4 in time (black), 
OFD2 with forward first-order in time (magenta), and
SFD2 with forward first-order in time (red).
The blue dashed-star line is the  
analytical solution.}
\end{figure}

\colk{
\subsection{Wave Equation}
Another equation that has been used to assess the 
performance of numerical schemes is the second order 
linear PDE for the discription of waves \cite{HARAS1994}, commonly 
used in acoustics and given by,
\begin{equation}
\frac{\partial^2 u}{\partial t^2}=c^2
\frac{\partial^2 u}{\partial x^2}
\end{equation}
where $c$ is the propagation speed. The general 
solution to this equation is a standing wave formed by 
the superposition of two travelling
waves moving with velocity $c$ to the right and 
-$c$ to the left respectively. 
Following \cite{HARAS1994}, this equation can be expressed 
as the following system of equations,
\begin{equation}
\begin{pmatrix} \partial u/\partial t \\ 
\partial v/\partial t
\end{pmatrix}=
\begin{pmatrix} 0&1 \\ 
c^2\partial^2/\partial x^2&0
\end{pmatrix}
\begin{pmatrix} u \\ 
v
\end{pmatrix}.
\end{equation}
%in a domain of size $2\pi$ \colr{[DD: is this right?]}.
We solved this system 
for a Gaussian initial condition, $u_0(x)=0.2e^{-64x^2}$, 
that has a narrow width in the 
physical domain and therefore the spectrum 
spans a wide range of wavenumbers in Fourier space.
We use the optimized second-order OFD2 for the second 
derivative in space computed with a weight 
function, $\gamma(w)=e^{-w^2/256}$ for $w=[0,2]$ 
and $\gamma(w)=0$ elsewhere in \eqn{error_SC}. 
This was done to emphasize wavenumbers relevant to the problem. 
Because of the use of a
$CFL$ condition given by $r_c=c\dt/\dx=0.1$, we 
also employ second-order temporal
discretization (RK2).
The solution was advanced until a physical time of $t=2$ for 
periodic boundary conditions and the results have been plotted 
in \rfig{wave}. 
We can see that the SFD2 scheme has already developed significant 
oscillations that trigger instabilities, 
whereas the OFD2 is comparable to both SFD8 and the 
exact solution. The space averaged 2-norm of the
 error for OFD2 is an order 
of magnitude smaller than the error for SFD2. The 
2-norm of the error for OFD2 
scheme is comparable to the error obtained for SFD8. Besides this, 
the OFD2 scheme is three times computationally less 
expensive as compared to the SFD8.
\begin{figure}%[ht]
 \centering
 \includegraphics[width=0.51\textwidth]{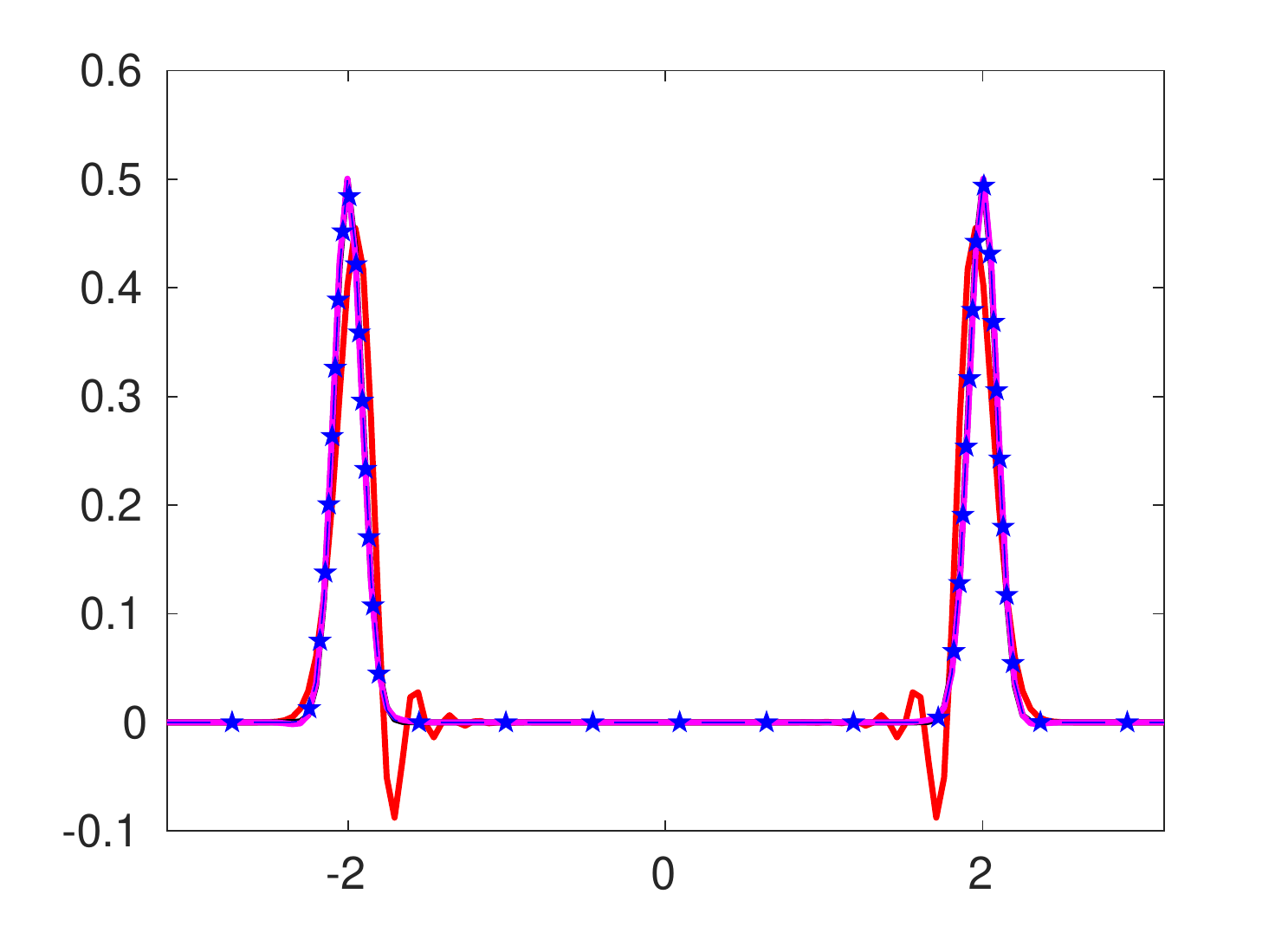}
  \begin{picture}(0,0)
        \put(-120,-8){$x$}
        \put(-245,80){\rotatebox{90} {$u(x,t)$}}
\end{picture}
\caption{\figlabel{wave} \colk{Solution of the wave-equation for $t=2$
with a Gaussian initial condition.
Different lines correspond to different numerical schemes: 
SFD8 with RK4 in time (black), 
OFD2 with RK2 in time (magenta), and
SFD2 with RK2 in time (red).
The blue dashed-star line is the  
analytical solution.}}
\end{figure}
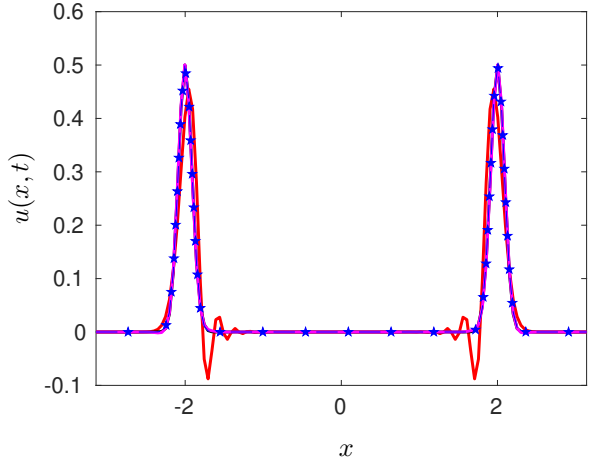
}

%{\color{blue}
\subsection{Space-time errors}
%In order to solve a particular PDE, 
%spatial schemes are typically coupled with a temporal 
%discretization that preserves the formal
%order of accuracy of the fully discrete system.
%This is accomplished by selecting a temporal 
%discretization with a truncation error in time that, through 
%the appropriate CFL condition, matches the spatial 
%order of accuracy in space.
%
%for the time evolution of systems as shown in the 
%stability analysis of the fully discrete system in the previous 
%section. 
%In order to conserve the numerical properties of the spatial 
%schemes, the error introduced by temporal 
%scheme should be contained. The primary criteria is to 
%conserve the global order of the truncation error by 
%taking a temporal scheme with an order compatible with the 
%spatial scheme. 
%
%However, as noted above, 
%the spectral behavior as well as other physically meaningful properties, 
%such as the dispersion relation also need to be carefully assessed. 
In addition to spectral accuracy, other performance metrics 
have been utilized in the literature to assess the appropriateness of 
a numerical scheme to reproduce physics of interest. For 
strongly convective problems, such as in acoustics, 
the dispersion relation provides important information 
about propagation speeds and characteristics.
Thus, the interest in so-called dispersion-relation preserving schemes \cite{TW1992}.

Consider again the convection-diffusion equation:
\begin{equation}
\frac{\partial u}{\partial t}+c\frac{\partial u}{\partial x}=
\alpha\frac{\partial^2 u}{\partial x^2},~~c,\alpha>0 \ ,
\eqnlabel{conv}
\end{equation} 
%in particular when convection dominates.
This linear equation propagates the initial condition 
to the right at the speed $c$ and 
because of the diffusive term the amplitude decreases with time
at a rate determined by the diffusivity coefficient $\alpha$. 
Using $u=\hat{u}e^{-j\omega t}e^{jkx}$ 
it is easy to obtain the exact dispersion relation for this
equation, $\omega=ck-j\alpha k^2$, where $k$ is the wavenumber. 
Two physically meaningful quantities related to the dispersion 
relation are the phase and group velocities ($c_p$ and $c_g$ respectively) 
which are given by the real part of 
$c_p=\omega/k$ and $c_g=d\omega/d k$.
Note that both are equal to $c$ for the analytical solution.
The equivalent numerical dispersion relation is, on
the other hand, given by
$\omega^*=c^*k-j\alpha^*k$, 
where $c^*$ and $\alpha^*$ are the numerical velocity and diffusivity
respectively.  
One is thus interested in how the phase and group velocities from the
computed solution ($c_p^*$ and $c_g^*$) compares to their exact values.

A convenient way to obtain dispersion relations and derived quantities
is
through the so-called amplification factor $G^*$ defined as the ratio
of Fourier modes of the numerical solution at two consecutive
time steps. In general, one can write 
%. Using a Fourier 
%mode for $u$ in \eqn{dis_conv} we can write, 
$G^*=|G^*|e^{-j\beta}$,
where $|G^*|$ is the magnitude of the amplification factor 
 and the phase $\beta$ is related to the numerical phase speed
 $(c_p^*)$.
Note that the exact amplification factor is given by $G=e^{-i\omega \dt}$. 
It is then readily shown that the numerical group 
velocity can be written as \cite{TKS2006},
\begin{equation}
\frac{c_{g}^*}{c_g}=\frac{1}{r_c\dx}\frac{d\beta}{dk}.
\end{equation}
where, as noted above, $c_g=c$.
%For a scheme used in \eqn{dis_conv} we can 
%write the general numerical group velocity as,
%\begin{equation}
%\frac{V_{g}^*}{c}=\frac{(k\dx)^2 \sum_{m=1}^M c_m  cos(mk\dx)}
%{(k\dx)^2 +(w\dt)^2 \sum_{m=0}^{2M} b_m  cos(mk\dx)},
%\end{equation} 
%%where $b_m,c_m$ depend upon the spatial 
%coefficients $a_m$ and 
%where, $r_c=c\dt/\dx=w\dt/k\dx$, is the convective $CFL$. 
In a dispersion-relation preserving scheme, this ratio 
as well as $c_p^*/c_p$
should be close to unity. 
 
As an illustration, consider \eqn{conv} discretized with a forward
difference in time and the 
spatial schemes in \rtab{second_optim_ad} and 
\rtab{first_optim_ad}:
\begin{equation}
\frac{u_i^{n+1}-u^n_i}{\dt}=\frac{-c}{\dx}\sum_{m=-M}^{M}a_{m,1} u_{i+m}
+\frac{\alpha}{\dx^2}\sum_{m=-M}^{M}a_{m,2} u_{i+m}.
\eqnlabel{dis_conv}
\end{equation}
Since $c_g^*/c_g$ and $c_p^*/c_p$ depend on both $k\dx$ and 
$\omega \dt$, one can assess performance by 
measuring the area in the $k\dx-\omega \dt$ 
plane in which these ratios area within
some percentage of unity, say $5\%$ \cite{TKS2003}. 
The larger this area, the larger the range of wavenumbers
that preserve the dispersion relation. 
And since the phase errors are more prominent in convection dominated 
problems, we will consider relatively small values of $r_d$.
From table \ref{table:vg_area}, we can see that when the optimized 
schemes derived 
above are coupled with forward difference in time, the region in 
the $k\dx-\omega \dt$ for which $c_{g}^*/c_g$ is close to 
unity, increases considerably for different values of $r_d$. 
For $M=4$, this area is five times larger than 
that for the standard second order scheme.
We can also see that the 
ratio $c_p^*/c_p$ is close to unity for a much larger area for
optimized schemes.
This is in fact not unexepected 
as the objective function in the minimization problem is designed 
to bring the modified wavenumber close to the actual wavenumber,
resulting in more accurate derivatives leading, 
in general, to lower errors in phase and group velocities. 
%In table \ref{table:vg_area}
%we observe similar behavior for the two values of $r_d$.
Similar results were also observed for the diffusivity ratio 
$\alpha^*/\alpha$. % and there is no anti-diffusion \cite{TKS2006} 
 %in these schemes within the stable region. 

% The 
%improved numerical performance of the optimzed schemes in the spectral
%space is
% also highlighted in the numerical section in terms of energy 
%spectrum and numerical velocity.
%Part of our ongoing research is the optimization of 
%the spatial and temporal schemes simultaneously in the $k\dx-\omega \dt$ plane. 
%This results in added constraints and modified objective function depending upon 
%the physics of the problem, for example,
% difference between the numerical 
%and exact group velocity. This will not be discussed in this paper. 

\begin{table}[h]
\begin{center}
 \begin{tabular}{|c|c|c|c|c|}
        \hline
& \multicolumn{2}{c|}{$r_d=0$}& \multicolumn{2}{c|}{$r_d=0.02$}\\
\hline
 $M$ & $c_g^*/c_g$ &$c_p^*/c_p$& $c_g^*/c_g$ &$c_p^*/c_p$\\ \hline
 1(S) & 0.58 &1.73 & 0.62  & 1.84 \\
 2(O) & 2.29 &6.87& 2.27&6.89\\
 3(O) & 1.82 &5.13& 1.87&5.34\\
 4(O) & 2.91& 10.09&3.04 &10.59\\
% 4(S) & 10.97&24.52& &\\
\hline
\end{tabular}
\end{center}
	  \caption{\colk{Percentage area in the $k\dx-\omega \dt$ plane
	  for which $c_{g}^*/c_g$ and $c_p^*/c_p$ 
	  is within $5\%$ of unity for
	  $r_d=0$ and $r_d=0.02$. In parenthesis, 
	  O stands for optimized scheme, and S for standard
	  scheme. The schemes used here are those in 
	  \rtab{second_optim_ad} and
	  \rtab{first_optim_ad}.
	  }
	}
        \label{table:vg_area}
\end{table}
%}

\colk{
We close this section by noting that although
optimized schemes seem to naturally preserve dispersion relations 
better, constraints on these dispersion relations 
can be directly included in the general framework
proposed here. For example, one can construct 
objective functions as in \cite{TW1992} where one can 
emphasize either propagation characteristics or damping 
characteristics. In the present framework, this can be done
in such a way to ensure, at the same time, stability of the resulting
scheme.
Some of the effects of such 
an implementation were discussed in \rsec{even}. 
}

\section{Conclusions and final remarks}\label{sec:conclusions}

Standard finite difference schemes are commonly derived 
to maximize its formal order of accuracy for a given stencil. 
The spectral accuracy and stability of the schemes so obtained 
are typically checked \textit{a posteriori}. 
While efforts have been devoted to merge order-of-accuracy 
constraints and spectral accuracy, no general mathematical 
framework has been 
put forth which, perhaps more importantly, introduced 
stability as an additional constraint.
Here we develop such a framework to derive finite differences 
that accounts for order of accuracy, 
spectral resolution and stability.
The most general formulation is given by \eqn{explicitSynth}.

While order of accuracy is defined as the power of the lowest order 
term in the truncation error, 
spectral accuracy is defined through 
an objective function that 
minimizes the spectral error in some specific way. 
This definition of this error is 
rather general and includes a weighting function 
$\gamma(\kdx)$ which can be used to emphasize
different scales relevant to the physical problem being solved. 
The fusion of order-of-accuracy  constraints and spectral 
accuracy leads to a minimization problem which is convex 
and thus leads to a global minimum.
The optimal coefficients, which are obtained analytically
from the minimization problem, were shown to be symmetric for the 
even derivatives and anti-symmetric for the odd derivatives. 
This was shown to correspond to vanishing imaginary and 
real part of the error $e(\kdx)$.
This is the case, regardless of the functional form of 
the weighting function $\gamma(\kdx)$. 
In other words, we have shown that the minimization of
spectral errors leads to symmetric or anti-symmetric stencils 
for even and odd derivatives respectively, 
regardless of which range of wavenumber is optimized.  

We have also incorporated stability into the unified framework. 
The semi-discrete system is shown to depend upon the sign of the 
parameters $\beta_d$ for even $d$ in the PDE.
For the fully-discrete system,
stability is assured by requiring the spectral radius 
of the evolution matrix to be smaller than unity which can 
in turn be bounded by its 2-norm. The latter can be written as a 
linear matrix inequality.
This provides an additional 
constraint imposed on the minimization 
problem to ensure that the resulting scheme is stable.
The final unified mathematical framework consists on the
minimization of an objective function representative of 
spectral error constrained by given order of accuracy and
stability guarantees.
Due to the non-linearity inherent in this constraint, this problem is no 
longer convex and therefore cannot be solved directly using standard 
optimization tools. Two approaches were proposed to tackle this
non-linear optimization. 

In the first approach, given an optimal scheme (subjected to 
constraints in order and spectral accuracy) one finds the 
largest $\dt$ for which the scheme remains stable.
This approach is similar to standard practices in which 
given a scheme, one aims at finding the largest time step
that assures stability. Here, however, the largest time step 
results from a convex
optimization which gives global extrema and can be solved efficiently.
Through an example, we showed that the stability region 
decreased with increasing stencil size $M$ consistent 
with a reduction of (artificial) dissipation at higher wavenumbers.
This approach 
illustrates how accuracy and stability are separate requirements
in the formulation: 
the additional degrees of freedom available for longer stencils are
used to maximize spectral resolution regardless of stability. The result
is a reduced stability region in the $r_d$-$r_c$ space.
In the second approach, the three elements are combined: 
optimal coefficients are obtained with both order of 
accuracy and stability constraints for a given time-step $\dt$.
This is a common situation when the time step is set by physical considerations
(e.g.\ shortest time scale in the problem).
In this approach, the additional information provided by neighboring 
grid points is used to extend stability which 
make the use of much larger time-step feasible
with spectral error comparable to the standard schemes. 
We showed that explicit schemes both in time and space, can 
remain stable for very large time steps. This can provide 
significant advantages for massively parallel simulations for which 
implicit schemes become increasingly challenging at large processor counts. 

Several numerical results were presented to illustrate the 
numerical performance against standard finite differences 
of different orders. In particular, we compared optimized schemes 
against standard schemes of the same order and the same stencil size.
While the latter presents the same computational cost in terms of
spatial derivatives, the fully discrete system may be less expensive 
if the order of the temporal discretization is to be compatible with spatial order.
Introducing the effect of temporal discretization in the unified framework 
presented here is part of our own ongoing research. 
Another 
%important point of consideration for numerical simulation is the
%correct 
%implementation of the boundary conditions. 
application of the framework presented here is the inclusion
of more general boundary conditions. 
Although we have limited 
our results to periodic domains,
the framework can straightforwardly 
be extended to solve problems with non-periodic boundary conditions. 
The overall minimization problem and the constraints remain the same 
but the structure of 
%underlying coefficient matrix $\Adp$ and $\vo{X}_d$ 
some matrices 
has to be changed to restrict the stencil used close to boundaries.
%so that only grid points
%inside the domain are used. This naturally leads to 
%biased schemes close to the boundary.
%Note that this can be done in a very general way as follows.
In particular, 
the first few and last elements of $\vo{A}_d$ in \eqn{Ad} will 
be biased so that only grid points to the right and left of the 
boundary are used respectively. The size of each stencil 
as well as the order or accuracy can be set individually through
the corresponding entries in $\vo{X}_d$ in \eqn{domainOrder}. 
As commonly done in simulations of complex flows, a
progressive reduction of order of accuracy 
close to the boundary can thus be easily 
accomodated in this formulation. The stability constraint 
is identical to that presented for periodic domains.
Thus, here too 
\eqn{explicitSynth} will yield spectrally optimal schemes
for a given $\dt$ and \eqn{ana_sol_approach1} 
will yield the maximum $\dt$ 
for which a spectrally optimal scheme remain stable. 

\colk{We finally make some remarks about further potential 
generalizations of the framework.
The stability constraint used here ensure 
that solutions are non-growing and are applied to the entire PDE
which may involve terms of 
different characteristics (convection, diffusion, etc).
While this type of constraint is standard practice, 
it seems possible to extend the framework
to other approaches that can capture other details of the 
error dynamics \cite{SR2011}. This is part of future work.
Furthermore, our focus was on spatial discretization with a
given temporal scheme which could be of arbitrary order. 
Because of the potential additional effects when 
coupling space-time operators, it is a natural next step 
to optimize spatial and temporal schemes simultaneously. 
Unfortunately, this coupled optimization 
problem yields non-linear objective functions \cite{HARAS1994} 
which makes the mathematical problem much more challenging.
%but is more equipped to 
%define the physics of the problem. 
%Additional constraints such as 
%in \cite{TKS2010} can also be added for problem specific use of the framework.
 This is also part of our ongoing research 
and will be discussed in detail elsewhere.
}

%\colk{
%Even though the optimization will gives 
%schemes with minimum spectral error at each grid point,
%the effect of boundary schemes in non-periodic domains on
%the overall spectral behavior and stability needs to carefully examined.
%The control parameter $\sigma$ (\rsec{even}) can take different 
%values close to the boundary to put more emphasis on the dissipation or 
%dispersion errors for the biased schemes.
%In the current paper our focus was 
%only on the spatial schemes and the temporal scheme for the 
%was fixed to be forward difference. 
%Since the coupling of space-time operators alters the numerical 
%properties of the spatial scheme, it is important to optimize 
%spatial and temporal schemes simultaneously. This coupled optimization 
%problem yields non-linear objective function \cite{HARAS1994} 
%but is more equipped to 
%define the physics of the problem. Additional constraints such as 
%in \cite{TKS2010} can also be added for problem specific use of the framework.
% This is a part of our ongoing research 
%and will be discussed in detail in the next paper.
%}
%

In summary, we integrated order of accuracy, spectral resolution 
and stability in the derivation of 
finite differences in a unified framework. 
We have shown 
specific properties of the resulting schemes in terms of the kind 
of error expected.
The coupling of these three critical elements in a unified 
formulation allows one to decouple 
requirements in terms of e.g.\ order of accuracy and spectral
accuracy.
This coupling also manifests itself in the trading of accuracy with stability.
We showed, for example, how spectrally optimal finite differences 
bias odd order derivatives to maintain stability at the expense of 
accuracy. Other tradeoffs can be analyzed similarly 
within the framework presented here.

\section{Acknowledgments}

The authors gratefully acknowledge NSF
(Grant OCI-1054966 and CCF-1439145)
for financial support. The authors also
thank XSEDE for computer time on their systems. 

\appendix
\section{Imaginary component of the optimal spectral error $e(\kdx)$ is zero
for even derivatives}
\label{sec:app_imag_zero}
\begin{proof}
The spectral error for even derivative is
 \begin{align*}
 	 e(\kdx) &= \left(\vo{C}^T(\kdx)\ad-(-1)^q\kdx^{d}\right) + j\vo{S}(\kdx)^T\ad.
 \end{align*}  
 Therefore, 
 \begin{align*}
 \|e(\kdx)\|_{\mathcal{L}_2}^2 &= \int_0^\pi \gamma(\kdx)\left[ \left(\vo{C}^T(\kdx)\vo{a}_d-(-1)^q\kdx^{d}\right)^2 + (\vo{S}(\kdx)^T\vo{a}_d)^2\right]d\kdx,
 \end{align*}

or 

$$
 \min_{\ad} \|e(\kdx)\|_{\mathcal{L}_2}^2 = \min_{\ad}  \int_0^\pi \gamma(\kdx)\left[\left(\vo{C}^T(\kdx)\ad-(-1)^q\kdx^{d}\right)^2\right]d\kdx  + \min_{\ad} \int_0^\pi\gamma(\kdx) \left(\vo{S}(\kdx)^T\ad\right)^2d\kdx.
$$
Since the cost function is sum of squares, it is minimized if and only if individual terms are minimized.

Without loss of generality, we can write $\ad:=\ad^s + \ad^{as}$, where $\ad^s$ is symmetrical about the central element, and $\ad^{as}$ is anti-symmetric about central element. Therefore,
\begin{align*}
\vo{S}^T(\kdx)\ad &= \vo{S}^T(\kdx)(\ad^s + \ad^{as}) = \vo{S}^T(\kdx)\ad^{as},\\
\vo{C}^T(\kdx)\ad &= \vo{C}^T(\kdx)(\ad^s + \ad^{as}) = \vo{C}^T(\kdx)\ad^{s},
\end{align*}
since $\vo{S}^T(\kdx)\ad^s = 0$ for symmetric coefficients and $\vo{C}^T(\kdx)\ad^{as} = 0$ for anti-symmetric coefficients.

Consequently,
\begin{align*}
&\min_{\ad}  \int_0^\pi \gamma(\kdx)\left[\left(\vo{C}^T(\kdx)\ad-(-1)^q\kdx^{d}\right)^2\right]d\kdx  + \min_{\ad} \int_0^\pi \gamma(\kdx)\left(\vo{S}(\kdx)^T\ad\right)^2d\kdx, \\
= & \min_{\ad^s}  \int_0^\pi \gamma(\kdx)\left[\left(\vo{C}^T(\kdx)\ad^s-(-1)^q\kdx^{d}\right)^2\right]d\kdx  + \min_{\ad^{as}} \int_0^\pi \gamma(\kdx)\left(\vo{S}(\kdx)^T\ad^{as}\right)^2d\kdx.
\end{align*}
Therefore, the two optimizations are independent of each other.
For 
a positive real valued function $\gamma(\kdx)$, the second term is zero if
and only if $\ad^{as} = 0$. Consequently, for $\mathcal{L}_2$ optimal
spectral errors, the imaginary part of the spectral error is zero for even
derivatives, and the optimal $\ad$ is symmetrical about central element.

We next analyze the feasibility of the order constraint with $\ad^{as} = 0$. Let $\vo{T}_{s}:=(\vo{I} + \vo{J})/2$, and $\vo{T}_{as}:=(\vo{I} - \vo{J})/2$, where $\vo{J}$ represents anti-diagonal matrix. Wth these transformation matrices, we can write $\ad^{s} := \vo{T}_{s}\ad$ and $\ad^{as} := \vo{T}_{as}\ad$. Therefore, 
$$ \vo{a}_d^T\vo{X}_d = (\ad^s + \ad^{as})^T\vo{X}_d = (\ad^s)^T\vo{X}_d = \ad^T\vo{T}^T_{s}\vo{X},$$ and the order accuracy constraint can be written as 
$$
\ad^T\vo{T}^T_{s}\vo{X} = \vo{y_d}.
$$

We observe that the structure of $\vo{X}_d$ is such that the odd columns 
are symmetric and the even columns are anti-symmetric, about the 
central element. Therefore, for the even columns, $\vo{T}^T_{s}(\vo{X}_d)_i=0$. Noting that the even columns of $\vo{y_d}$ are zero, we can conclude that the constraints corresponding to the even columns are trivially satisfied for symmetric coefficients. For the odd columns of $\vo{X}_d$, we observe that $\vo{T}^T_{s}(\vo{X}_d)_i=(\vo{X}_d)_i$. That is, the constraints corresponding to the odd columns are unaffected. Consequently, if $\vo{a}_d^T\vo{X}_d=\vo{y}_d$ is feasible, then $\ad^T\vo{T}^T_{s}\vo{X} = \vo{y_d}$ is also feasible.

%Using the above argument, consider odd columns of $\vo{X}_d$, 
%$$
%(\vo{a}_d^s)^T(\vo{X}_d)_i=(\vo{T}_s\vo{a}_d)^T(\vo{X}_d)_i=
%\vo{a}_d^T\vo{T}_s^T(\vo{X}_d)_i=
%\vo{a}_d^T\vo{T}_s(\vo{X}_d)_i=
%\vo{a}_d^T(\vo{X}_d^s)_i=
%\vo{a}_d^T(\vo{X}_d)_i.
%$$
%
%This implies that the constraint remains unchanged 
%for the odd columns. Also, the only non-zero 
%element on the RHS \textit{i.e} $\vo{y_d}_i\ne0$ is for 
%$i=d+1$, or $i$ being odd. \\
%Now for the even columns of $\vo{X}_d$,
%$$
%(\vo{a}_d^s)^T(\vo{X}_d)_i=(\vo{T}_s\vo{a}_d)^T(\vo{X}_d)_i=
%\vo{a}_d^T\vo{T}_s^T(\vo{X}_d)_i=
%\vo{a}_d^T\vo{T}_s(\vo{X}_d)_i=
%\vo{a}_d^T(\vo{X}_d^s)_i=
%0.
%$$
%But on the RHS we have  
%$(\vo{y}_d)_i=0$ for all even $i$. Therefore, the constraint 
%is trivially satisfied for the even columns of $\vo{X}_d$.  
%
%Thus we can conclude that, 
%
%$$ %\begin{align}
%(\vo{a}_d^s)^T\vo{X}_d=\vo{y}_d.
%$$ %\end{align}
%is equivalent to
%$$ %\begin{align}
%\vo{a}_d^T\vo{X}_d=\vo{y}_d.
%$$ %\end{align}
 
\end{proof}

\section{Real component of the optimal spectral error $e(\kdx)$ is zero for odd derivatives}
\label{sec:app_real_zero}
\begin{proof}
The proof follows similarly to above. The spectral error in this case is
\begin{align*}
e(\kdx) &:= \vo{C}^T(\kdx)\vo{a}_d + j\left(\vo{S}^T(\kdx)\vo{a}_d -(-1)^q \kdx^{d}\right).\end{align*}
Using the same decomposition for $\ad$ as above, we get
\begin{align*}
\min_{\ad} \|e(\kdx)\|_{\mathcal{L}_2}^2 &= \min_{\ad} \int_0^\pi\gamma(\kdx) \left(\vo{C}^T(\kdx)\ad\right)^2 d\kdx  + \min_{\ad} \int_0^\pi\gamma(\kdx) \left[\vo{S}^T(\kdx)\ad -(-1)^q \kdx^{d}\right]^2 d\kdx,\\
&\min_{\ad^s} \int_0^\pi\gamma(\kdx) \left(\vo{C}^T(\kdx)\ad^s\right)^2 d\kdx  + \min_{\ad^{as}} \int_0^\pi \gamma(\kdx) \left[\vo{S}^T(\kdx)\ad^{as} -(-1)^q \kdx^{d}\right]^2 d\kdx.
\end{align*}
Using similar arguments as above, the optimal solution will guarantee $\ad^s = 0$ and 
consequently, the real part of the spectral error is zero for odd derivatives. With $\ad^s = 0$, the optimal $\ad$ will be anti-symmetric about its central element.

The proof for feasibility of the accuracy order constraint, with $\ad^s = 0$, is similar to the feasibility proof for the even derivative.
%
%\textbf{Invariance of constraint for} $\vo{a}_d=\vo{a}_d^{as}$\\
%
%Order of accuracy constraint for the minimization is given by
%$$ %\begin{align}
%\vo{a}_d^T\vo{X}_d=\vo{y}_d.
%$$ %\end{align}
%Exploiting the structure of columns of $\vo{X}_d$ about 
%its central element this constraint can be simplified as in the 
%case of even dervatives.
%For the anti-symmetric even columns, $\vo{T}_{as}\vo{X}_d=\vo{X}_d$.
%Therefore for $i=even$, 
%$$
%(\vo{a}_d^{as})^T(\vo{X}_d)_i=%(\vo{T_{as}a_d})^T(\vo{X_d})_i=
%%\vo{a_d}^T\vo{T_{as}^T}(\vo{X_d})_i=
%\vo{a}_d^T\vo{T}_{as}(\vo{X}_d)_i=
%\vo{a}_d^T(\vo{X}_d^{as})_i=
%\vo{a}_d^T(\vo{X}_d)_i, 
%$$
%the constraint remains unchanged. Also, the only non-zero
%element on the RHS \textit{i.e} $(\vo{y}_d)_i\ne0$ is for
%$i=d+1$, or $i$ being $even$. \\
%Since all the odd columns are symmetric about the central element, 
%$\vo{T}_{as}\vo{X}_d=0$.  
%So for $i=odd$
%$$
%(\vo{a}_d^{as})^T(\vo{X}_d)_i=%(\vo{T_{as}a_d})^T(\vo{X_d})_i=
%%\vo{a_d}^T\vo{T_{as}^T}(\vo{X_d})_i=
%\vo{a}_d^T\vo{T}_{as}(\vo{X}_d)_i=
%\vo{a}_d^T(\vo{X}_d^{as})_i=
%0.
%$$
%But on the RHS we have
%$(\vo{y}_d)_i=0$ for all odd $i$. Therefore, the constraint
%is trivially satisfied.
%Thus,
%$$ %\begin{align}
%(\vo{a}_d^{as})^T\vo{X}_d=\vo{y}_d.
%$$ %\end{align}
%is equivalent to
%$$ %\begin{align}
%\vo{a}_d^T\vo{X}_d=\vo{y}_d.
%$$ %\end{align}
%
%

\end{proof}

\section{$\Ad$ Invariance of grid point location}
\label{sec:app_invariance}
\begin{proof}
We know that the inverse of a partitioned matrix can be written as
\begin{align}
\begin{bmatrix}\mathbf {A} &\mathbf {B} \\\mathbf {C} & \mathbf {D} \end{bmatrix}^{-1}={\begin{bmatrix}\mathbf {A} ^{-1}+\mathbf {A} ^{-1}\mathbf {B} (\mathbf {D} -\mathbf {CA} ^{-1}\mathbf {B} )^{-1}\mathbf {CA} ^{-1}&-\mathbf {A} ^{-1}\mathbf {B} (\mathbf {D} -\mathbf {CA} ^{-1}\mathbf {B} )^{-1}\\-(\mathbf {D} -\mathbf {CA} ^{-1}\mathbf {B} )^{-1}\mathbf {CA} ^{-1}&(\mathbf {D} -\mathbf {CA} ^{-1}\mathbf {B} )^{-1}\end{bmatrix}}.
\eqnlabel{partitionInverse}
\end{align}

Using this result, let 
\begin{align}
\begin{bmatrix} \vo{Q}_d & \vo{X}_d \\
 \vo{X}_d^T & \vo{0}\end{bmatrix}^{-1} := \begin{bmatrix}\vo{M}_1 & \vo{M}_2\\ \vo{M}_3 & \vo{M}_4 \end{bmatrix},
\end{align}
for suitably defined $\vo{M}_i$ using \eqn{partitionInverse}. Using the kronecker product result
\begin{align}
(\vo{A}\otimes\vo{B})(\vo{C}\otimes\vo{D}) = \vo{AC}\otimes\vo{BD},
\end{align}
and \eqn{partitionInverse}, we can write
\begin{align}
\begin{bmatrix} 
\vo{I}_{N}\otimes \vo{Q}_d & \vo{I}_{N}\otimes \vo{X}_d\\
\vo{I}_{N}\otimes \vo{X}_d^T\ & \vo{I}_{N}\otimes\vo{0}
\end{bmatrix}^{-1} = \begin{bmatrix}\vo{I}_{N}\otimes \vo{M}_1 & \vo{I}_{N}\otimes \vo{M}_2\\ \vo{I}_{N}\otimes \vo{M}_3 & \vo{I}_{N}\otimes \vo{M}_4 \end{bmatrix}.
\end{align}
Therefore, the optimal solution 
\begin{align*}
\begin{pmatrix}
\vo{v}_d\\\boldsymbol{\Lambda}_d 
\end{pmatrix}^\ast &= \begin{pmatrix}\left(\vo{I}_{N}\otimes \vo{M}_1\right)\left(\vo{1}_{N\times 1}\otimes \vo{r}_d\right) + \left(\vo{I}_{N}\otimes \vo{M}_2\right)\left(\vo{1}_{N\times 1}\otimes \vo{y}_d^T\right)\\
\left(\vo{I}_{N}\otimes \vo{M}_3\right)\left(\vo{1}_{N\times 1}\otimes \vo{r}_d\right) + \left(\vo{I}_{N}\otimes \vo{M}_4\right)\left(\vo{1}_{N\times 1}\otimes \vo{y}_d^T\right)\end{pmatrix},\\
& = \begin{pmatrix} \vo{1}_{N\times 1}\otimes \left(\vo{M}_1\vo{r}_d + \vo{M}_2\vo{y}_d^T\right)\\[2mm]
\vo{1}_{N\times 1}\otimes \left(\vo{M}_3\vo{r}_d + \vo{M}_4\vo{y}_d^T\right)\end{pmatrix},\\
& = \begin{pmatrix}
\vo{1}_{N\times 1}\otimes \vo{a}_d^\ast\\
\vo{1}_{N\times 1}\otimes \boldsymbol{\lambda}_d^\ast
\end{pmatrix},
\end{align*}
where $(\vo{a}_d^\ast,\boldsymbol{\lambda}_d^\ast)$ is the solution of \eqn{optimal:Explicit}. Therefore, the optimal solution is identical for all grid points.
\end{proof}

\section{Construction of coefficient matrix $\Adp$}
\label{sec:app_adp}
\begin{proof}
Define a shift operator $\shift{k}$, which is an $N\times N$ matrix,
with elements

\begin{align}
        \Phi_{k_{ij}} := \delta ((i-j-k) \mod N),
\end{align}
where $\delta(\cdot)$ is the Kronecker  delta function defined as
$$
\delta(i) = \left\{\begin{array}{c} 0 \text{ if } i\neq0,\\ 1 \text{ if } i = 0.\end{array}\right.
$$
For a column vectors, the operator  $\shift{k}$ cyclically shifts the elements down, $k$ times. For example, for
$$
\vo{v}:=\begin{pmatrix} 1\\2\\3\\4 \end{pmatrix}, \shift{1}\vo{v} = \begin{pmatrix} 4\\1\\2\\3 \end{pmatrix},
$$
where
$$
\shift{1} :=  \delta ((i-j-1) \mod 4) = \begin{bmatrix}
        0 & 0 & 0 & 1 \\
        1 & 0 & 0 & 0 \\
        0 & 1 & 0 & 0 \\
        0 & 0 & 1 & 0
\end{bmatrix}.
$$
For a row vector, the operator cyclically shifts the elements left, $k$ times. That is,
$$
\vo{v}^T:=\begin{pmatrix} 1 &2&3&4 \end{pmatrix}, \vo{v}^T\shift{1} = \begin{pmatrix} 2\\3\\4\\1 \end{pmatrix},
$$

From the definition of vector $\vo{F}$ and $\vo{F}^{(d)}$
\begin{align}
\vo{F} := \begin{pmatrix}
%       f(x_1) \\ \vdots \\ f(x_{N})
        f_1 \\ \vdots \\ f_N
 \end{pmatrix},
 \text{ and }
 \vo{F}^{(d)} := \begin{pmatrix}
%       \left.\frac{\partial^d f}{\partial x^d}\right|_{x_1} \\ \vdots \\
%       \left.\frac{\partial^d f}{\partial x^d}\right|_{x_{N}}
        f_1^{(d)} \\ \vdots \\
        f_N^{(d)}
 \end{pmatrix}.
\end{align}
we can write the finite difference approximation at the $i$-th grid point as
$$
f_i^{(d)} = \frac{1}{(\dx)^d}\vo{a}^T_{i,d}\vo{T}\shift{\ns_\text{max}+i-1}\vo{F},
$$
for $i=\{1,\cdots,N\}$, and $\vo{T} \in \real^{S\times N}$ is a transformation matrix defined by
\begin{align}
        \vo{T} :=\begin{bmatrix}\vo{0}_{S\times (\ns_\text{max} - \ns)} & \vo{I}_S & \vo{0}_{S\times (\ns_\text{max} - \ns)}\end{bmatrix}.
\end{align}

The matrix $\vo{T}\shift{\ns_\text{max}+i-1}$ is a linear operator, or simply a mask, that picks the correct elements from $\vo{F}$ in determining the derivative at the $i^\text{th}$ location.

Now, let $\Ad^{\vo{\Phi}}$ be the vertical stacking of $\vo{a}^T_{i,d}\vo{T}\shift{\ns_\text{max}+i}$, for $i=\{1,\cdots,N\}$, i.e.
\begin{align}
        \Adp := \begin{bmatrix}
        \vo{a}^T_{1,d}\vo{T}\shift{\ns_\text{max}} \\
        \vdots \\
        \vo{a}^T_{N,d}\vo{T}\shift{\ns_\text{max}+N-1}
        \end{bmatrix} := \sum_i^{N}\bdel_i\bdel_i^T\vo{A}_d\vo{T}\shift{\ns_\text{max}+i-1},
\end{align}
where $\bdel_i\in\real^{n}$ is a vector whose $k^\text{th}$ element is defined by $\delta(i-k)$, i.e. the $i^\text{th}$ element of $\bdel_i\in\real^{n}$ is equal to one and the rest are zero. The vector $\bdel_i$ in \eqn{linearAd} is defined for $n=N$.

The definition of $\Adp$ can be compactly written as
\begin{align}
\Adp = \vo{M}_1\left(\vo{I}_N\otimes \vo{A}_d\right)\vo{M_2},
\eqnlabel{linearAd}
\end{align}
where
\begin{align}
\vo{M}_1 &:= \begin{bmatrix} \bdel_1\bdel_1^T & \cdots & \bdel_N\bdel_N^T\end{bmatrix}, & \vo{M}_2 &:= \begin{bmatrix} \vo{T}\shift{\ns_\text{max}} \\ \vdots \\ \vo{T}\shift{\ns_\text{max}+N-1}\end{bmatrix}.
\end{align}
Equation \eqn{linearAd} shows that $\Ad^{\vo{\Phi}}$ is linear in $\Ad$.

Thus, the finite-difference approximation for the $d^{th}$ derivative
for all the grid points is
\begin{align}
        \vo{F}^{(d)} =  \frac{1}{(\dx)^d} \Adp \vo{F}.
        \eqnlabel{fdDomain_ap}
\end{align}

\end{proof}
% References
%
% Following citation commands can be used in the body text:
% Usage of \cite is as follows:
%   \cite{key}          ==>>  [#]
%   \cite[chap. 2]{key} ==>>  [#, chap. 2]
%   \citet{key}         ==>>  Author [#]

% References with bibTeX database:
\newpage
\bibliographystyle{model1-num-names}
\bibliography{main,raktim}

% Authors are advised to submit their bibtex database files. They are
% requested to list a bibtex style file in the manuscript if they do
% not want to use model1-num-names.bst.

% References without bibTeX database:


% \bibitem must have the following form:
%   \bibitem{key}...
%


% \end{thebibliography}

\end{document}